\documentclass[12pt]{article}
\usepackage{geometry}
\usepackage[maxnames=50,sorting=nyt]{biblatex}
\usepackage{float}
\usepackage[english]{babel}
\usepackage{amssymb, latexsym, amsmath, amsthm}
\usepackage{float}
\usepackage[]{mdframed}

\usepackage{lscape}
\usepackage{fixme}
\fxsetup{status=draft}
\usepackage{appendix}
\usepackage{tikz-cd}
\usetikzlibrary{decorations.markings}
\usepackage{pgfplots}
\usepackage{caption}
\usepackage{float}
\usepackage{subcaption}
\usepackage{enumitem}
\numberwithin{equation}{section}
\usepackage{verbatim}
\usepackage[utf8]{inputenc}
\usepackage{blindtext}
\usepackage{titlesec}
\titleformat{\section}
   {\normalfont\fontsize{12pt}{14pt}\selectfont\bfseries}{\thesection}{1em}{\MakeUppercase}
\usepackage{sectsty}
\sectionfont{\centering}
\usepackage{bbm}
\usepackage{longtable}
\usepackage{fullpage} 
\usepackage{parskip} 
\usepackage{hyperref}
\usepackage{enumitem}
\usepackage{comment}
\usepackage{mathtools}
\usepackage{xcolor}
\hypersetup{backref=true,       
    pagebackref=true,               
    hyperindex=true,                
    colorlinks=true,                
    breaklinks=true,                
    urlcolor= black,                
    linkcolor= blue,                
    bookmarks=true,                 
    bookmarksopen=false,
    filecolor=black,
    citecolor=black,
    linkbordercolor=red
}
\newcommand{\C}{\mathbb{C}}
\newcommand{\R}{\mathbb{R}}
\newcommand{\N}{\mathbb{N}}
\newcommand{\Z}{\mathbb{Z}}

\usepackage{pgfplots}
\usepgfplotslibrary{fillbetween}
\theoremstyle{plain}
\newtheorem{thm}{Theorem}[section]
\newtheorem{prop}[thm]{Proposition}

\newtheorem{cor}[thm]{Corollary}
\theoremstyle{definition}
\newtheorem{defn}[thm]{Definition}

\newtheorem{remark}[thm]{Remark}

\usepackage{sectsty}
\sectionfont{\bfseries\Large\raggedright}
\usepackage{url}
\usepackage{geometry}
\usepackage{ragged2e}
\AtBeginDocument{\hypersetup{pdfborder={0 0 1}}}
\begin{document}
\title{\textsc{Mock Modularity In CHL Models}}
\author{Ajit Bhand$^1$, Ashoke Sen$^2$, Ranveer Kumar Singh$^3$}
\bigskip

\date{\today}
\maketitle
\centerline{ \it ~$^1$Department of Mathematics, Indian Institute of Science Education and Research
Bhopal}
\centerline{ \it Bhopal bypass, Bhopal 462066, India}

\centerline{ \it ~$^2$International Centre for Theoretical Sciences - TIFR}
\centerline{ \it  Bengaluru - 560089, India}
\centerline{ \it ~$^3$NHETC, Department of Physics and Astronomy, Rutgers University}
\centerline{ \it  126 Frelinghuysen Rd., Piscataway, NJ, 08855, USA.}

\bigskip

\centerline{E-mail: abhand@iiserb.ac.in, ashoke.sen@icts.res.in,}\centerline{rks158@scarletmail.rutgers.edu}
\abstract{
\noindent Dabholkar, Murthy and Zagier (DMZ)  proved that there is a canonical decomposition of a meromorphic Jacobi form of integral index for $\mathrm{SL}(2, \mathbb{Z})$ with poles on torsion points into polar and finite parts, and showed that the finite part is a mock Jacobi form. In this paper we  generalize the results of DMZ to meromorphic Jacobi forms of rational index for congruence subgroups of  $\mathrm{SL}(2, \mathbb{Z})$.  As an application, we establish that a large class of  single-centered black hole degeneracies in CHL models are given by the Fourier coefficients of mock Jacobi forms. In this process we refine the result of DMZ regarding the set of charges for which the single-centered black hole degeneracies are given by a mock modular form. In particular, in the case studied by DMZ, we present examples of charges for which the single-centered degeneracies are not captured by the mock modular form of the expected index.   
}\newpage
\tableofcontents
\section{Introduction and Summary}
\subsection{Introduction}
Modular forms are holomorphic functions from the upper half plane $\mathbb{H}$ to the complex plane which satisfy a growth condition and a certain symmetry property called modularity with respect to a finite-index subgroup of $\mathrm{SL}(2,\mathbb{Z})$. 

Mock modular forms first appeared in Ramanujan's famous deathbed letter written to Hardy in 1920 in which he provided a list of 17 functions which have ``modular-like" properties. Ramanujan called these functions mock theta functions. These functions were not completely understood until Zwegers \cite{zwegers2008mock} provided the correct formulation of mock modular forms in his thesis in 2007. Mock modular forms are holomorphic functions which do not satisfy the modularity property, but afford a nonholomorphic completion which is modular. In the last 15 years, mock modular forms have found numerous applications in number theory, combinatorics, moonshine and theoretical physics. For more details, we refer to the excellent survey articles \cite{FOLSOM2017500,Ono2009}.
In DMZ \cite{Dabholkar:2012nd} the authors introduce  \hyperref[mixedmmfdef]{mixed mock modular forms}, which are sums of products of modular forms and mock modular forms of appropriate weights.

For integer $n\geq 1$, the Siegel upper half plane $\mathbb{H}_n$ is defined as the set of all symmetric $n\times n$ complex matrices with positive definite imaginary part. A \hyperref[siegeldef]{Siegel modular form} of degree $n$ is a holomorphic function from $\mathbb{H}_n$ to $\mathbb{C}$ which satisfies the modularity property with respect to $\mathrm{Sp}(n, \mathbb{Z}).$  Siegel modular forms appear naturally in the context of counting quarter-BPS states in Type IIB superstring theory compactified over $K3\times T^2.$

Each Siegel modular form $f$ of degree 2 has a Fourier-Jacobi expansion. Identifying each $\begin{pmatrix}
    \tau & z\\
    z & \sigma \end{pmatrix}=:\Omega \in \mathbb{H}_2$ with the triple $(\tau,z,\sigma)\in\mathbb{H}\times\mathbb{C}\times\mathbb{H}$, we can write 
$$
f(\Omega)= f(\tau,z,\sigma)=\sum_{m=0}^{\infty}\varphi_m(\tau,z)p^{m}, 
$$
where $p=e^{2\pi i \sigma}$. The functions $\varphi_m: \mathbb{H}\times \mathbb{C}\rightarrow \mathbb{C}$ turn out to be what are called \hyperref[jacobidef]{\textit{Jacobi forms}}. 

In DMZ \cite{Dabholkar:2012nd}, the authors have defined the notion of (pure and mixed) \hyperref[mockjacobidef]{\emph{mock} Jacobi forms}. Let $\varphi(\tau,z)$ be a meromorphic Jacobi form of weight $k$ and index $m$ with respect to $\mathrm{SL}(2,\mathbb{Z})\ltimes\mathbb{Z}^2$. Assume that $\varphi$ has simple or double poles only at $z$ belonging to a discrete subset of $\mathbb{Q}\tau+\mathbb{Q}$. It is proved in DMZ \cite{Dabholkar:2012nd} that  $\varphi$ can be decomposed as $\varphi= \varphi^P+ \varphi^F$, where the \textit{polar part} $\varphi^P$ of $\varphi$, which contains all the poles of $\varphi$, is obtained using Appell-Lerch sums. It is then shown that the finite part $\varphi^F$ is a mixed mock Jacobi form.

An interesting application of mock Jacobi forms is in counting single-centered black hole degeneracy in Type IIB superstring theory compactified on $K3\times T^2.$ Single-centered degeneracies in these theories are characterised by electric and magnetic charge vectors. The general prescription to compute the single-centered degeneracy is to integrate the partition function, which is the inverse of the Igusa cusp form, on a specific contour, see \eqref{eq 3.2} for the precise expression. One of the main results of \cite{Dabholkar:2012nd} is to show that the generating function of the single-centered degeneracies with fixed magnetic charge is a mock Jacobi form. \par There are more general such physical models called CHL models obtained as orbifolds of $K3\times T^2$ or $T^6$ compactifications of Type IIB string theory \cite{Chaudhuri:1995fk}. The problem of computing the partition function of these models and the prescription for calculating the single-centered degeneracies was solved in \cite{David:2006ji,David:2006ud, David:2006yn, Jatkar:2005bh, Sen:2007qy,Sen:2009md,Sen:2010ts}. The results are similar to those for Type IIB compactification on $K3\times T^2$ except that the Igusa cusp form is replaced by a more general Siegel modular form for certain subgroups of $\mathrm{Sp}(2,\mathbb{Z}).$ One can ask if a similar mock modular structure emerges in these more general models when considering the generating function of single-centered black hole degeneracies. For some cases, mock modularity of single-centered black hole partition function with fixed small magnetic charge was explicitly demonstrated in \cite{Chattopadhyaya_2019}. Here we provide a more uniform proof of mock modularity for these compactifications and for all magnetic charges. 
\par Below we summarize the results obtained in this paper.
\subsection{Summary of results}
\begin{enumerate} 
\item  In Section \ref{subsec:dmzmockjacform}, we generalize the mathematical results of DMZ \cite{Dabholkar:2012nd} to meromorphic Jacobi forms for congruence subgroups. In particular, 
if $\varphi(\tau,z)$ is a meromorphic Jacobi form of weight $k$ and index $m$ with respect to $\Gamma\ltimes\mathbb{Z}^2$ with $\Gamma\in\{\Gamma_0(N),\Gamma_1(N)\}$ for some $N\geq 1$, with simple or double poles only at $z$ belonging to a discrete subset of $\mathbb{Q}\tau+\mathbb{Q}$, then we show that $\varphi$ can be decomposed as $\varphi= \varphi^P+ \varphi^F$, where $\varphi^P$ is the polar part which contains all the poles of $\varphi$ and  the finite part $\varphi^F$, which is holomorphic, is a mixed mock Jacobi form.
\item In Section \ref{rational_index}, we introduce a generalisation of the Appell-Lerch sums considered in \cite{Dabholkar:2012nd}. For $n \in \mathbb{N}$ and $k, m \in \mathbb{Z}$, let $J_{k,\frac{m}{N}}^N(\Gamma)$ be the space of meromorphic functions from $\mathbb{H}\times \mathbb{C}$ to $\mathbb{C}$ which transform like a Jacobi form of weight $k$ and index $m/N$ for $\Gamma\ltimes (N\mathbb{Z}\times\mathbb{Z})$. Let $\psi\in J_{k,\frac{m}{N}}^N(\Gamma)$ be a function with simple or double poles on a discrete subset $S(\psi)$ of $\mathbb{Q}\tau+  \mathbb{Q}$. For $\ell\in\Z/2m\Z$, let 
$$
h_{\ell}(\tau):= h_{\ell}^{\left(-\frac{\ell N}{2 m}\tau\right)}(\tau):=q^{-\frac{\ell^{2}N}{4m}} \int_{-\frac{\ell N}{2 m}\tau}^{-\frac{\ell N}{2 m}\tau+1} \psi(\tau, z) \zeta^{-\ell} d z,
$$
and 
$$
\psi^{F}(\tau, z):=\sum_{\ell(\bmod~ 2 m)} h_{\ell}(\tau) \widetilde{\vartheta}_{\frac{m}{N}, \ell}(\tau, z),
$$
where $q=e^{2\pi i\tau},~
\zeta = e^{2\pi i z}$ and
$$
\widetilde{\vartheta}_{\frac{m}{N} , \ell } ( \tau , z )=\sum \limits_ {\substack{ r\in \mathbb{Z}\\ r \equiv \ell(\textrm{mod}\ 2m) } } q ^ {N r ^ { 2 } / 4 m } \zeta ^ { r } = \sum \limits_ { n \in \mathbb { Z } } q ^ { N( \ell + 2 m n ) ^ { 2 } / 4 m } \zeta ^ { \ell + 2 m n }.
$$
We show that each $h_{\ell}$ is a mixed mock modular form and the function $\psi^{F}$ is a mixed mock Jacobi form. More precisely, we prove that for each $\ell \in \mathbb{Z} / 2 m \mathbb{Z}$ the completion of $h_{\ell}$ defined by
$$
\widehat{h}_{\ell}(\tau):=h_{\ell}(\tau)-\sum_{s \in S(\psi) / (N\mathbb{Z}\times\mathbb{Z})} \left(D_{s}(\tau) \widetilde{\Theta}_{\frac{m}{N}, \ell}^{s *}(\tau)+E_s(\tau)\widetilde{\Xi}_{\frac{m}{N}, \ell}^{s *}(\tau)\right)
$$
where $\widetilde{\Theta}_{\frac{m}{N}, \ell}^{s *}(\tau),\widetilde{\Xi}_{\frac{m}{N}, \ell}^{s *}(\tau)$ are certain non-holomorphic functions of $\tau$ as defined in \eqref{eq 2.29} and $D_s(\tau),E_s(\tau)$ are certain coefficients of the Laurent expansion of $\psi$ defined in \eqref{eq:resdoupole}, transforms like a modular form of weight $k-\frac{1}{2}$ with respect to $\Gamma(4mN)$  and the completion of $\psi^{F}$ defined by
$$
\widehat{\psi}^{F}(\tau, z):=\sum_{\ell(\bmod 2 m)} \widehat{h}_{\ell}(\tau) \widetilde{\vartheta}_{\frac{m}{N}, \ell}(\tau, z)
$$
transforms like a Jacobi form of weight $k$ and index $\frac{m}{N}$ with respect to $\Gamma\ltimes(N\mathbb{Z}\times\mathbb{Z})$. If $\psi$ has only simple poles then $E_s=0$ and the completion of $h_{\ell}$ simplifies; see Theorem \ref{thm_rational_index}.
\item The partition function for a general CHL model is given by $\widetilde{\Phi}_k^{-1}$, where $\Phi_k$ is a Siegel modular form of weight $k$ for a certain subgroup $\widetilde{G}$ of $\mathrm{Sp}(2, \mathbb{Z})$. We prove that the partition function $\widetilde{\Phi}_k^{-1}$ has the following Fourier-Jacobi expansion: 
\begin{equation}
\begin{split}
&\frac{1}{\widetilde{\Phi}_k(\tau,z,\sigma)}=\sum_{m=-\widetilde{\alpha}}^{\infty}\phi_m(\sigma,z)q^{m},\quad q=e^{2\pi i\tau},\quad \mathrm{Im}(\tau)\to\infty\\&\frac{1}{\widetilde{\Phi}_k(\tau,z,\sigma)}=\sum_{n=-\widetilde{\gamma}N}^{\infty}\psi_n(\tau,z)p^{n/N},\quad p=e^{2\pi i\sigma},\quad \mathrm{Im}(\sigma)\to\infty,
\end{split} 
\label{eq:FJexpPhikDD'}
\end{equation}
for some functions $\phi_m(\sigma,z)$ and $\psi_n(\tau,z)$ and constants $\widetilde{\alpha},\widetilde{\gamma}$.
In Theorem \ref{CHL_thm1}, we prove that  $\phi_m$ is a meromorphic Jacobi form of weight $-k$ and index $m$ with respect to $\Gamma^1(N)\ltimes\mathbb{Z}^2$ with poles at $z\in\mathbb{Z}\sigma+\mathbb{Z}$ and $\psi_n$ is a meromorphic Jacobi form of weight $-k$ and index $n/N$ with respect to $\Gamma_1(N)\ltimes(N\mathbb{Z}\times\mathbb{Z})$ with poles at $z\in N\mathbb{Z}\tau+\mathbb{Z}$.

\item As a consequence of results in Section \ref{sec:mockjacchl}, we show that the meromorphic Jacobi forms $\phi_m$ and $\psi_n$ admit a canonical decomposition $\phi_m=\phi_m^P+\phi_m^F$ and $\psi_n=\psi_n^P+\psi_n^F$, where the polar parts take the form 
\begin{equation}
\begin{split}
&\phi_m^P(\sigma,z)={\hat f^{(k)}(m)}{\tilde g^{(k)}(\sigma)}\sum_{s\in\mathbb{Z}}\frac{p^{ms^2+s}\zeta^{2ms+1}}{(1-p^{s}\zeta)^2},\\& \psi_n^P(\tau,z)=
{\hat g^{(k)}(n)}{\tilde f^{(k)}(\tau)}\sum_{s\in\mathbb{Z}}\frac{q^{N(ns^2+s)}\zeta^{2ns+1}}{(1-q^{Ns}\zeta)^2},\end{split}  
\end{equation}
where $\zeta=e^{2\pi iz}$ and $\tilde f^{(k)},\tilde g^{(k)}$ are certain functions appearing as residues of the partition function at the poles. Consequently, $\phi_m^F$ and $\psi_m^F$
are mixed mock Jacobi forms.
\item In Theorem \ref{CHL_thm2}, we provide a uniform proof of mock modularity of single-centered black hole partition function for $\mathbb{Z}_N$ orbifold of $K3\times T^2$ as well as $T^6$ compactifications of Type IIB string theory. The precise statement is the following.
Let $A$ denote the set of integers $(n,\ell,m)$ with $m,n>0,4mn-N\ell^2>0$, 
satisfying 
\begin{eqnarray}\label{eq:AdefMinintri}
&& 
    \left[\left({m_1\ell+2ms +m\over m}\right)^2 +{m_1^2\over m^2}
    \left(4 {n\over N} m -\ell^2\right)\right] \ge 1\, , \nonumber \\ &&
    \hskip 1in \forall \ m_1\in N\mathbb{Z},\quad s\in \mathbb{Z},\quad s(s+1)= 0\  {\rm mod}\ m_1\, ,
\end{eqnarray}
and $B$ denote the set of integers $(n,\ell,m)$ with $m,n>0,4mn-N\ell^2>0$,
satisfying 
\begin{eqnarray}\label{eq:BdefMinintro}
    &&  \left[\left({n_1 N\ell-2ns-n\over n}\right)^2 
    +{n_1^2N^2\over n^2}
    \left(4 {n\over N} m -\ell^2\right)\right] \ge 1\,, \nonumber \\
 && \hskip 1in \forall\ {n_1,s\in\mathbb{Z},\quad s(s+1)=0 \ {\rm mod}\ n_1}\, . 
\end{eqnarray}
We prove that
\begin{itemize}
\item The single-centered black hole degeneracies for charges $(Q,P)$ with $Q^2/2=m/N,(Q\cdot P)=\ell$ and $ P^2/2=n$ with $(n,\ell,m)$ belonging to set $A$ are Fourier coefficients of the mock Jacobi form $\phi_m^F$.
\item The single-centered black hole degeneracies for charges $(Q,P)$ with $Q^2/2=m/N,Q\cdot P=\ell, P^2/2=n$ with $(n,\ell,m)$ belonging to  set $B$ are Fourier coefficients of the mock Jacobi form $\psi_n^F$.
\end{itemize}
We note however that the sets $A$ and $B$ differ somewhat from the one specified by DMZ even for $N=1$ for which the original DMZ analysis was performed. In particular, we give examples of charges that are neither in the set $A$ nor in the set $B$ and hence their degeneracies are not given by Fourier coefficients of $\phi_m^F$ or $\psi_n^F$ (for $N=1$, $\phi_m^F=\psi_m^F$). 
\end{enumerate}
\text{}\\
This article is organized as follows. In Section \ref{sec2}, we review the definitions and examples of classical modular forms, mock modular forms, Siegel modular forms and Jacobi forms. Next, we review the results of DMZ in Section \ref{subsec:dmzmockjacform} and prove certain generalizations in Section \ref{rational_index}. Finally in Section \ref{sec:mockjacchl} we apply these results to establish the mock modularity of the single-centered black hole partition functions in CHL models. Appendix \ref{app:Phiktrans} contains some technical details of the transformation rules of the CHL partition functions and Appendix \ref{app:Cdeform} contains the proof of the validity of deformation of a certain integration contour required in Section \ref{sec:mockjacchl}.  
\section{Mathematical preliminaries}
\label{sec2}
In this section, we recall some basic facts about modular forms and mock modular forms. We follow \cite{bringmann2017harmonic,cohen2017modular,koblitz2012introduction,murty2016problems} for this exposition. Let   
\begin{equation}
\mathrm{SL}(2,\mathbb{Z}):=\left\{\begin{pmatrix}
a&b\\c&d
\end{pmatrix}:a,b,c,d\in\mathbb{Z},ad-bc=1\right\}.
\end{equation}  
Let $\mathbb{H}:=\{\tau=x+iy\in\mathbb{C}:y>0\}$ be the upper half-plane. Then $\mathrm{SL}(2,\mathbb{Z})$ acts on $\mathbb{H}$ as follows: 
\begin{equation}
\gamma \tau=\frac{a\tau+b}{c\tau+d},\quad \gamma=\begin{pmatrix}
a&b\\c&d\\
\end{pmatrix}\in \mathrm{SL}(2,\mathbb{Z}), \quad \tau\in\mathbb{H}.
\label{eq:sl2zonh}
\end{equation}
Let $\mathbb{P}^1(\mathbb{Q})=\mathbb{Q}\cup\{i\infty\}$.
We extend the action of $\mathrm{SL}(2,\mathbb{Z})$ on the \textit{extended upper half plane} $\mathbb{H}^{\star}:=\mathbb{H}\cup\mathbb{P}^1(\mathbb{Q})$ as follows. For $\gamma=\begin{pmatrix}
    a & b\\c & d\end{pmatrix},$
we define
\begin{equation}
\label{gamma_ext}
\gamma(\tau)=\left\{\begin{array}{cc}\frac{a\tau+b}{c\tau+d}, & \quad \tau \neq i\infty\\
\frac{a}{c},&\quad \tau= i\infty\\
i\infty, &\quad (c\tau+ d)= 0\;\text{or}\; c= 0.
\end{array}\right.
\end{equation}

Let $\Gamma$ be a finite index subgroup of $\mathrm{SL}(2,\mathbb{Z})$. A \emph{cusp} of $\Gamma$ is a $\Gamma$-orbit in $\mathbb{P}^1(\mathbb{Q})$. 
The set of $\Gamma$-orbits of $\mathbb{P}^1(\mathbb{Q})$ is called the set of cusps of $\Gamma.$ 
Note that for any $s=\frac{p}{q}\in\mathbb{Q},$ there is a $\gamma_s\in\mathrm{SL}(2,\mathbb{Z})$ such that $\gamma_s(i\infty)=\frac{p}{q}$. Since gcd$(p,q)=1$, $\gamma_s$ has the form
\begin{equation}
\gamma_s=\begin{pmatrix}
p&\star\\q&\star\\
\end{pmatrix}.
\end{equation}
In particular, $\mathrm{SL}(2,\mathbb{Z})$ has only one cusp $i\infty$. 
Let $[\rho]$ be any cusp of $\Gamma$ and $\rho$ be a representative of this cusp. Let $\gamma_\rho \in \mathrm{SL}(2,\mathbb{Z})$ be such that $\gamma_\rho (i\infty)=\rho$. Denote by $\Gamma_\rho$ the subgroup of $\gamma_\rho^{-1} \Gamma \gamma_\rho$ that fixes $i\infty$. Explicitly, $\Gamma_\rho=\gamma_\rho^{-1}\{g \in \Gamma : g \rho = \rho\}\gamma_\rho$. Then $\Gamma_\rho$ is
generated by $-I_2$ and $\begin{pmatrix} 
1&t_\rho\\0&1\end{pmatrix}$ for some positive integer $t_\rho$. The smallest positive integer $t_\rho$ satisfying this property is called the {\em cusp width} of the cusp 
$[\rho]$ with respect to $\Gamma$.  
We often  identify a cusp with a representative from the corresponding orbit and drop the $[~]$ from the symbol of
the cusp.
\\\\ For a positive integer $N$, we now define various finite-index subgroups of $\mathrm{SL}(2, \mathbb{Z})$ that we will have occasion to use throughout the paper.
\begin{equation}
\begin{split}
&\Gamma_0(N)\coloneqq \left\{\begin{pmatrix}
a&b\\c&d\\
\end{pmatrix}\in \mathrm{SL}(2,\mathbb{Z}):c~\equiv ~0~(\text{mod} ~N)\right\}\\
&\Gamma^0(N)\coloneqq \left\{\begin{pmatrix}
a&b\\c&d\\
\end{pmatrix}\in \mathrm{SL}(2,\mathbb{Z}):b~\equiv ~0~(\text{mod} ~N)\right\}\\
&\Gamma_1(N)\coloneqq \left\{\begin{pmatrix}
a&b\\c&d\\
\end{pmatrix}\in \mathrm{SL}(2,\mathbb{Z}):\begin{array}{c}
a \equiv d \equiv 1(\bmod~ N) \\
c \equiv 0(\bmod ~N)\end{array} \right\}\\
&\Gamma^1(N):=\left\{\begin{pmatrix}
a & b \\
c & d
\end{pmatrix}
\in \mathrm{S L}(2,\mathbb{Z}): \begin{array}{c}
a \equiv d \equiv 1(\bmod~ N) \\
b \equiv 0(\bmod~ N)
\end{array}\right\}\\
&\Gamma(N):=\left\{\begin{pmatrix}
a & b \\
c & d
\end{pmatrix} \in \mathrm{S L}(2,\mathbb{Z}): \begin{array}{c}
a \equiv d \equiv 1(\bmod~ N) \\
b\equiv c \equiv 0(\bmod~ N)
\end{array}\right\}
\end{split}
\end{equation}
A subgroup $\Gamma<\mathrm{SL}(2,\mathbb{Z})$ is called a \textit{congruence subgroup} of level $N$ if $\Gamma(N)\subset\Gamma$ and $N$ is the smallest positive integer with this property. Thus $\Gamma_0(N),\Gamma_1(N),\Gamma^0(N),\Gamma^1(N)$ are all congruence subgroups of level $N$.
\par Let $\kappa\in\frac{1}{2}\mathbb{Z}$. For a function $f:\mathbb{H}\longrightarrow\mathbb{C},$ define the weight-$\kappa$ slash operator as follows:
\begin{equation}
(f\underset{\kappa}{|}\gamma)(\tau)=f(\gamma \tau)\begin{cases}(c\tau+d)^{-\kappa}&\text{if $\kappa\in\mathbb{Z},~\gamma=\left(\begin{smallmatrix}a&b\\c&d\end{smallmatrix}\right)\in\mathrm{SL}(2,\mathbb{Z})$}\\\left(\frac{c}{d}\right)\varepsilon_d^{2\kappa}(c\tau+d)^{-\kappa}&\text{if $\kappa\in\frac{1}{2}+\mathbb{Z},~\gamma=\left(\begin{smallmatrix}a&b\\c&d\end{smallmatrix}\right)\in\Gamma_0(4)$}\end{cases},
\label{eq:automorphymodform}
\end{equation}
where $\left(\frac{c}{d}\right)$ is the Jacobi symbol and for odd $d$
\begin{equation}
\varepsilon_d=\begin{cases}
1&\text{if}~d\equiv 1\bmod ~4;\\i&\text{if}~d\equiv -1\bmod~ 4 
\end{cases}
\end{equation}
where we have chosen the principal branch of square root which maps the upper half space to the first quadrant and the lower half plane to the fourth quadrant. We assume that $\gamma\in\{\Gamma_0(N),\Gamma_1(N),\Gamma^1(N),\Gamma(N)\}$ for integral $\kappa$ and $\gamma\in\{\Gamma_0(N),\Gamma_1(N),\Gamma(N)\}$ with $4|N$ for half-integral $\kappa$. This restriction for $\gamma$ depending on weight is assumed throughout this paper.
\subsection{Mock modular forms}
In simple terms, modular forms are complex functions on the upper half-plane which have a special transformation rule with respect to the action of the modular group on the upper half-plane. The precise definition is given below. 
\begin{defn}
A holomorphic function $f:\mathbb{H}\longrightarrow\mathbb{C}$ is called a holomorphic modular form (resp. weakly holomorphic modular form resp. cusp form) of weight $\kappa$ for $\Gamma$ if $f\underset{\kappa}{|}\gamma=f$ for every $\gamma\in\Gamma$ and  $f$ is holomorphic (resp. weakly holomorphic, resp. vanishes) at all cusps of $\Gamma$ in the following sense: for every $\gamma\in\mathrm{SL}(2,\mathbb{Z})$, we have the following Fourier expansion
\begin{equation}
\left(f\underset{\kappa}{|}\gamma\right)(\tau)=\sum_{n\in\mathbb{Z}}c_{\gamma}(n)q^{n/h},
\label{eq:fourexpcusp}
\end{equation}
where $q=e^{2\pi i\tau}$ and $h$ is the cusp width of $\gamma(i\infty)$ with respect to $\Gamma$. Then $f$ is said to be holomorphic (resp. weakly holomorphic resp. vanishes) at the cusp $\gamma(i\infty)$ if $c_{\gamma}(n)=0$ unless $n\geq 0$ (resp. $n\geq n_0$ for some (possibly negative) integer, resp. $n>0$).
\end{defn}
\begin{remark}\label{rem:cuspcond}
The behaviour of a function $f:\mathbb{H}\to\mathbb{C}$ satisfying $f\underset{k}{|}\gamma=f$ for every $\gamma\in\Gamma$ at a (representative of a) cusp $\rho$ of $\Gamma$ is prescribed in terms of the behaviour of the function $f\underset{k}{|}\gamma_\rho$ at\footnote{Recall that $\gamma_\rho(i\infty)=\rho$.} $i\infty$. 
Note that, in general, $\Gamma<\mathrm{SL}(2,\mathbb{Z})$ has more that one inequivalent cusp. So the condition on the Fourier expansion \eqref{eq:fourexpcusp} for every $\gamma\in\mathrm{SL}(2,\mathbb{Z})$ simply guarantees the same condition at every $\Gamma$-inequivalent cusp. This is because if $s\in\mathbb{Q}\cup \{i\infty\}$ is a cusp then $\gamma s$ is a $\Gamma$-inequivalent cusp if $\gamma\in\mathrm{SL}(2,\mathbb{Z})\backslash\Gamma$. 
\end{remark} 
\noindent We have the following notation for the corresponding $\mathbb{C}$-vector spaces :
\begin{equation}
\begin{split}
&M_{\kappa}(\Gamma) ~:~\text{space of holomorphic modular forms}\\
&M_{\kappa}^!(\Gamma) ~:~\text{space of weakly holomorphic modular forms}\\&S_{\kappa}(\Gamma)~~:~\text{space of cusp forms}.
\end{split}
\end{equation}
The first examples of modular forms are Eisenstein series which can be defined by the following Fourier expansion: for $4\leq k\in 2\mathbb{Z}$; define the  Eisenstein series of weight $k$ by
\begin{equation}\label{eq:Ekqexp}
E_k(\tau)=1-\frac{2k}{B_k}\sum_{n=1}^{\infty}\sigma_{k-1}(n)
q^n,\quad q=e^{2\pi i\tau},
\end{equation}
where $\sigma_{k}(n)=\sum_{d|n}d^k$ is the $k^{\text{th}}$-divisor function and $B_k$ are Bernoulli numbers defined by 
 \begin{equation}
 \frac{x}{e^x-1}=\sum_{k=0}^{\infty}B_k\frac{x^k}{k!}.
 \end{equation}
$E_k$ is a modular form of weight $k$ for $\mathrm{SL}(2,\mathbb{Z})$. In fact, one can show that 
\begin{equation}
M_k(\mathrm{SL}(2,\Z))=\bigoplus_{\substack{4\ell_1+6\ell_2=k\\\ell_1,\ell_2\geq 0}}\mathbb{C}\cdot E_4^{\ell_1}E_6^{\ell_2}.
\end{equation} 
There are no modular forms of odd weight, negative weight and weight 2 for $\mathrm{SL}(2,\mathbb{Z})$. Moreover, $M_0(\mathrm{SL}(2,\mathbb{Z}))=\mathbb{C}$. There are no nontrivial cusp forms for $\mathrm{SL}(2,\mathbb{Z})$ for weight less that 12. Moreover there is a unique cusp form of weight 12 called the Ramanujan discriminant function:
\begin{equation}
\Delta(\tau)=q\prod_{n=1}^\infty(1-q^n)^{24}=\frac{E_4(\tau)^3-E_6(\tau)^2}{1728}.
\end{equation}
Note that if we define $E_2$ by substituting $k=2$ in \eqref{eq:Ekqexp}, it is a well defined holomorphic function. But it does not transform as a modular form. It is an example of a mock modular form which will be described below. It transforms as:
\begin{equation}\label{eq:E2tras}
E_2\left(\frac{a\tau+b}{c\tau+d}\right)=(c\tau+d)^2E_2(\tau)+\frac{6c(c\tau+d)}{\pi i},\quad \begin{pmatrix}
a&b\\c&d
\end{pmatrix}\in \mathrm{SL}(2,\mathbb{Z}).
\end{equation}
Eisenstein series can be generalised to congruence subgroups of $\mathrm{SL}(2,\mathbb{Z})$. \par Examples of half-integral weight modular forms include the classical theta function defined by 
\begin{equation}
\Theta(\tau)=\sum_{n\in\mathbb{Z}}q^{n^2},\quad q=e^{2\pi i\tau}.
\end{equation}   
It is a modular form of weight 1/2 for $\Gamma_0(4)$. Another important function is the Dedekind eta function defined as 
\begin{equation}
\eta(\tau)=q^{\frac{1}{24}}\prod\limits_{n=1}^{\infty}(1-q^n),\quad q=e^{2\pi i\tau},
\end{equation}
so that $\eta^{24}=\Delta$. It transforms similar to a modular form of weight 1/2 but with a phase:\begin{equation}
\eta(\tau+1)=e^{\frac{\pi i}{12}}\eta(\tau),\quad\eta\left(-\frac{1}{\tau}\right)=\sqrt{-i\tau}\eta(\tau).
\end{equation} 
\par We now define mock modular forms.
For $\kappa\in\R$, introduce the weight-$\kappa$ hyperbolic Laplacian (here $\tau=x+iy$)
\begin{equation}
\Delta_{\kappa}\coloneqq -y^2\left(\frac{\partial^2}{\partial x^2}+\frac{\partial^2}{\partial y^2}\right)+i\kappa y\left(\frac{\partial}{\partial x}+i\frac{\partial}{\partial y}\right)=-4y^2\frac{\partial}{\partial \tau}\frac{\partial}{\partial \bar{\tau}}+2i\kappa y\frac{\partial}{\partial \bar{\tau}}.
\end{equation}
\begin{defn}
A smooth function (in real sense) $f:\mathbb{H}\rightarrow\mathbb{C}$ is called a harmonic Maass form\footnote{In some papers, the term harmonic weak Maass form is also used to mean harmonic Maass form.} of weight $\kappa\in\frac{1}{2}\Z$ for the group $\Gamma$ if 
\begin{enumerate}
\item $f\underset{\kappa}{|}\gamma=f$ for every $\gamma\in\Gamma.$
\item $\Delta_{\kappa}(f)=0$.
\item There exists a polynomial $P_f(\tau)\in\mathbb{C}[q^{-1}]$ such that $f(\tau)-P_f(\tau)=O(e^{-\epsilon y})$ as $y\rightarrow\infty$ for some $\epsilon>0$. Similar conditions hold at other cusps (see Remark \ref{rem:cuspcond}). The polynomial $P_f(\tau)$ is called the principal part of $f$. 
\end{enumerate}
If the third condition in the above definition is replaced by $f(\tau)=O(e^{\epsilon y})$, then $f$ is said to be a harmonic weak Maass form of \emph{manageable growth} which we shall simply call a harmonic Maass form. The space of harmonic weak Maass forms of weight $\kappa$ for $\Gamma$ is denoted by $H_{\kappa}(\Gamma)$ and the space of harmonic weak Maass forms of manageable growth is denoted by $H_{\kappa}^{!}(\Gamma)$.  
\end{defn}
For $\Gamma\in\{\Gamma_0(N),\Gamma_1(N)\}$, any   
$f\in H_{\kappa}^{!}(\Gamma)$ for $\kappa\neq 1$ has a Fourier expansion \cite[Lemma 4.3]{bringmann2017harmonic} of the form  
\begin{equation}
f(\tau)=f(x+iy)=\sum\limits_{n>> -\infty}c_f^{+}(n)q^{n}+ c_f^{-}(0)y^{1-\kappa}+\sum\limits_{\substack{n<< \infty\\n\neq 0}}c_f^{-}(n)\Gamma(1-\kappa,-4\pi ny)q^{n},
\label{eq 1.1}
\end{equation}
where $\Gamma(s,z)$ is the incomplete gamma function defined as 
\begin{equation}
\Gamma(s,z)\coloneqq\int\limits_z^{\infty}e^{-t}t^s\frac{dt}{t}.
\end{equation}
The notation $\sum\limits_{n>> -\infty}$ means $\sum\limits_{n=\alpha_f}^{\infty}$ for some $\alpha_f\in\mathbb{Z}$ and $\sum\limits_{n<<\infty}$ means $\sum\limits_{n=-\infty}^{\beta_f}$ for some $\beta_f\in\mathbb{Z}$. When $\kappa=1$, the term $y^{1-\kappa}$ is replaced by $\log y$ in the Fourier expansion \eqref{eq 1.1}.
We call 
\begin{equation}
f^+(\tau)=\sum\limits_{n>> -\infty}c_f^{+}(n)q^{n}
\end{equation}
the \emph{holomorphic part} of $f$ and 
\begin{equation}
f^-(\tau)=c_f^{-}(0)y^{1-\kappa}+\sum\limits_{\substack{n<< \infty\\n\neq 0}}c_f^{-}(n)\Gamma(1-\kappa,-4\pi ny)q^{n}
\end{equation}
 the \emph{nonholomorphic part} of $f$. If $f\in H_{\kappa}(\Gamma)$ then $c_f^-(0)=0$ and the sum in the nonholomorphic part runs only over negative integers. \par Harmonic weak Maass forms are related to modular forms via the \emph{shadow operator} defined as\footnote{Note that our definition of shadow operator differs from \cite{Dabholkar:2012nd} by a factor of $-(4\pi)^{\kappa-1}$, that is, $\overline{\xi_\kappa^{\text{DMZ}}}=-(4\pi)^{\kappa-1}\xi_\kappa$. This normalization
 follows \cite{bringmann2017harmonic}.}
\begin{equation}\label{e218x}
\xi_{\kappa}\coloneqq 2iy^{\kappa}\overline{\frac{\partial}{\partial\bar{\tau}}},
\end{equation} 
where $\overline{\frac{\partial}{\partial\bar{\tau}}}$ corresponds to complex conjugation after taking the deivative with respect to $\Bar{\tau}$.
The image of the shadow operator is a weakly holomorphic modular form \cite[Theorem 5.10]{bringmann2017harmonic}. Precisely, for $\kappa\neq 1$, we have
\begin{equation}
\xi_{2-\kappa}(f)\in\begin{cases}M_{\kappa}^!(\Gamma)&\text{if $f\in H_{2-\kappa}^{!}(\Gamma)$}\\
S_{\kappa}(\Gamma) &\text{if $f\in H_{2-\kappa}(\Gamma).$}
\end{cases}
\end{equation}
The image of $f$ under $\xi$ is called the \emph{shadow} of $f$. Moreover this map is surjective \cite{bringmann2017harmonic}. 
\begin{defn}
A \emph{pure mock modular form} of weight $\kappa$ is the holomorphic part $f^+$ of a harmonic Maass form of weight $\kappa$ for which $f^-$ is nontrivial. The weakly holomorphic modular form $\xi_{\kappa}(f)$ is called the shadow of the mock modular form $f^+$. The harmonic Maass form $f$ is called the \emph{completion} of the mock modular form $f^+$. 
\end{defn}
\noindent Given a mock modular form $f^+$ of weight $\kappa$ with shadow 
\begin{equation}
g(\tau)=\sum_{n=0}^{\infty}c_g(n)q^n\in M^!_{2-\kappa}(\Gamma),
\end{equation}
one can construct the corresponding harmonic weak Maass form $\widehat{f}=f^++f^-$ using the \textit{nonholomorphic Eichler integral}:  
\begin{equation}
f^-(\tau)=-\frac{\overline{c_g(0)} y^{1-\kappa}}{\kappa-1}-(4\pi)^{\kappa-1}\sum_{n>0} n^{\kappa-1} \overline{c_g(n)} \Gamma\left(1-\kappa, 4 \pi n y\right) q^{-n}.
\label{eq:eichintsum}
\end{equation}
The first term should be replaced by $-\overline{c_g(0)} \log y$ if $\kappa=1$. The nonholomorphic part $f^-$ satisfies the \emph{Eichler differential equation} \cite{bringmann2017harmonic,Dabholkar:2012nd} 
\begin{equation}
2i y^{\kappa} \frac{\partial f^-(\tau)}{\partial \bar{\tau}}=- \overline{g(\tau)} .
\end{equation}
When $\kappa>1$ or $c_{g}(0)=0$, we can define $f^-$ by the integral
\begin{equation}
f^-(\tau)=2^{\kappa-1}i \int_{-\bar{\tau}}^{i\infty}d \omega \frac{\overline{g(-\bar{\omega})}}{(-i(\omega+\tau))^{\kappa}}  .
\label{eq:eichintint}
\end{equation}
Moreover, since $f^+$ is holomorphic, it is clear that 
\begin{equation}
2i y^{\kappa}  \frac{\partial \widehat{f}(\tau)}{\partial \bar{\tau}}=- \overline{g(\tau)} .
\label{eq:diffeqharmass}
\end{equation}
\par The first example of a mock modular form is $E_2$ with its completion given by 
\begin{equation}
E_2^*(\tau)=E_2(\tau)-\frac{3}{\pi y},\quad \tau=x+iy.
\end{equation}
Using the transformation \eqref{eq:E2tras}, it is easy to show that $E_2^*$ transforms like a modular form of weight $2$ for $\mathrm{SL}(2,\Z)$. 
The second example that we will be interested in is the generating function of Hurwitz class numbers:
\begin{equation}
    \textbf{H}(\tau):=-\frac{1}{12}+\sum\limits_{n=1}^{\infty}H(n)q^n,
\end{equation}
where $H(n)$ is the Hurwitz class number of discriminant $(-n)$, see for example \cite{cohen2017modular} for the definition of Hurwitz class numbers. It was proved by Zagier \cite{Zagier1975NombresDC} that $\textbf{H}(\tau)$ is a mock modular form of weight 3/2 for $\Gamma_0(4)$ with shadow $-\frac{1}{16}\Theta$. More precisely,
\begin{equation}\label{eq:Htmock3/2}
\mathcal{H}(\tau)\coloneqq -\frac{1}{12}+\sum\limits_{n=1}^{\infty}H(n)q^n+\frac{1}{4\sqrt{\pi}}\sum\limits_{n=1}^{\infty}n\Gamma\left(-\frac{1}{2},4\pi n^2y\right)q^{-n^2}+\frac{1}{8\pi\sqrt{y}}
\end{equation}
is a harmonic Maass form of weight $3/2$ for $\Gamma_0(4)$, see \cite{bhand2022zagiers} for a detailed proof. \\\\These are the simplest kind of mock modular forms. Let us now define mixed mock modular forms. Let $\kappa_1,\kappa_2\in\frac{1}{2}\Z$.
\begin{defn}\label{mixedmmfdef}
A \emph{mixed mock modular form} of mixed weight $(\kappa_1,\kappa_2)$ is a holomorphic function $h$  which admits a completion of the form
\begin{equation}\label{eq:mixmockcomp}
   \widehat{h}(\tau):=h(\tau)+\sum_{j=1}^{r} f_{j}(\tau) g_{j}^*(\tau) 
\end{equation}
where for each $j,$ we have $f_{j} \in M_{\kappa_1}^{!}(\Gamma)$ and $g_{j} \in M_{2-\kappa_2}^{!}(\Gamma)$, where $g_j^*$ is the Eichler integral of $g_j$ as in \eqref{eq:eichintsum}, such that $\widehat{h}(\tau)\in H_{\kappa_1+\kappa_2}^!(\Gamma)$.  A \textit{mixed mock modular form} of weight $\kappa\in\frac{1}{2}\Z$ is a holomorphic function $h$ with completion as in \eqref{eq:mixmockcomp} with $f_{j} \in M_{\kappa_j}^{!}(\Gamma)$ and $g_{j} \in M_{2-\kappa'_j}^{!}(\Gamma)$ with $\kappa_j+\kappa_j'=\kappa$ and $\kappa_j,\kappa_j'\in\frac{1}{2}\Z$ for each $j=1,\dots,r$.   
\end{defn}
\noindent The shadow of a mixed harmonic Maass form $h$ as above is defined as 
\begin{equation}
\xi_{\kappa_2}(h)=\sum_{i=1}^r\overline{f_j(\tau)}\xi_{\kappa_2}(g_j^*(\tau))=\sum_{i=1}^r\overline{f_j(\tau)}g_j(\tau).
\end{equation}
We will construct examples of mixed mock modular forms in subsequent sections.
\subsection{Weak Jacobi forms}\label{subsec:weakjac}
In this section, we review Jacobi forms. We follow \cite{eichler2013theory} for this exposition. 
Let $\mathrm{GL}^+(2,\mathbb{R})$ be the group of $2\times 2$ real matrices with positive determinant. Let $k,m\in\mathbb{R}$ with $m$ positive. Let $\mathbf{e}(t):=e^{2\pi it}$. Define the weight $k$ and index $m$ slash operator with respect to the group $\mathrm{GL}^+(2,\mathbb{R})\ltimes\mathbb{R}^2$ on functions $\varphi:\mathbb{H}\times\mathbb{C}\longrightarrow\mathbb{C}$ as follows:
 \begin{enumerate}
 \item \emph{Modular transformation:} 
\begin{equation}
\begin{split}
\left(\varphi\underset{k, m}{|}\gamma\right)(\tau, z):=(\text{det}\gamma)^{k/2}(c \tau+d)^{-k} \mathbf{e}\left(-\frac{m}{\text{det}\gamma}\left(\frac{c z^{2}}{c \tau+d}\right)\right) \varphi\left(\frac{a \tau+b}{c \tau+d}, \frac{z}{c \tau+d}\right),\\\gamma=\begin{pmatrix}
a & b \\
c & d
\end{pmatrix}\in \mathrm{GL}^+(2,\mathbb{R})
\end{split}
\end{equation}
\item \emph{Elliptic transformation:}
\begin{equation}
\left(\varphi\underset{m}{|}(\lambda, \mu)\right)(\tau, z):= \mathbf{e}\left(m\left(\lambda^{2} \tau+2 \lambda z+\lambda\mu\right)\right) \varphi(\tau, z+\lambda \tau+\mu),\quad(\lambda, \mu) \in \mathbb{R}^{2}.
\end{equation}
\end{enumerate}
One can easily check that for every $\gamma,\gamma'\in \mathrm{GL}^+(2,\mathbb{R})$ and $X=(\lambda,\mu),X'=(\lambda',\mu')\in\mathbb{R}^2$, we have 
\begin{equation}
\left(\varphi\underset{k, m}{|}\gamma\right)\underset{k, \frac{m}{\text{det}\gamma}}{|}\gamma'=\varphi\underset{k, m}{|}\left(\gamma\gamma'\right),
\label{eq:phigg'mdetg}
\end{equation}
\begin{equation}
(\varphi\underset{m}{|}X)\underset{m}{|} X^{\prime}=\mathbf{e}(m(\lambda\mu'-\lambda'\mu))\varphi\underset{m}{|}\left(X+X^{\prime}\right),
\end{equation}
and 
\begin{equation}
\begin{split}
\left(\varphi\underset{k, m}{|}\gamma\right)\underset{\frac{m}{\text{det}\gamma}}{|}X\gamma=\left(\varphi\underset{m}{|}X\right)\underset{k, m}{|}\gamma.
\end{split}
\label{eq:Xgg'Xmdetg}
\end{equation}
For a subgroup $\Gamma\subseteq\mathrm{SL}(2,\mathbb{Z})$ we can combine the elliptic and modular action into one slash operator of weight $k$ and index $m$ as follows: 
\begin{eqnarray}\label{eq:phi|[gX]}
\varphi\underset{k, m}{|}[\gamma,X]=\left(\varphi\underset{k,m}{|}\gamma\right)\underset{m}{|}X.  
\end{eqnarray}
More explicitly we have
\begin{equation}
\begin{split}
\left(\varphi\underset{k, m}{|}\left[\begin{pmatrix}
a & b \\
c & d
\end{pmatrix},(\lambda,\mu)\right]\right)(\tau,z)=(c\tau+d)^{-k}&\mathbf{e}\left(m\left(-\frac{c (z+\lambda\tau+\mu)^{2}}{c \tau+d}+\lambda^{2} \tau+2 \lambda z+\lambda\mu\right)\right)\\&\times\varphi\left(\frac{a\tau+b}{c\tau+d},\frac{z+\lambda\tau+\mu}{c\tau+d}\right)
\end{split}
\end{equation}
We can put a group structure on $\Gamma\times\mathbb{Z}^2$ as
\begin{equation}
[\gamma,X]\cdot[\gamma',X']=[\gamma\gamma',X'+X\gamma']
\end{equation} 
which turns $\Gamma\times\mathbb{Z}^2$ into the semidirect product group $\Gamma\ltimes\mathbb{Z}^2$. 
Then using \eqref{eq:phi|[gX]}, one can show that the weight $k$ and index $m$ slash operator defines a group action of $\Gamma\ltimes\mathbb{Z}^2$ on complex functions on $\mathbb{H}\times\mathbb{C}$; that is 
\begin{equation}
\left(\varphi\underset{k, m}{|}[\gamma,X]\right)\underset{k, m}{|}[\gamma',X']=\varphi\underset{k, m}{|}[\gamma\gamma',X'+X\gamma'].
\end{equation}
\begin{defn}\label{jacobidef}
Let $k,m\in\mathbb{Z}$ and $\varphi:\mathbb{H}\times\mathbb{C}\longrightarrow\mathbb{C}$ be a function satisfying
\begin{enumerate}
\item\label{it:jacobi_transformation} $\varphi\underset{k, m}{|}[M,X]=\varphi$ for all $[M,X]\in\Gamma\ltimes\mathbb{Z}^2$.
\item For every $M\in\mathrm{SL}(2,\mathbb{Z})$, $\varphi\underset{k,m}{|}M$ has a Fourier expansion of the form:
\begin{equation}
\left(\varphi\underset{k,m}{|}M\right)(\tau,z)=\sum\limits_{n,r\in\mathbb{Z}}c_M(n,r)q^{n/h}\zeta^r
\end{equation}
where $q=\mathbf{e}(\tau), \zeta=\mathbf{e}(z).$
\end{enumerate}
The function 
$\varphi$ is called a holomorphic (cusp or weak) Jacobi form of weight $k$ and index $m$ if it is holomorphic on $\mathrm{H}\times\mathbb{C}$ and 
\begin{equation}
\begin{aligned}
c_M(n,r)=0\quad\text{unless}\quad &(4nm-r^2h)\geq 0\quad&(\text{holomorphic})\\
&(4nm-r^2h)> 0\quad&(\text{cusp})
\\&n\geq 0\quad&(\text{weak})
\\&n\geq n_0~(\text{possibly negative})\quad&(\text{weakly holomorphic}),
\label{eq 1.2}
\end{aligned}
\end{equation}
where $h$ is the width of the cusp $M(i\infty)$ with respect to $\Gamma$.
A function $\varphi:\mathbb{H}\times\mathbb{C}\to\mathbb{C}$ is called a \textit{meromorphic Jacobi form} if it satisfies Condition \ref{it:jacobi_transformation} above but is meromorphic in $z$ and weakly holomorphic in $\tau$ (i.e., could have pole 
at $\tau=i\infty$ for fixed $z$.) 

The space of Jacobi forms (resp. Jacobi cusp forms, resp. weak Jacobi forms) of weight $k$ and index $m$ for $\Gamma$ is denoted by $J_{k,m}(\Gamma)$ (resp. $J_{k,m}^{\text{cusp}}(\Gamma)$, resp. $J_{k,m}^{\mathrm{weak}}(\Gamma)$). 
\end{defn}
For $m\in\mathbb{Z}$, any smooth function $\phi(\tau,z)$, such that $z\to\phi(\tau,z)$ is holomorphic for every $\tau\in\mathbb{H}$, satisfies $\phi\underset{m}{|}(\lambda,\mu)=\phi$ for all $\lambda,\mu\in\mathbb{Z}$ if and only if it has a \emph{theta decomposition} of the form \cite{eichler2013theory}
\begin{equation}
\phi ( \tau , z ) = \sum \limits_ { \ell \in \mathbb { Z } / 2 m \mathbb{Z} } h _ { \ell } ( \tau ) \vartheta _ { m , \ell } ( \tau , z ),
\label{eq 1.4}
\end{equation} 
where $\vartheta_{m,\ell}$ is the index $m$ Jacobi theta function defined by 
\begin{equation}\label{eq:thetamldef}
\vartheta _ { m , \ell } ( \tau , z )\coloneqq \sum \limits_ {\substack{ r\in \mathbb{Z}\\ r \equiv \ell(\textrm{mod}\ 2m) } } q ^ { r ^ { 2 } / 4 m } \zeta ^ { r } = \sum \limits_ { n \in \mathbb { Z } } q ^ { ( \ell + 2 m n ) ^ { 2 } / 4 m } \zeta ^ { \ell + 2 m n }.
\end{equation}
Indeed, from $\phi\underset{m}{|}(0,1)=\phi$, we get $\phi(\tau,z+1)=\phi(\tau,z)$ which gives us a Fourier expansion for $\phi:$
\begin{equation}
    \phi(\tau,z)=\sum_{\ell\in\Z}c_\ell(\tau)\zeta^\ell
\end{equation}
for some smooth functions $c_\ell:\mathbb{H}\longrightarrow\mathbb{C}$. Then the identity $\phi\underset{m}{|}(1,0)=\phi$  implies that 
\begin{equation}
    \sum_{\ell\in\Z}c_\ell(\tau)q^{\ell}\zeta^\ell
    =q^{-m}\zeta^{-2m}\sum_{\ell\in\Z}c_\ell(\tau)\zeta^\ell=q^{-m}\sum_{\ell\in\Z}c_{\ell+2m}(\tau)\zeta^\ell,
\end{equation}
which implies 
\begin{eqnarray}
 c_\ell(\tau)q^{\ell+m}= c_{\ell+2m}(\tau).  
\end{eqnarray}
As a result, $c_\ell(\tau)q^{-\ell^2/4m}$ depends only on $\ell\bmod 2m$:
\begin{eqnarray}
    c_\ell(\tau)q^{-\ell^2/4m}=c_{\ell+2mk}(\tau)q^{-(\ell+2mk)^2/4m},\quad k\in\Z.
\end{eqnarray}
This gives us the expansion \eqref{eq 1.4} with the identification $h_\ell(\tau)=c_\ell(\tau)q^{-\ell^2/4m}$.\par
The theta function $\vartheta_{m,\ell}$ transforms as follows:
\begin{equation}
\begin{split}
&\vartheta_{m,\ell}(\tau+1,z)=e^{2\pi i\ell^2/4m}\vartheta _ { m , \ell } ( \tau , z ),\\& \vartheta _ { m , \ell } \left( -\frac{1}{\tau} , \frac{z}{\tau} \right)=\sqrt{\frac{\tau}{2mi}}\exp\left(\frac{2\pi imz^2}{\tau}\right)\sum_{r~(\text{mod}~2m)}e^{-2\pi ir\ell/2m}\vartheta _ { m , r } ( \tau , z ),\\&\vartheta_{m,\ell}(\tau,z+\lambda\tau+\mu)=e^{-2\pi im(\lambda^2\tau+2\lambda z+\lambda\mu)}\vartheta _ { m , \ell } ( \tau , z ),\quad (\lambda,\mu)\in\Z^2.
\end{split}
\end{equation}
If $\phi$ is holomorphic on $\mathbb{H}\times\mathbb{C}$, then the theta coefficients $h_{\ell}(\tau)$ can be given by the integral formula:
\begin{equation}\label{eq:hlthcoeffint}
h_{\ell}(\tau):=h_{\ell}^{\left(z_{0}\right)}(\tau)=q^{-\ell^{2} / 4 m} \int_{z_{0}}^{z_{0}+1} \varphi(\tau, z) \mathbf{e}(-\ell z) d z,
\end{equation}
where $z_0$ is an arbitrary point of $\mathbb{C}.$ If in particular $\phi\in J_{k,m}^{\mathrm{weak}}(\Gamma)$ with $\Gamma\in\{\Gamma_0(N),\Gamma_1(N)\}$ then each $h_{\ell}(\tau)$ is a weakly holomorphic modular form of weight $k-\frac{1}{2}$ with respect to $\Gamma(4mN)$, see \cite{DAS2015351} for a proof. \par An important class of useful examples of index 1 is provided by the functions \cite{eichler2013theory} 
\begin{equation}\label{eq:phik1defCk}
\varphi_{k,1}(\tau,z)=\sum_{n,r\in\Z}C_k(4n-r^2)q^n\zeta^r,\quad k=-2,0,10,12,
\end{equation}
where $C_k(4n-r^2)$ are some coefficients. More explicitly, one has 
\begin{equation}
\begin{split}
&\varphi_{-2,1}=\frac{(\zeta-1)^2}{\zeta}-2\frac{(\zeta-1)^4}{\zeta^2}q+2\frac{(\zeta-1)^4(\zeta^2-8\zeta+1)}{\zeta^3}q^2+\dots\\
&\varphi_{0,1}=\frac{\zeta^2+10\zeta+1}{\zeta}+2\frac{(\zeta-1)^2(5\zeta^2-22\zeta+5)}{\zeta^2}q+\dots
\end{split}
\end{equation}
$\varphi_{10,1}$ and $\varphi_{12,1}$ is determined by 
\begin{equation}
\varphi_{10,1}=\varphi_{-2,1}\Delta,\quad \varphi_{12,1}=\varphi_{0,1}\Delta.
\end{equation}
In terms of Jacobi theta functions we have 
\begin{equation}\label{eq:ABdef}
\begin{split}
&A(\tau,z):=\varphi_{-2,1}(\tau,z)=\frac{\vartheta_1^2(\tau,z)}{\eta^6(\tau)},\quad \varphi_{10,1(\tau,z)}=\eta^{24}(\tau)\varphi_{-2,1}(\tau,z)=\eta^{18}(\tau)\vartheta_1^2(\tau,z)\\&B(\tau,z):=\varphi_{0,1}(\tau,z)=4\left(\frac{\vartheta_2^2(\tau,z)}{\vartheta_2^2(\tau)}+\frac{\vartheta_3^2(\tau,z)}{\vartheta_3^2(\tau)}+\frac{\vartheta_4^2(\tau,z)}{\vartheta_4^2(\tau)}\right),
\end{split}
\end{equation}
where $\vartheta_i(\tau)=\vartheta_i(\tau,0)$ and $\vartheta_i(\tau,z)$ are Jacobi theta functions defined by 
\begin{equation}
\begin{split}
&\vartheta_1(\tau,z)=\sum_{n\in\mathbb{Z}}(-1)^nq^{\frac{1}{2}(n-\frac{1}{2})^2}\zeta^{n-\frac{1}{2}},\quad \vartheta_2(\tau,z)=\sum_{n\in\mathbb{Z}}q^{\frac{1}{2}(n-\frac{1}{2})^2}\zeta^{n-\frac{1}{2}},\\&\vartheta_3(\tau,z)=\sum_{n\in\mathbb{Z}}q^{\frac{n^2}{2}}\zeta^{n},\quad \vartheta_4(\tau,z)=\sum_{n\in\mathbb{Z}}(-1)^nq^{\frac{n^2}{2}}\zeta^{n}.
\end{split} 
\end{equation} 
The functions $A$ and $B$ are the unique weak Jacobi forms of index 1 and weight $\leq 0$ and generate the ring of weak Jacobi forms of even weight over the ring of modular forms for $\mathrm{SL}(2,\mathbb{Z})$:
\begin{equation}
J_{k,m}^{!}(\mathrm{SL}(2,\mathbb{Z}))=\bigoplus_{j=0}^mM_{k+2j}(\mathrm{SL}(2,\mathbb{Z}))\cdot A^jB^{m-j}.
\end{equation}
On the other hand $\varphi_{10,1}$ and $\varphi_{12,1}$ are Jacobi cusp forms of index 1 and smallest possible weight. 
\subsection{Siegel modular forms}
Siegel modular forms are generalisations of modular forms and are intimately connected to Jacobi forms. We will restrict to the discussion of Siegel modular forms of degree 2. We define the degree 2 symplectic group as 
\begin{equation}
\mathrm{Sp}(2,\mathbb{R})\coloneqq \{g\in \mathrm{GL}(4,\mathbb{R}):gJg^t=J\}
\end{equation}  
where 
\begin{equation}
J:=\begin{pmatrix}
0&-I_2\\I_2&0
\end{pmatrix}
\end{equation}
with $I_2$ the $2\times 2$ identity matrix. 
The Siegel upper half space $\mathbb{H}_2$ of degree two is a generalisation of the upper half space $\mathbb{H}$ and is defined as
\begin{equation}
\mathbb{H}_2=\bigg\{\Omega=\begin{pmatrix}
    \tau&z\\z&\sigma
    \end{pmatrix}:\tau,\sigma\in\mathbb{H};z\in\mathbb{C},\text{det}(\text{Im}(\Omega))>0\bigg\}.
\end{equation} 
The symplectic group acts on $\mathbb{H}_2$ as follows \cite{77c4633f-ad35-333d-8e41-59ec3cf33f0b} 
\begin{equation}
(M,\Omega)\mapsto M\circ \Omega\coloneqq (A\Omega+B)(C\Omega+D)^{-1},\quad M=\begin{pmatrix}
A&B\\C&D
\end{pmatrix}\in \mathrm{Sp}(2,\mathbb{R}).
\label{eq:MonH2}
\end{equation}
Let $\Phi:\mathbb{H}_2\longrightarrow\mathbb{C}$ be a function. Then we define the weight $k\in\mathbb{Z}$ slash operator on $f$ as follows \footnote{Note that the determinant of any matrix in $\mathrm{Sp}(2,\mathbb{R})$ is 1, so we do not have any $(\mathrm{det}(M))^{k/2}$ factor in this case as compared to the action of $\mathrm{GL}^+(2,\mathbb{R})$ case.}:
\begin{equation}
(\Phi\underset{k}{|}M)(\Omega):=\text{det}(C\Omega+D)^{-k}\Phi(M\circ\Omega)
\label{eq:siegslashopprop}
\end{equation}
where $M$ is as in \eqref{eq:MonH2}. One can show that 
\begin{equation}
(\Phi\underset{k}{|}M_1)\underset{k}{|}M_2=(\Phi\underset{k}{|}M_1M_2),\quad  M_1,M_2\in\mathrm{Sp}(2,\mathbb{R}).
\label{eq:siegslashopprop}
\end{equation}
Note that a function $\Phi:\mathbb{H}_2\longrightarrow\mathbb{C}$ can be considered as a complex function of three variables and then one can talk about its analytic properties in the usual way. Most of the times we will write $\Phi(\Omega)=\Phi(\tau,z,\sigma)$ where 
\begin{equation}
\Omega=\begin{pmatrix}
    \tau&z\\z&\sigma
    \end{pmatrix}.
\end{equation} 
\begin{defn}\label{siegeldef}
Let $G$ be a subgroup of $\mathrm{Sp}(2,\mathbb{Z})$. A holomorphic function $\Phi:\mathbb{H}_2\longrightarrow\mathbb{C}$ is called a Siegel modular form of weight $k$ for $G$ if 
\begin{equation}
(\Phi\underset{k}{|}M)(\Omega)=\Phi,\quad \forall~~M\in G.
\end{equation}  
\end{defn}
Depending on $G$, the function $\Phi$ satisfies periodicity in the three variables because of the modular invariance. Let $N_1,N_2,N_3$ be the least positive integers such that
\begin{equation}
\begin{pmatrix}
1&0&N_1&N_2\\
0&1&N_2&N_3\\
0&0&1&0\\
0&0&0&1
\end{pmatrix}\in G.
\end{equation}
Then the modular invariance for these matrices shows that 
\begin{equation}
\Phi(\tau+N_1,z+N_2,\sigma+N_3)=\Phi(\tau,z,\sigma).
\end{equation}
For example when $G=\mathrm{Sp}(2,\mathbb{Z})$, we have $N_1=N_2=N_3=1$. We have a Fourier expansion for $\Phi$:
\begin{equation}
\Phi(\tau,z,\sigma)=\sum_{n,m,r\in\mathbb{Z}}a(n,m,r)q^{n/N_1}\zeta^{r/N_2}p^{m/N_3},\quad q=\mathbf{e}(\tau),\zeta=\mathbf{e}(z),p=\mathbf{e}(\sigma).
\end{equation} 
The sum over indices $n,m,r$ may be restricted depending on the analytic properties of $\Phi$. Rearranging the Fourier expansion, one can write 
\begin{equation}
\Phi(\tau,z,\sigma)=\sum_{m}\varphi_m(\tau,z)p^{m/N_3}.
\label{eq:fourjacexp}
\end{equation} 
Then using the embedding of $\mathrm{SL}(2,\mathbb{Z})$ and $\mathbb{Z}^2$ into $\mathrm{Sp}(2,\mathbb{Z})$ given by 
\begin{equation}
\begin{pmatrix}
a&b\\c&d
\end{pmatrix}\mapsto\begin{pmatrix}
a&0&b&0\\
0&1&0&0\\
c&0&d&0\\
0&0&0&1
\end{pmatrix},\quad (\lambda,\mu)\mapsto\begin{pmatrix}
1&0&0&\lambda\\
-\mu&1&\lambda&0\\
0&0&1&\mu\\
0&0&0&1
\end{pmatrix}, 
\end{equation}
one can show that each $\varphi_m$ transforms like a Jacobi form of weight $k$ and index $m$ with respect to some subgroup of $\mathrm{SL}(2,\mathbb{Z})\ltimes\mathbb{Z}^2$. For this reason the expansion \eqref{eq:fourjacexp} is called the Fourier-Jacobi expansion of $\Phi$.\\\\One of the first examples of a Siegel modular form is the Igusa cusp form defined as 
\begin{equation}
\Phi_{10}(\tau,z,\sigma)=q\zeta p\prod_{(r,s,t)>0}(1-q^s\zeta^t p^r)^{2C_0(4rs-t^2)}
\end{equation}
where $(r,s,t)>0$ means that $r,s,t\in\Z$ with with either $r>0$ or $r=0,s>0$, or $r=s=0,t<0$ and $C_0$ is the coefficient appearing in \eqref{eq:phik1defCk}. $\Phi_{10}$ is a Siegel cusp form of weight 10 for the Siegel modular group Sp$(2,\Z)$. For another class of Siegel modular forms of weight $k$ and a subgroup $\widetilde{G}$ of Sp$(2,\Z)$; see Section \ref{sec:mockjacchl} and Appendix \ref{app:Phiktrans}.    
\section{Results of Dabholkar-Murthy-Zagier}\label{subsec:dmzmockjacform}
In this section, we review the results of \cite{Dabholkar:2012nd} on mock Jacobi forms of $\mathrm{SL}(2,\mathbb{Z})$ and generalize them to congruence subgroup of $\mathrm{SL}(2,\mathbb{Z})$. 
For those parts of DMZ analysis which can be used without any modification in our analysis of fractional index mock Jacobi forms in Section \ref{rational_index}, we provide detailed proofs of the results. For the other parts, we state the results but postpone the proof to Section \ref{rational_index} where we prove more general versions of the result that can be used for fractional index mock Jacobi forms. \\\\ The notion of a (mixed) mock Jacobi form was introduced in \cite{Dabholkar:2012nd}. We start with the following definition. 
\begin{defn}\label{mockjacobidef}
Let $\phi:\mathbb{H}\times\mathbb{C}\longrightarrow\mathbb{C}$ be a holomorphic function with theta decomposition as in \eqref{eq 1.4}. The function $\phi$ is called a pure (resp. mixed) mock Jacobi form of weight $k$,  index $m$ and level $N$ if the theta coefficients $h_{\ell}$ are pure (resp. mixed) mock modular forms of weight $k-\frac{1}{2}$ for $\Gamma(4mN)$ and if $\widehat{h}_{\ell}$ are the completions of $h_{\ell}$ then the completion $\widehat{\phi}(\tau,z)$ of the mock Jacobi form $\phi$ defined by 
\begin{equation}
\widehat{\phi}(\tau,z):=\sum \limits_ { \ell \in \mathbb { Z } / 2 m \mathbb{Z} } \widehat{h} _ { \ell } ( \tau ) \vartheta _ { m , \ell } ( \tau , z )
\end{equation}
is invariant under the modular action for some congruence subgroup $\Gamma$ of level $N$ of $\mathrm{SL}(2,\mathbb{Z})$.
\end{defn} 
Using Zagier's weight 3/2 mock modular form \eqref{eq:Htmock3/2}, one can construct a mock Jacobi form of weight 2 and index 1 as follows \cite{Dabholkar:2012nd}. Define 
\begin{equation}
\begin{split}
    &\textbf{H}_0(\tau)=\sum_{n=0}^{\infty} H(4 n) q^n\\&\textbf{H}_1(\tau)=\sum_{n=1}^{\infty} H(4 n-1) q^{n-\frac{1}{4}},
\end{split}    
\end{equation}
where $H(0)=-\frac{1}{12}$.
Then 
\begin{equation}
\textbf{H}_0(\tau)+\textbf{H}_1(\tau)=\textbf{H}(\tau / 4).    
\end{equation}
Then the function
\begin{equation}\label{eq:mockjfhurwitz}
\textbf{H}(\tau, z)=\textbf{H}_0(\tau) \vartheta_{1,0}(\tau, z)+\textbf{H}_1(\tau) \vartheta_{1,1}(\tau, z)=\sum_{\substack{n, r \in \mathbb{Z} \\ 4 n-r^2 \geq 0}} H\left(4 n-r^2\right) q^n \zeta^r
\end{equation}
is a mock Jacobi form of weight 2 and index 1 with
shadow a constant multiple of
\begin{equation}
    \vartheta_{1,0}(\tau, 0) \overline{\vartheta_{1,0}(\tau, z)}+\vartheta_{1,1}(\tau, 0) \overline{\vartheta_{1,1}(\tau, z)}.
\end{equation}
Given that mock Jacobi forms are also defined in terms of theta series expansion, one might wonder if they are related to classical Jacobi forms. This relation was the central theme of \cite{Dabholkar:2012nd} which we now review in detail. \par
Let $\varphi(\tau,z)$ be a meromorphic Jacobi form of weight $k$ and index $m$ with respect to $\mathrm{SL}(2,\mathbb{Z})\ltimes\mathbb{Z}^2$. Assume that $\varphi$ has simple or double poles only at $z$ belonging to a discrete subset of $\mathbb{Q}\tau+\mathbb{Q}$. Let $S(\varphi)$ be the set of pairs $s=(\alpha,\beta)\in\mathbb{Q}^2$ such that $z=z_s=\alpha\tau+\beta$ is a pole of $\varphi.$ Since $\varphi$ is now meromorphic, it may not be possible to write down the theta decomposition of $\varphi$ as in \eqref{eq 1.4}. But we can still make sense of the \emph{finite part} of $\varphi$  by defining it as follows:
\begin{equation}
\varphi^F(\tau,z)=\sum \limits_ { \ell \in \mathbb { Z } / 2 m \mathbb{Z} } h _ { \ell } ( \tau ) \vartheta _ { m , \ell } ( \tau , z ),
\label{eq 2.1}
\end{equation}
where the theta coefficients are now given by  
\begin{equation}
h _ { \ell } ( \tau )=q^{-\ell^{2} / 4 m} \int\limits_{-\frac{\ell}{2m}\tau}^{-\frac{\ell}{2m}\tau+1} \varphi(\tau, z) \mathbf{e}(-\ell z) d z.
\label{eq 2.2}
\end{equation}
For the sum in \eqref{eq 2.1} to be well defined, $h_{\ell}$ must be defined uniquely $\bmod ~2m$. This is indeed true as explained in \cite[Section 8.1]{Dabholkar:2012nd}. The main result of \cite{Dabholkar:2012nd} is that the finite part $\varphi^F$ of a meromorphic Jacobi forms $\varphi$ with simple/double poles on a discrete subset of $\mathbb{Q}\tau+\mathbb{Q}$ is mixed mock Jacobi form. The essential idea of \cite{Dabholkar:2012nd} is the construction of the \textit{polar part} $\varphi^P$ of $\varphi$ using Appell-Lerch sums and then showing that $\varphi=\varphi^F+\varphi^P$. Generalising it to Jacobi forms for 
s of $\mathrm{SL}(2,\mathbb{Z})$ is straightforward. In fact, the result holds without any modification in the construction as we explain below.  
\\\\
Let $\varphi(\tau,z)$ be a meromorphic Jacobi form of weight $k$ and index $m$ with respect to $\Gamma\ltimes\mathbb{Z}^2$ with $\Gamma\in\{\Gamma_0(N),\Gamma_1(N)\}$ for some $N\geq 1$. Assume that $\varphi$ has simple or double poles only at $z$ belonging to a discrete subset of $\mathbb{Q}\tau+\mathbb{Q}$. Let $S(\varphi)$ be the set of pairs $s=(\alpha,\beta)\in\mathbb{Q}^2$ such that $z=z_s=\alpha\tau+\beta$ is a pole of $\varphi.$ 
The finite part $\varphi^F$ of $\varphi$ is defined as in \eqref{eq 2.1} and \eqref{eq 2.2} and is well defined as explained in \cite{Dabholkar:2012nd}. We now show that $\varphi^F$ is a mixed mock Jacobi form \textit{i,e.} we describe the completion of $\varphi^F$. First suppose $\varphi$ only has simple poles. For each $s\in S(\varphi)$, put 
\begin{equation}
D_{s}(\tau)=2 \pi i \mathbf{e}\left(m \alpha z_{s}\right) \operatorname{Res}\limits_{z=z_{s}}(\varphi(\tau, z)) \quad\left(s=(\alpha, \beta) \in S(\varphi), z_{s}=\alpha \tau+\beta\right).
\label{eq:1resn=1}
\end{equation}
\begin{prop}
Let $\varphi$ be a meromorphic Jacobi forms with simple poles at $z=z_s=\alpha\tau+\beta$ for $s=(\alpha,\beta)\in S(\varphi)\subset\mathbb{Q}^2$. Then the residues $D_s(\tau)$ defined in \eqref{eq:1resn=1} satisfies the following properties:
\begin{equation}\label{eq:Dstrans}
\begin{split}
& D_{(\alpha+\lambda, \beta+\mu)}(\tau)=\mathbf{e}(m(\mu \alpha-\lambda \beta)) D_{(\alpha, \beta)}(\tau),\quad(\lambda, \mu) \in \mathbb{Z}^{2}\\
& D_{s}\left(\frac{a \tau+b}{c \tau+d}\right)=(c \tau+d)^{k-1} D_{s \gamma}(\tau),\quad\gamma=\begin{pmatrix}
a& b \\ c& d
\end{pmatrix}\in \Gamma,
\end{split}
\end{equation}
where $s\gamma=(a\alpha+c\beta,b\alpha+d\beta).$
\label{prop:Dstrans}
\begin{proof}
We have 
\begin{equation}
\begin{split}
D_{(\alpha+\lambda, \beta+\mu)}(\tau)&=2\pi i\mathbf{e}(m(\alpha+\lambda)((\alpha+\lambda)\tau+(\beta+\mu))\operatorname{Res}\limits_{z=z_{s+(\lambda,\mu)}}(\varphi(\tau, z))\\&=2\pi i\mathbf{e}(m(\lambda^2\tau+2\lambda z_s)) \mathbf{e}\left(m \alpha z_{s}\right)\mathbf{e}(m(\mu \alpha-\lambda \beta))\operatorname{Res}\limits_{z=z_{s+(\lambda,\mu)}}(\varphi(\tau, z)).
\end{split}
\end{equation}
We can now write 
\begin{equation}
\begin{split}
\mathbf{e}(m(\lambda^2\tau+2\lambda z_s))\operatorname{Res}\limits_{z=z_{s+(\lambda,\mu)}}(\varphi(\tau, z))&=\mathbf{e}(m(\lambda^2\tau+2\lambda z_s))\operatorname{Res}_{z=z_{s}}(\varphi(\tau, z+\lambda\tau+\mu))\\&=\operatorname{Res}_{z=z_{s}}(\varphi(\tau, z))
\end{split}
\end{equation}
where we used the elliptic transformation of $\varphi$. This gives us the first identity after using \eqref{eq:1resn=1}. For the second identity, we write 
\begin{equation}
\mathbf{e}\left(m \alpha z_{s}\right)\varphi(\tau, z_s+\varepsilon)=\frac{D_s(\tau)}{2\pi i\varepsilon}+O(1),\quad\text{as}\quad \varepsilon\to 0.
\label{eq:Dsdefeps}
\end{equation}
Let us write $s=(\alpha_s,\beta_s)$ and $z_s=z_s(\tau)$ so that we have the following identities
\begin{equation}
z_s(\gamma \tau)=\frac{z_{s \gamma}(\tau)}{c \tau+d}, \quad \alpha_{s \gamma} z_{s \gamma}(\tau)-\alpha_s z_s(\gamma \tau)=\frac{c z_{s \gamma}(\tau)^2}{c \tau+d}; \quad\left(\gamma \tau:=\frac{a \tau+b}{c \tau+d}\right)
\label{eq:zssgamma}
\end{equation} 
Then replacing $\varepsilon\to \varepsilon/(c\tau+d)$ and $\tau\to\gamma\tau$ in \eqref{eq:Dsdefeps} we get 
\begin{equation}
\begin{split}
\frac{(c\tau+d)D_s(\gamma\tau)}{2\pi i\varepsilon}&\stackrel{.}{=}\mathbf{e}\left(m \alpha_s z_{s}(\gamma\tau)\right)\varphi\left(\gamma\tau, z_s(\gamma\tau)+\frac{\varepsilon}{c\tau+d}\right)
\\&\stackrel{.}{=}\mathbf{e}\left(m \alpha_s z_{s}(\gamma\tau)\right)\varphi\left(\gamma\tau, \frac{z_{s\gamma}(\tau)+\varepsilon}{c\tau+d}\right)\\&\stackrel{.}{=}(c\tau+d)^k\mathbf{e}\left(m \alpha_s z_{s}(\gamma\tau)+m\frac{c(z_{s\gamma}(\tau)+\varepsilon)^2}{c\tau+d}\right)\varphi\left(\tau, z_{s\gamma}(\tau)+\varepsilon\right)\\&\stackrel{.}{=}(c\tau+d)^k\mathbf{e}\left(m \alpha_{s\gamma} z_{s\gamma}(\tau)\right)\varphi\left(\tau, z_{s\gamma}(\tau)+\varepsilon\right)\\&\stackrel{.}{=}(c\tau+d)^k\frac{D_{s\gamma}(\tau)}{2\pi i\varepsilon}
\end{split}
\end{equation} 
where $\stackrel{.}{=}$ means that we have ignored terms which are bounded as $\varepsilon\to 0$ and we have used \eqref{eq:zssgamma} multiple times. This gives us the required result.
\end{proof}
\end{prop}
\noindent To construct the polar part of $\varphi$, the idea is to construct a function which has poles exactly at $s=(\alpha,\beta)\in S(\varphi)$ with residue 1. Then the polar part will just be the linear combination of such function weighted by $D_s(\tau)$. To do this we define the \emph{Appell-Lerch sums.} Following \cite{Dabholkar:2012nd}, we begin by defining a rational function $\mathcal{R}_{c}(\zeta)$ for $c\in\mathbb{R}$ as follows
\begin{equation}
\mathcal{R}_{c}(\zeta)=\begin{cases}
\frac{\zeta^{c}}{2}  \frac{\zeta+1}{\zeta-1} & \text { if } c \in \mathbb{Z} \\
\frac{\zeta^{\lceil c\rceil}}{\zeta-1} & \text { if } c \in \mathbb{R} \backslash \mathbb{Z}.
\end{cases}
\label{eq 2.6}
\end{equation}
Here $\lceil c\rceil$ denotes the smallest integer $\geq c$. The function $\mathcal{R}_{c}(\zeta)$ has a simple pole of residue 1 at $\zeta=1.$ We have the following two expressions \cite[Eq. (8.15), Proposition 8.1]{Dabholkar:2012nd} for $\mathcal{R}_{c}:$
\begin{equation}
\begin{split}
&\mathcal{R}_{c}(\mathbf{e}(z))=-\sum_{\ell \in \mathbb{Z}} \frac{\operatorname{sgn}(\ell-c)+\operatorname{sgn}\left(z_{2}\right)}{2} \mathbf{e}(\ell z), \quad\left( z_{2}=\operatorname{Im}(z) \neq 0\right),\\
&\mathcal{R}_{c}(\mathbf{e}(z))=\text{Av}_{\mathbb{Z}}\left[\frac{\mathbf{e}(c z)}{2 \pi i z}\right],\quad c \in \mathbb{R}, z \in \mathbb{C} \backslash \mathbb{Z}
\end{split}
\label{eq 2.7}
\end{equation}
where $\text{Av}_{\mathbb{Z}}[f(z)]$ is the $\mathbb{Z}$-average of $f$ defined as 
\begin{equation}
\text{Av}_{\mathbb{Z}}[f(z)]:=\sum_{\mu\in\mathbb{Z}}f(z+\mu).
\label{eq 2.5}
\end{equation}
These expressions will be important later.  
We now want to construct a function which has poles at all images of $z_s$ under the elliptic transformation $z\to z+\lambda\tau+\mu$ with $\lambda,\mu\in\mathbb{Z}$. To do this, define the averaging operator of index $m\in\mathbb{Z}$ as follows: 
\begin{equation}
\text{Av}_{\zeta}^{(m)}[F(\zeta)]:=\sum_{r\in\mathbb{Z}} q^{m r^{2}} \zeta^{2 m r} F\left(q^{r} \zeta\right).
\label{eq 2.4}
\end{equation}
Note that the theta functions can be written in terms of these averaging operator:
\begin{equation}
\vartheta_{m,\ell}(\tau,z)=q^{\frac{\ell^2}{4m}}\text{Av}^{(m)}_{\zeta}[\zeta^{\ell}].
\label{eq:thetaavgop}
\end{equation} 
Following \cite{Dabholkar:2012nd}, for $s=(\alpha,\beta)\in \mathbb{Q}^2$, we now define the \emph{universal Appell-Lerch sum} of \textit{index} $m$ as 
\begin{equation}
\begin{split}
\mathcal{A}_{m}^{s}(\tau, z)&=\mathbf{e}\left(-m \alpha z_{s}\right) \mathrm{Av}_{\zeta}^{(m)}\left[\mathcal{R}_{-2 m \alpha}\left(\zeta / \zeta_{s}\right)\right], \quad\zeta_{s}=\mathbf{e}\left(z_{s}\right)=\mathbf{e}(\beta) q^{\alpha}\\&=\mathbf{e}\left(-m \alpha z_{s}\right)\sum_{r\in\mathbb{Z}} q^{m r^{2}} \zeta^{2 m r} \mathcal{R}_{-2 m \alpha}\left(q^{r}\zeta / \zeta_{s}\right).
\end{split}
\label{eq 2.8}
\end{equation}
If we change $s=(\alpha,\beta)$ by an element of $\mathbb{Z}^2$ then we have 
\begin{equation}
\begin{aligned}
\mathcal{A}_m^{(\alpha+\lambda, \beta+\mu)}(\tau, z)&=\mathbf{e}\left(-m(\alpha+\lambda)((\alpha+\lambda) \tau+\beta+\mu)\right) \\
&\times \sum_{r \in \mathbb{Z}} q^{m r^2} \zeta^{2 m r} \mathcal{R}_{-2 m(\alpha+\lambda)}\left(\frac{q^r \zeta}{\zeta_{s+(\lambda, \mu)}}\right) \\
&=\mathbf{e}\left(-m\left(\alpha z_s+2 \lambda z_s-\lambda \beta+\alpha \mu+\lambda^2 \tau\right)\right) \\&\times\sum_{r \in \mathbb{Z}} q^{m r^2} \zeta^{2 m r}\left(\frac{q^r \zeta}{q^\lambda \zeta_s}\right)^{-2 m \lambda} \mathcal{R}_{-2 m\alpha}\left(\frac{q^r \zeta}{q^\lambda \zeta_s}\right) \\
&=\zeta_s^{-2 m \lambda} q^{-m \lambda^2} \mathbf{e}(-m(\alpha z_s+\alpha \mu-\lambda \beta)) \\
&\times \sum_{r \in \mathbb{Z}} q^{m r^2} \zeta^{2 m(r-\lambda)} \zeta_s^{2 m \lambda} q^{-2 m \lambda r} q^{2 m \lambda^2} \mathcal{R}_{-2 m \alpha}\left(\frac{q^{r-\lambda}\zeta}{\zeta_s}\right) \\
&=\mathbf{e}(-m(\alpha z_s+\alpha \mu-\lambda \beta)) \sum_{r \in \mathbb{Z}} q^{m(r-\lambda)^2} \zeta^{2 m(r-\lambda)} \mathcal{R}_{-2 m \alpha}\left(\frac{q^{r-\lambda}\zeta}{\zeta_s}\right) \\
&=\mathbf{e}(-m(\alpha \mu-\lambda \beta)) \mathcal{A}_m^{(\alpha, \beta)}(\tau, z),
\end{aligned}
\label{eq:s+lmAm}
\end{equation}
where we used the periodicity: $\mathcal{R}_{c+n}(\zeta)=\zeta^n\mathcal{R}_c(\zeta)$ for $n\in\mathbb{Z}$.
We also have elliptic invariance of the Appell-Lerch sum:
\begin{equation}\label{eq:ellipinvAsm}
\begin{split}
\mathcal{A}_{m}^{s}(\tau, z+\lambda\tau+\mu)&=\mathbf{e}\left(-m \alpha z_{s}\right)\sum_{r\in\mathbb{Z}} q^{m r^{2}} \zeta^{2 m r}\zeta^{2mr\lambda} \mathcal{R}_{-2 m \alpha}\left(q^{r}\zeta q^\lambda / \zeta_{s}\right)\\&=q^{-m\lambda^2}\zeta^{-2m\lambda}\mathbf{e}\left(-m \alpha z_{s}\right)\sum_{r\in\mathbb{Z}} q^{m (r^{2}+\lambda^2+2r\lambda)} \zeta^{2 m (r+\lambda)} \mathcal{R}_{-2 m \alpha}\left(q^{r}\zeta q^\lambda / \zeta_{s}\right)\\&=\mathbf{e}(-m(\lambda^2\tau+2\lambda z))\mathbf{e}\left(-m \alpha z_{s}\right)\sum_{r\in\mathbb{Z}} q^{m (r+\lambda)^2} \zeta^{2 m (r+\lambda)} \mathcal{R}_{-2 m \alpha}\left(q^{r+\lambda}\zeta / \zeta_{s}\right)\\&=\mathbf{e}(-m(\lambda^2\tau+2\lambda z))\mathcal{A}_{m}^{s}(\tau, z).    
\end{split}    
\end{equation}
\par The \emph{polar part} of $\varphi$ is defined by 
\begin{equation}
\varphi^{P}(\tau, z):=\sum_{s \in S(\varphi) / \mathbb{Z}^{2}} D_{s}(\tau) \mathcal{A}_{m}^{s}(\tau, z).
\end{equation}
Here $S(\varphi)/\mathbb{Z}^2$ is the set of orbits under the action $((\lambda,\mu),s)\mapsto s+(\lambda,\mu)$. This sum is a finite sum and well defined due to \eqref{eq:Dstrans} and \eqref{eq:s+lmAm}. Since the residues and pole structure of $\varphi$ and $\varphi^P$ is the same, so $\varphi-\varphi^P$ is holomorphic. In fact, one can show that (see Theorem \ref{thm 2.12} for proof in a more general case)
\begin{equation}
\varphi^F(\tau,z)=\varphi(\tau,z)-\varphi^P(\tau,z).
\label{eq 2.9}
\end{equation}
We now describe the completion of $\mathcal{A}^s_m(\tau,z)$ to get the completion of $\varphi^F$. For $m \in \mathbb{N},$ and $\ell \in \mathbb{Z} / 2 m \mathbb{Z}$ with $s=(\alpha, \beta) \in \mathbb{Q}^{2}$ we define the unary theta series
\begin{equation}
\Theta_{m, \ell}^{s}(\tau)=\mathbf{e}(-m \alpha \beta) \sum_{r \in \mathbb{Z}+\alpha+\ell / 2 m} r ~\mathbf{e}(2 m \beta r) q^{m r^{2}}
\label{eq 2.10}
\end{equation}
of weight $3 / 2$ and its preimage under the shadow operator given by the nonholomorphic Eichler integral \eqref{eq:eichintsum}
\begin{equation}
\Theta_{m, \ell}^{s *}(\tau)=\frac{\mathbf{e}(m \alpha \beta)}{2\sqrt{\pi}} \sum_{r \in \mathbb{Z}+\alpha+\ell / 2 m} \operatorname{sgn}(r) \mathbf{e}(-2 m \beta r)\Gamma\left(\frac{1}{2},4r^2\pi my\right) q^{-m r^{2}},
\label{eq 2.11}
\end{equation}
where, as before, $y=\text{Im}(\tau).$ We have removed a factor of $-1/\sqrt{m}$ from the Eichler integral for convenience. 
One can check that 
\begin{equation}
 \xi_{\frac{1}{2}}(\Theta_{m, \ell}^{s *}(\tau))=-\sqrt{m}\Theta_{m, \ell}^{s}(\tau)~,  
\end{equation}
where $\xi_\kappa$ has been defined in (\ref{e218x}).
Observe that, for $(\lambda,\mu)\in\mathbb{Z}^2$,
\begin{equation}
\begin{aligned}\Theta^{(\alpha+\lambda, \beta+\mu)}_{m,\ell}(\tau)&=\mathbf{e}(-m(\alpha+\lambda)(\beta+\mu))\sum_{r \in \mathbb{Z}+(\alpha+\lambda)+\frac{\ell}{2 m}}r~\mathbf{e}\left(2 m (\beta+\mu)r\right)q^{m r^2}
\\&=\mathbf{e}(-m(\alpha \mu+\beta \lambda)) \mathbf{e}(-m \alpha \beta) \mathbf{e}\left(2 m \mu\left(\alpha+\frac{\ell}{2 m}\right)\right)\sum_{r \in \mathbb{Z}+\alpha+\frac{\ell}{2 m}}r~ \mathbf{e}(2 m \beta r) q^{m r^2} \\
&=\mathbf{e}(m(\alpha \mu-\beta \lambda)) \Theta_{m, \ell}^{(\alpha, \beta)}(\tau) . 
\end{aligned}
\label{eq:s+lmTheta}
\end{equation}
Similarly we have 
\begin{equation}
\Theta_{m, \ell}^{(\alpha+\lambda, \beta+\mu) *}(\tau)=\mathbf{e}(-m(\alpha \mu-\beta \lambda)) \Theta_{m, \ell}^{(\alpha, \beta)*}(\tau) .
\label{eq:s+lmTheta*}
\end{equation}
The following result is from DMZ {\cite[Proposition 8.2]{Dabholkar:2012nd}}.
\begin{prop}
For $m \in \mathbb{N}$ and $s \in \mathbb{Q}^{2}$ the completion $\widehat{\mathcal{A}}_{m}^{s}$ of $\mathcal{A}_{m}^{s}$ defined by
$$
\widehat{\mathcal{A}}_{m}^{s}(\tau, z):=\mathcal{A}_{m}^{s}(\tau, z)+\sum_{\ell(\bmod 2 m)} \Theta_{m, \ell}^{s *}(\tau) \vartheta_{m, \ell}(\tau, z)
$$
satisfies
$$
\begin{aligned}
\widehat{\mathcal{A}}_{m}^{(\alpha+\lambda, \beta+\mu)}(\tau, z) &=\mathbf{e}(-m(\mu \alpha-\lambda \beta)) \widehat{\mathcal{A}}_{m}^{(\alpha, \beta)}(\tau,z) &\lambda, \mu \in \mathbb{Z} \\
\widehat{\mathcal{A}}_{m}^{s}(\tau, z+\lambda \tau+\mu) &=\mathbf{e}\left(-m\left(\lambda^{2} \tau+2 \lambda z\right)\right) \widehat{\mathcal{A}}_{m}^{s}(\tau,z) &\lambda, \mu \in \mathbb{Z} \\
\widehat{\mathcal{A}}_{m}^{s}\left(\frac{a \tau+b}{c \tau+d}, \frac{z}{c \tau+d}\right) &=(c \tau+d) \mathbf{e}\left(\frac{m c z^{2}}{c \tau+d}\right) \widehat{\mathcal{A}}_{m}^{s \gamma}(\tau, z) &\gamma=\begin{pmatrix}
a& b \\
c &d
\end{pmatrix}\in \mathrm{SL}(2, \mathbb{Z}).
\end{aligned}
$$
\label{prop 2.2}
\begin{proof}
The first two transformation properties follow from \eqref{eq:s+lmAm}, \eqref{eq:s+lmTheta}, \eqref{eq:s+lmTheta*} and \eqref{eq:ellipinvAsm}. For the modularity transformation, we start by defining two functions: for $m\in\mathbb{N},\tau\in\mathbb{H},z,u\in\mathbb{C}$ define
\begin{equation}\label{eq:fu(m)def}
f_u^{(m)}(z ; \tau)=\operatorname{Av}_{\zeta}^{(m)}\left[\frac{1}{1-\zeta / \mathbf{e}(u)}\right], \quad \tilde{f}_u^{(m)}(z ; \tau)=f_u^{(m)}(z ; \tau)-\frac{1}{2} \sum_{\ell(\bmod 2 m)} R_{m, \ell}(u ; \tau) \vartheta_{m, \ell}(\tau, z)
\end{equation}
where
\begin{equation}
R_{m, \ell}(u ; \tau)=\sum_{r \in \ell+2 m \mathbb{Z}}\left\{\operatorname{sgn}\left(r+\frac{1}{2}\right)-\text{erf}\left(\sqrt{\pi} \frac{ry+2 m\text{Im}(u)}{\sqrt{my}}\right)\right\} q^{-r^2 / 4 m} \mathbf{e}(-r u),
\end{equation}
where erf($x$) is the error function defined as 
\begin{equation}
\text{erf}(x)=\frac{2}{\sqrt{\pi}}\int_0^xe^{-t^2}dt.
\end{equation} 
It is related to the incomplete gamma function as 
\begin{equation}
\text{erf}(x)=\frac{\text{sgn}(x)}{\sqrt{\pi}}\left(\sqrt{\pi}-\Gamma\left(\frac{1}{2},x^2\right)\right).
\label{eq:relerfincgamma}
\end{equation}
We then have the following modular transformation property \cite[Proposition 3.5]{zwegers2008mock} for $\tilde{f}_u^{(m)}$:
\begin{equation}
\tilde{f}_{u /(c \tau+d)}^{(m)}\left(\frac{z}{c \tau+d} ; \frac{a \tau+b}{c \tau+d}\right)=(c \tau+d) \mathbf{e}\left(\frac{m c\left(z^2-u^2\right)}{c \tau+d}\right) \tilde{f}_u^{(m)}(z ; \tau),\quad\begin{pmatrix}
a & b \\ c & d
\end{pmatrix} \in\mathrm{SL}(2,\mathbb{Z}).
\label{eq:modtranscurlfm}
\end{equation}
Using \eqref{eq:relerfincgamma} and replacing $u=z_s=\alpha\tau+\beta$ we can write 
\begin{equation}
\begin{split}
\frac{1}{2} R_{m, \ell}\left(z_s ; \tau\right)&=\sum_{r\in \ell+2 m \mathbb{Z}} \frac{\operatorname{sgn}\left(r+\frac{1}{2}\right)-\operatorname{sgn}(r+2 m \alpha)}{2} q^{-r^2 / 4 m} \zeta_s^{-r}\\&+\frac{1}{2\sqrt{\pi}}\sum_{r\in \ell+2 m \mathbb{Z}} \operatorname{sgn}(r+2 m \alpha)\Gamma\left(\frac{1}{2},\frac{(ry+2m\alpha y)^2\pi}{my}\right)q^{-r^2 / 4 m} \zeta_s^{-r}\\&=\sum_{r\equiv\ell(\bmod 2 m)} \frac{\operatorname{sgn}\left(r+\frac{1}{2}\right)-\operatorname{sgn}(r+2 m \alpha)}{2} q^{-r^2 / 4 m} \zeta_s^{-r}\\&+\frac{1}{2\sqrt{\pi}}\sum_{r\in \ell/2m+\mathbb{Z}+\alpha} \operatorname{sgn}(r)\Gamma\left(\frac{1}{2},4\pi mr^2y\right)q^{-m(r-\alpha)^2} \zeta_s^{-2m(r-\alpha)}\\&=\sum_{r\equiv\ell(\bmod 2 m)} \frac{\operatorname{sgn}\left(r+\frac{1}{2}\right)-\operatorname{sgn}(r+2 m \alpha)}{2} q^{-r^2 / 4 m} \zeta_s^{-r}+\mathbf{e}\left(m \alpha z_s\right) \Theta_{m, \ell}^{s *}(\tau)
\end{split}
\label{eq:RThetarel}
\end{equation}
where in the last step we used \eqref{eq 2.11} and the identity
\begin{equation}
q^{-m(r-\alpha)^2}\zeta_s^{-2m(r-\alpha)}=q^{-mr^2}\mathbf{e}(-m\alpha\beta)\mathbf{e}(\alpha m(\alpha\tau+\beta)).
\end{equation}
Next, note that 
\begin{equation}
\begin{split}
\sum_{r\in\mathbb{Z}}\frac{\operatorname{sgn}\left(r+\frac{1}{2}\right)-\operatorname{sgn}(r+2 m \alpha)}{2}\zeta^r&=\sum_{r\in\mathbb{Z}}\frac{\operatorname{sgn}\left(r+\frac{1}{2}\right)+\text{sgn}(\text{Im}(z))-\text{sgn}(\text{Im}(z))-\operatorname{sgn}(r+2 m \alpha)}{2}\zeta^r\\&=-\frac{1}{\zeta-1}-\sum_{r\in\mathbb{Z}}\frac{\text{sgn}(\text{Im}(z))+\operatorname{sgn}(r+2 m \alpha)}{2}\zeta^r\\&=-\frac{1}{\zeta-1}+\mathcal{R}_{-2m\alpha}(\zeta)
\end{split}
\label{eq:curRfrel}
\end{equation}
where we used \eqref{eq 2.6}. Replacing $\zeta$ by $\zeta/\zeta_s$ in the above relation \eqref{eq:curRfrel} and applying the averaging operator $\text{Av}^{(m)}_{\zeta}$ we get 
\begin{equation}
\begin{split}
\mathbf{e}(m\alpha z_s)\mathcal{A}^s_{m}(\tau,z)&=-\text{Av}^{(m)}_{\zeta}\left[\frac{1}{1-\zeta/\zeta_s}\right]+\text{Av}^{(m)}_\zeta\left[\sum_{r\in\mathbb{Z}}\frac{\operatorname{sgn}\left(r+\frac{1}{2}\right)-\operatorname{sgn}(r+2 m \alpha)}{2}\zeta^r\zeta_s^{-r}\right]\\&=-f^{(m)}_{z_s}(z;\tau)+\sum_{r\in\mathbb{Z}}\frac{\operatorname{sgn}\left(r+\frac{1}{2}\right)-\operatorname{sgn}(r+2 m \alpha)}{2}q^{-r^2/4m}\zeta_s^{-r}\vartheta_{m,r}(\tau,z)\\&=-f^{(m)}_{z_s}(z;\tau)+\sum_{\substack{\ell(\bmod 2m)\\r\equiv\ell (\bmod 2m)}}\frac{\operatorname{sgn}\left(r+\frac{1}{2}\right)-\operatorname{sgn}(r+2 m \alpha)}{2}q^{-r^2/4m}\zeta_s^{-r}\vartheta_{m,\ell}(\tau,z)
\end{split}
\label{eq:curlAfrel}
\end{equation} 
where we used \eqref{eq:thetaavgop}. Using \eqref{eq:curlAfrel}, we obtain 
\begin{equation}\label{eq:eAhatsmftilde}
\begin{split}
\mathbf{e}(m\alpha z_s)\widehat{\mathcal{A}}^s_{m}(\tau,z)&=-f^{(m)}_{z_s}(z;\tau)+\sum_{\substack{\ell(\bmod 2m)\\r\equiv\ell (\bmod 2m)}}\frac{\operatorname{sgn}\left(r+\frac{1}{2}\right)-\operatorname{sgn}(r+2 m \alpha)}{2}q^{-r^2/4m}\zeta_s^{-r}\vartheta_{m,\ell}(\tau,z)\\&+\mathbf{e}(m\alpha z_s)\sum_{\ell(\bmod 2m)}\Theta_{m, \ell}^{s *}(\tau)\vartheta_{m,\ell}(\tau,z)\\&=-\tilde{f}^{(m)}_{z_s}(z;\tau)
\end{split}
\end{equation}
where we used \eqref{eq:fu(m)def} and \eqref{eq:RThetarel}. Using  \eqref{eq:zssgamma} and \eqref{eq:modtranscurlfm}, we get 
\begin{equation}\label{eq:Ahatsmmodtrans}
\begin{split}
\widehat{\mathcal{A}}^s_{m}\left(\frac{a\tau+b}{c\tau+d},\frac{cz}{c\tau+d}\right)&=-\mathbf{e}\left(-m\alpha_sz_{s}(\gamma\tau)\right)\tilde{f}^{(m)}_{z_s(\gamma\tau)}\left(\frac{cz}{c\tau+d};\frac{a\tau+b}{c\tau+d}\right)\\&=-\mathbf{e}\left(-m\left(\alpha_{s\gamma}z_{s\gamma}(\tau)-\frac{cz_{s\gamma}(\tau)^2}{c\tau+d}\right)\right)\tilde{f}^{(m)}_{\frac{z_{s\gamma}(\tau)}{(c\tau+d)}}\left(\frac{cz}{c\tau+d};\frac{a\tau+b}{c\tau+d}\right)\\&=\mathbf{e}\left(m\frac{cz_{s\gamma}(\tau)^2}{c\tau+d}\right)(c\tau+d)\mathbf{e}\left(\frac{m c\left(z^2-z_{s\gamma}^2\right)}{c \tau+d}\right)\widehat{\mathcal{A}}^{s\gamma}_{m}(\tau,z)\\&=(c \tau+d) \mathbf{e}\left(\frac{m c z^{2}}{c \tau+d}\right) \widehat{\mathcal{A}}_{m}^{s \gamma}(\tau, z).
\end{split}
\end{equation}
\end{proof}
\end{prop}
\noindent This proposition implies that the Appell-Lerch sum $\mathcal{A}^s_m(\tau,z)$ is a meromorphic mock Jacobi form of weight $1$ with shadow $\sum_{\ell\bmod~2m}\overline{\Theta^s_{m,\ell}(\tau)}\vartheta_{m,\ell}(\tau,z)$ and completion $\widehat{\mathcal{A}}^s_m(\tau,z)$ for some congruence subgroup of $\mathrm{SL}(2,\mathbb{Z})$ containing elements $\gamma$ such that  $s\gamma=s$ upto shift by integers.\par
The mock modularity of $\varphi^F$ is described in the theorem below (see Theorem \ref{thm 2.9} for proof in a more general case)
\begin{thm}
\label{mockjacobi_congruence}
Let $\varphi$ be a meromorphic Jacobi form of weight $k$, index $m$ with respect to $\Gamma$ and  having simple poles on the set $S(\varphi)\subset\mathbb{Q}^2$ as above. Let $h_{\ell}$ and $\varphi^{F}$ be as in \eqref{eq 2.2} and \eqref{eq 2.1} respectively. Then each $h_{\ell}$ is a mixed mock modular form of mixed weight $\left(k-1,\frac{1}{2}\right)$ with shadow $-\sqrt{m}\sum_{s \in S(\varphi) / \mathbb{Z}^{2}} D_{s}(\tau) \overline{\Theta_{m, \ell}^{s}(\tau)},$ and the function $\varphi^{F}$ is a mixed mock Jacobi form. More precisely, for each $\ell \in \mathbb{Z} / 2 m \mathbb{Z}$, the completion of $h_{\ell}$, defined by
$$
\widehat{h}_{\ell}(\tau):=h_{\ell}(\tau)-\sum_{s \in S(\varphi) / \mathbb{Z}^{2}} D_{s}(\tau) \Theta_{m, \ell}^{s *}(\tau)
$$
with $\Theta_{m, \ell}^{s *}$ as in \eqref{eq 2.11}, transforms like a modular form of weight $k-\frac{1}{2}$ with respect to $\Gamma(4mN)$  and the completion of $\varphi^{F}$ defined by
$$
\widehat{\varphi}^{F}(\tau, z):=\sum_{\ell(\bmod 2 m)} \widehat{h}_{\ell}(\tau) \vartheta_{m, \ell}(\tau, z)
$$
transforms like a Jacobi form of weight $k$ and index $m$ with respect to $\Gamma\ltimes\mathbb{Z}^2$.
\label{thm 2.3}
\end{thm}
Now suppose $\varphi$ has double poles on $S(\varphi)$. We need to define second order residues as well as Appell-Lerch sum, To this end, for $s\in S(\varphi)$ define the residues $D_s(\tau)$ and $E_s(\tau)$ by
\begin{equation}
\mathbf{e}(m \alpha z_{s}) \varphi\left(\tau, z_{s}+\varepsilon\right)=\frac{E_{s}(\tau)}{(2 \pi i \varepsilon)^{2}}+\frac{D_{s}(\tau)-2 m \alpha E_{s}(\tau)}{2 \pi i \varepsilon}+O(1) ~~\text{as}~~\varepsilon\to 0.
\label{eq 2.12}
\end{equation}
Similar arguments as in the simple pole case show that Proposition \ref{prop:Dstrans} is still true and the function $E_s$ also satisfies
\begin{equation}
E_{(\alpha+\lambda, \beta+\mu)}(\tau)=\mathbf{e}(m(\mu \alpha-\lambda \beta)) E_{(\alpha, \beta)}(\tau),\quad(\lambda, \mu) \in \mathbb{Z}^{2}
\label{eq:Ess+lm}
\end{equation}
The modular transformation of the residues is given in Proposition \ref{prop:DsEstrans} below.
\begin{prop}
The functions $E_{s}(\tau)$ and $D_{s}(\tau)$ defined by \eqref{eq 2.12} satisfy 
$$
E_{s}\left(\frac{a \tau+b}{c \tau+d}\right)=(c \tau+d)^{k-2} E_{s \gamma}(\tau), \quad D_{s}\left(\frac{a \tau+b}{c \tau+d}\right)=(c \tau+d)^{k-1} D_{s \gamma}(\tau),\quad \begin{pmatrix}a& b \\ c &d\end{pmatrix} \in \Gamma.
$$
\label{prop:DsEstrans}
\begin{proof}
The proof is similar to Proposition \ref{prop:Dstrans} once we rewrite \eqref{eq 2.12} as 
\begin{equation}
\mathbf{e}(m \alpha (z_{s}+2\varepsilon)) \varphi\left(\tau, z_{s}+\varepsilon\right)=\frac{E_{s}(\tau)}{(2 \pi i \varepsilon)^{2}}+\frac{D_{s}(\tau)}{2 \pi i \varepsilon}+O(1).
\end{equation}
Indeed, using similar calculations as in proof of Proposition \ref{prop:Dstrans}, one can show that 
\begin{equation}
(c\tau+d)^2\frac{E_{s}(\gamma\tau)}{(2 \pi i \varepsilon)^{2}}+(c\tau+d)\frac{D_{s}(\gamma\tau)}{2 \pi i \varepsilon}\stackrel{.}{=}(c\tau+d)^k\left[\frac{E_{s\gamma}(\tau)}{(2 \pi i \varepsilon)^{2}}+\frac{D_{s\gamma}(\tau)}{2 \pi i \varepsilon}\right],
\end{equation}
from which the result follows.
\end{proof}
\end{prop} 
To define a second-order Appell-Lerch sum, we first need a second order rational function with a double pole at $z_s$. We make use of the second expression in \eqref{eq 2.7} to define such a function. For $m\in\mathbb{N}$ and $c\in\mathbb{R}$, put
\begin{equation}
\mathcal{R}_{c}^{(2)}(\mathbf{e}(z))=\operatorname{Av}_{\mathbb{Z}}\left[\frac{\mathbf{e}(c z)}{(2 \pi i z)^{2}}\right]=\frac{1}{(2 \pi i)^{2}} \sum_{n \in \mathbb{Z}} \frac{\mathbf{e}(c(z-n))}{(z-n)^{2}}.
\label{eq 2.13}
\end{equation}
Note that 
\begin{equation}\label{eq:Rcd/dz}
\mathcal{R}_{c}^{(2)}(\mathbf{e}(z))=-\frac{1}{2\pi i}\frac{\partial}{\partial z}\mathcal{R}_{c}(\mathbf{e}(z))+c\mathcal{R}_{c}(\mathbf{e}(z)).
\end{equation}
Using the first expression in \eqref{eq 2.7} and \eqref{eq:Rcd/dz}, we get 
\begin{equation}
\mathcal{R}_{c}^{(2)}(\zeta)=\sum_{\ell \in \mathbb{Z}} \frac{|\ell-c|+\operatorname{sgn}\left(\text{Im}(z)\right)(\ell-c)}{2} \zeta^{\ell}.
\end{equation}  
With this, for $s\in\mathbb{Q}^2$, define the second order Appell-Lerch sum 
\begin{equation}
\begin{split}
\mathcal{B}_{m}^{s}(\tau, z)&=\mathbf{e}\left(-m \alpha z_{s}\right) \operatorname{Av}_{\zeta}^{(m)}\left[\mathcal{R}_{-2 m \alpha}^{(2)}\left(\zeta / \zeta_{s}\right)\right],
\\&=\mathbf{e}\left(-m \alpha z_{s}\right)\sum_{\lambda\in\mathbb{Z}} q^{m \lambda^{2}} \zeta^{2 m \lambda} \mathcal{R}^{(2)}_{-2 m \alpha}\left(q^{\lambda}\zeta / \zeta_{s}\right)
\end{split}
\label{eq 2.14}
\end{equation} 
and define the polar part by 
\begin{equation}
\varphi^{P}(\tau, z)=\sum_{s \in S(\varphi) / \mathbb{Z}^{2}}\left(D_{s}(\tau) \mathcal{A}_{m}^{s}(\tau, z)+E_{s}(\tau) \mathcal{B}_{m}^{s}(\tau, z)\right).
\label{eq 2.15}
\end{equation}
As before, we can prove that \cite{Dabholkar:2012nd}
\begin{equation}
\varphi^F=\varphi-\varphi^P.
\end{equation}
We prove this in a more general setting in Section \ref{rational_index}.
In this case as well, $\varphi^F$ turns out to be a mixed mock Jacobi form but of more complicated type as we now describe. Define the unary theta series 
\begin{equation}
\Xi_{m, \ell}^{s}(\tau)=\mathbf{e}(-m \alpha \beta) \sum_{\lambda \in \mathbb{Z}+\alpha+\ell / 2 m} \mathbf{e}(2 m \beta \lambda) q^{m \lambda^{2}}
\label{eq 2.16}
\end{equation}
of weight $1/2$ and its preimage under the shadow operator again given by the nonholomorphic Eichler integral \eqref{eq:eichintsum}
\begin{equation}
\Xi_{m, \ell}^{s *}(\tau)=\frac{\mathbf{e}(m \alpha \beta)}{2\sqrt{\pi}} \sum_{\lambda \in \mathbb{Z}+\alpha+\ell / 2 m}|\lambda| \mathbf{e}(-2 m \beta \lambda) \Gamma\left(-\frac{1}{2},4\lambda^2\pi my\right) q^{-m \lambda^{2}},
\label{eq 2.17}
\end{equation}
where $y=\text{Im}(\tau)$. 
Again, we have removed a constant factor from the Eichler integral for convenience. One can check that 
\begin{equation}
 \xi_{\frac{3}{2}}(\Xi_{m, \ell}^{s*}(\tau))=-\frac{1}{4\pi \sqrt{m}}\Xi_{m, \ell}^{s}(\tau)~.   
\end{equation}
The completion of $\mathcal{B}_m^s(\tau,z)$ is then given by 
\begin{equation}
\widehat{\mathcal{B}}_{m}^{s}(\tau, z):=\mathcal{B}_{m}^{s}(\tau, z)+m \sum_{\ell(\bmod 2 m)} \Xi_{m, \ell}^{s *}(\tau) \vartheta_{m, \ell}(\tau, z).
\label{eq 2.18}
\end{equation}
\begin{prop}
For $m \in \mathbb{N}$ and $s \in \mathbb{Q}^{2}$ the completion $\widehat{\mathcal{B}}_{m}^{s}$ of $\mathcal{B}_{m}^{s}$ defined by \eqref{eq 2.18} satisfies
\begin{equation}
\begin{split}
&\widehat{\mathcal{B}}_{m}^{(\alpha+\lambda, \beta+\mu)}(\tau, z)=\mathbf{e}(-m(\mu \alpha-\lambda \beta+\lambda \mu)) \widehat{\mathcal{B}}_{m}^{(\alpha, \beta)}(\tau,z),\quad\lambda, \mu \in \mathbb{Z}\\
&\widehat{\mathcal{B}}_{m}^{s}(\tau, z+\lambda \tau+\mu)=\mathbf{e}\left(-m\left(\lambda^{2} \tau+2 \lambda z\right)\right) \widehat{\mathcal{B}}_{m}^{s}(\tau,z) \quad\lambda, \mu \in \mathbb{Z} \\
&\widehat{\mathcal{B}}_{m}^{s}\left(\frac{a \tau+b}{c \tau+d}, \frac{z}{c \tau+d}\right)=(c \tau+d)^{2} \mathbf{e}\left(\frac{m c z^{2}}{c \tau+d}\right) \widehat{\mathcal{B}}_{m}^{s \gamma}(\tau, z) \quad\gamma=\begin{pmatrix}
a &b \\
c &d
\end{pmatrix}
\in \mathrm{S L}(2, \mathbb{Z})
\end{split}
\end{equation}
\label{prop 2.5}
\begin{proof}
The proof of first two properties is analogous to the proof of Proposition \ref{prop 2.2}. For the modular transformation, we define the functions
\begin{equation}
\begin{split}
&g^{(m)}_u(z;\tau):=\left(\frac{1}{2\pi i}\frac{\partial}{\partial u}-\frac{2m\text{Im}(u)}{y}\right)f^{(m)}_u(z;\tau),\\&\tilde{g}^{(m)}_u(z;\tau):=\left(\frac{1}{2\pi i}\frac{\partial}{\partial u}-\frac{2m\text{Im}(u)}{y}\right)\tilde{f}^{(m)}_u(z;\tau)
\end{split}
\end{equation}
where as usual $y=\text{Im}(\tau)$. Let us rewrite 
\begin{equation}
\hat{f}^{(m)}(u,z;\tau)\equiv \tilde{f}^{(m)}_{u}(z;\tau).
\end{equation}
Then the modular transformation \eqref{eq:modtranscurlfm} can be written as 
\begin{equation}
\hat{f}^{(m)}\left(\frac{u}{c \tau+d},\frac{z}{c \tau+d} ; \frac{a \tau+b}{c \tau+d}\right)=(c \tau+d)~ \mathbf{e}\left(\frac{m c\left(z^2-u^2\right)}{c \tau+d}\right) \hat{f}^{(m)}(u,z ; \tau)
\label{eq:modtranshatfm}
\end{equation} 
for every $\left(\begin{smallmatrix}
a & b \\ c & d
\end{smallmatrix}\right) \in\mathrm{SL}(2,\mathbb{Z})$.
Then using chain rule we have 
\begin{equation}
\begin{split}
\frac{1}{2\pi i}\frac{\partial}{\partial u}\hat{f}^{(m)}\left(\frac{u}{c\tau+d},\frac{z}{c\tau+d};\frac{a\tau+b}{c\tau+d}\right)
=\frac{1}{2\pi i}\frac{1}{c\tau+d}\frac{\partial\hat{f}^{(m)}}{\partial u}\left(\frac{u}{c\tau+d},\frac{z}{c\tau+d};\frac{a\tau+b}{c\tau+d}\right),
\end{split}
\label{eq:partufhat}
\end{equation}
where on the right hand side $\frac{\partial \hat{f}^{(m)}}{\partial u}$ denotes the derivative of $\hat{f}^{(m)}$ with respect to the first argument. 
One the other hand, using \eqref{eq:modtranshatfm}, the left hand side of \eqref{eq:partufhat} becomes 
\begin{equation}
-2mcu~\mathbf{e}\left(\frac{m c\left(z^2-u^2\right)}{c \tau+d}\right) \hat{f}^{(m)}(u,z ; \tau)+(c\tau+d)~\mathbf{e}\left(\frac{m c\left(z^2-u^2\right)}{c \tau+d}\right)\frac{1}{2\pi i} \frac{\partial\hat{f}^{(m)}}{\partial u}(u,z ; \tau).
\end{equation}
We thus have 
\begin{equation}
\begin{split}
&\frac{1}{2\pi i}\frac{1}{c\tau+d}\frac{\partial\hat{f}^{(m)}}{\partial u}\left(\frac{u}{c\tau+d},\frac{z}{c\tau+d};\frac{a\tau+b}{c\tau+d}\right)\\&=-2mcu~\mathbf{e}\left(\frac{m c\left(z^2-u^2\right)}{c \tau+d}\right) \hat{f}^{(m)}(u,z ; \tau)+(c\tau+d)~\mathbf{e}\left(\frac{m c\left(z^2-u^2\right)}{c \tau+d}\right)\frac{1}{2\pi i} \frac{\partial\hat{f}^{(m)}}{\partial u}(u,z ; \tau).    
\end{split}    
\end{equation}
Under $u\to u/(c\tau+d),\tau\to \gamma\tau$, the second term in the definition of $\tilde{g}^{(m)}_u$ changes as 
\begin{equation}
\begin{split}
-2m\frac{\text{Im}(u/(c\tau+d))}{\text{Im}(\gamma\tau)}&\hat{f}^{(m)}\left(\frac{u}{c \tau+d},\frac{z}{c \tau+d} ; \frac{a \tau+b}{c \tau+d}\right)\\&=-2m\frac{\text{Im}(u(c\bar{\tau}+d))}{y}(c \tau+d)~ \mathbf{e}\left(\frac{m c\left(z^2-u^2\right)}{c \tau+d}\right) \hat{f}^{(m)}(u,z ; \tau),
\end{split}
\end{equation}
where we used the fact that \begin{equation}\text{Im}(\gamma\tau)=\frac{\text{Im}(\tau)}{|c\tau+d|^2}.\end{equation}
Thus we have 
\begin{equation}\label{eq:gtildetrans}
\begin{split}
\tilde{g}^{(m)}_{\frac{u}{c \tau+d}}\left(\frac{z}{c \tau+d} ; \frac{a \tau+b}{c \tau+d}\right)&=(c\tau+d)^2\mathbf{e}\left(\frac{m c\left(z^2-u^2\right)}{c \tau+d}\right)\left[\frac{1}{2\pi i} \frac{\partial\hat{f}^{(m)}}{\partial u}(u,z ; \tau)\right.\\&\left.\hspace{3cm}-2m\left(\frac{cu}{c\tau+d}+\frac{\text{Im}(u(c\bar{\tau}+d))}{y(c\tau+d)}\right)\hat{f}^{(m)}(u,z ; \tau)\right]\\&=(c\tau+d)^2\mathbf{e}\left(\frac{m c\left(z^2-u^2\right)}{c \tau+d}\right)\tilde{g}^{(m)}_u(z;\tau).
\end{split}
\end{equation}
Using similar manipulations as in Proposition \ref{prop 2.2}, one can show that 
\begin{equation}\label{eq:Bhatgtil}
\widehat{\mathcal{B}}^s_m(\tau,z)=\mathbf{e}(-m\alpha z_s)\tilde{g}^{(m)}_{z_s}(z;\tau).
\end{equation}
Equations \eqref{eq:gtildetrans} and \eqref{eq:Bhatgtil} are exact analogs of \eqref{eq:modtranscurlfm} and \eqref{eq:eAhatsmftilde} respectively. Then the proof follows exactly as in \eqref{eq:Ahatsmmodtrans}.
\end{proof}
\end{prop}
We describe the mock modular of $\varphi^F$ in this case below and prove a more general version in Section \ref{rational_index}.
\begin{thm}
Let $\varphi$ be as above, with double poles at $z=z_{s}\left(s \in S(\varphi) \subset \mathbb{Q}^{2}\right)$. Then the finite part $\varphi^{F}$ as defined by \eqref{eq 2.1} is a mixed mock Jacobi form. More precisely, the coefficients $h_{\ell}(\tau)$ as defined in \eqref{eq 2.2} are mixed mock modular forms of weight $k-\frac{1}{2}$, with completion given by
$$
\widehat{h}_{\ell}(\tau)=h_{\ell}(\tau)-\sum_{s \in S(\varphi) / \mathbb{Z}^{2}}\left(D_{s}(\tau) \Theta_{m, \ell}^{s *}(\tau)+E_{s}(\tau) \Xi_{m, \ell}^{s *}(\tau)\right).
$$
The first sum in $\widehat{h}_{\ell}$ has mixed weight $\left(k-1,\frac{1}{2}\right)$ and the second sum has mixed weight $\left(k-2,\frac{3}{2}\right)$. Moreover the completion $\widehat{\varphi}^{F}$ as in Theorem \ref{thm 2.3} transforms like a Jacobi form of weight $k$ and index $m$ with respect to $\Gamma\ltimes\mathbb{Z}^2.$
\label{thm 2.6}
\end{thm}
\section{Generalisation to certain rational index meromorphic Jacobi forms} 
\label{rational_index}
In this section, we will prove the mock modularity of finite parts of certain meromorphic Jacobi forms of rational index, see \cite{ojm/1502092828,Gritsenko:1999fk} for some related discussion on Jacobi forms of non-integral index. \subsection{The Space $J_{k,\frac{m}{N}}^N(\Gamma)$} 
For $m\in\mathbb{Z}$ and $N\geq 1$, consider the space of functions which transform as a Jacobi form of fractional index $m/N$ with respect to $\Gamma\ltimes (N\mathbb{Z}\times\mathbb{Z})$. Denote this space by $J_{k,\frac{m}{N}}^N(\Gamma).$ 
\par If a function $\varphi$ is invariant under the elliptic action of $N\mathbb{Z}\times\mathbb{Z}$ with index $m/N$, then its theta decomposition has the form 
\begin{equation}\label{eq:thetadecN}
\varphi ( \tau , z ) = \sum \limits_ { \ell \in \mathbb { Z } / 2 m \mathbb{Z} } h _ { \ell } ( \tau ) \widetilde{\vartheta} _ { \frac{m}{N} , \ell } (\tau , z ),
\end{equation}
where 
\begin{equation}
 \widetilde{\vartheta}_{\frac{m}{N} , \ell } ( \tau , z )=\vartheta _ { m , \ell } ( N\tau , z ),
\label{eq:thetatildeN}   
\end{equation}
Indeed, as in Section \ref{subsec:weakjac}, using $\varphi\underset{\frac{m}{N}}{|}(0,1)=\varphi$ we get $\varphi(\tau,z+1)=\varphi(\tau,z)$ and hence a Fourier expansion 
\begin{equation}
    \varphi(\tau,z)=\sum_{\ell\in\Z}c_\ell(\tau)\zeta^\ell.
\end{equation}
The invariance $\varphi\underset{\frac{m}{N}}{|}(N,0)=\varphi$ gives 
\begin{eqnarray}
 c_\ell(\tau)q^{N(\ell+m)}= c_{\ell+2m}(\tau),  
\end{eqnarray}
which implies 
\begin{eqnarray}
    c_\ell(\tau)q^{-N\ell^2/4m}=c_{\ell+2mk}(\tau)q^{-N(\ell+2mk)^2/4m},\quad k\in\Z.
\end{eqnarray}
This gives us the expansion \eqref{eq:thetadecN} with the identification $h_\ell(\tau)=c_\ell(\tau)q^{-N\ell^2/4m}$. 
Again, if $\varphi$ is holomorphic on $\mathbb{H}\times\mathbb{C}$, then the theta coefficients are given by 
\begin{equation}\label{eq 2.23}
h_{\ell}(\tau):=h_{\ell}^{\left(z_{0}\right)}(\tau)=q^{-\ell^{2}N / 4 m} \int_{z_{0}}^{z_{0}+1} \varphi(\tau, z) \mathbf{e}(-\ell z) d z,
\end{equation}
where $\mathbf{e}(t)=e^{2\pi it}$.
Moreover if $\varphi$ transforms as a Jacobi form of weight $k$ and index $m/N$ with respect to $\Gamma\ltimes (N\Z\times\Z)$, then each $h_{\ell}(\tau)$ is a weakly holomorphic modular form of weight $k-\frac{1}{2}$ with respect to $\Gamma(4mN),$ and the proof is analogous to integer index case \cite{DAS2015351}. 
Thus every element of $J_{k,\frac{m}{N}}^N(\Gamma)$ has a theta series expansion of the above form. 
\par It is not difficult to construct examples of functions in this space. For example, a natural procedure to obtain modular forms for congruence subgroups is by rescaling the variable and by using the Fricke involution. Indeed, if $f\in M_k^{!}(\mathrm{SL}(2,\mathbb{Z}))$, then $f(N\tau)\in M_k^{!}(\Gamma_0(N))$ for $k\in\mathbb{Z}$ and $N\geq 1$. Similarly, if $f(\tau)\in M_k^{!}(\Gamma)$ with $\Gamma\in\{\Gamma_0(N),\Gamma_1(N)\}$ and $k\in\mathbb{Z}$ then for
\begin{equation}
\widetilde{f}(\tau):=f\underset{k}{|}\omega(N)=N^{-k/2}\tau^{-k}f\left(-\frac{1}{N\tau}\right), 
\end{equation} 
we have that $\widetilde{f}\in M_k^{!}(\Gamma)$. Here $\omega(N)=\left(\begin{smallmatrix}0&-1\\N&0\end{smallmatrix}\right)$ is called the \textit{Fricke involution}. These simple operations on Jacobi forms give us functions in the space $J_{k,\frac{m}{N}}^N(\Gamma).$ Suppose $\varphi(\tau,z)$ transforms like a Jacobi form of weight $k$ and index $m$ with respect to $\Gamma\ltimes\mathbb{Z}^2$ where $\Gamma\subseteq\mathrm{SL}(2,\mathbb{Z})$. We define two functions $\varphi_1,\varphi_2$ as follows: 
\begin{equation}
\begin{split}
&\varphi_1(\tau,z)=\varphi(N\tau,z),\quad\text{if}\quad \Gamma\in\{\Gamma^0(N),\Gamma^1(N)\},\\
&\varphi_2(\tau,z)=\varphi\underset{k,m}{|}\omega(N)=N^{-k/2}\tau^{-k}\mathbf{e}\left(-mz^2/N\tau\right)\varphi\left(-\frac{1}{N\tau},\frac{z}{N\tau}\right)\\ & \hskip 3in \text{if}\quad \Gamma\in\{\Gamma_0(N),\Gamma_1(N)\}.
\end{split}
\label{eq:ratindjacform}
\end{equation}
Note that 
\begin{equation}
\varphi_1=N^{-k/2}\left(\varphi\underset{k,m}{|}r_N\right),\quad r_N=\begin{pmatrix}
N&0\\0&1
\end{pmatrix}.
\label{eq 1.8}
\end{equation} 
We have the following proposition.
\begin{prop}
Let $\varphi_i(\tau,z),~i=1,2$ be as above. Then we have 
\begin{equation}
\varphi_1\underset{k,\frac{m}{N}}{|}\gamma=\varphi_1,\quad \gamma\in\begin{cases}\Gamma_0(N)&\text{if}~\Gamma=\Gamma^0(N)\\\Gamma_1(N)&\text{if}~\Gamma=\Gamma^1(N),\end{cases}\quad \varphi_2\underset{k,\frac{m}{N}}{|}\gamma=\varphi_2,\quad\gamma\in\Gamma\in\{\Gamma_0(N),\Gamma_1(N)\}
\end{equation}
and 
\begin{equation}
\varphi_i\underset{\frac{m}{N}}{|}X=\varphi_i,\quad X\in N\mathbb{Z}\times\mathbb{Z}.
\end{equation}
\label{prop 1.7}
\begin{proof}
We first prove the modular transformation. From \eqref{eq:phigg'mdetg} we have 
\begin{equation}
\varphi_1\underset{k,\frac{m}{N}}{|}\gamma=\varphi\underset{k,m}{|}(r_N\gamma)=\left(\varphi\underset{k,m}{|}\widehat{\gamma}\right)\underset{k,m}{|}r_N=\varphi\underset{k,m}{|}r_N=\varphi_1,
\end{equation}
where we used the fact that $r_N\gamma=\widehat{\gamma}r_N$ with
\begin{equation}
\gamma=\begin{pmatrix}
a&b\\cN&d
\end{pmatrix},\quad \widehat{\gamma}=\begin{pmatrix}
a&bN\\c&d
\end{pmatrix}.
\end{equation}
Similarly
\begin{equation}
\varphi_2\underset{k,\frac{m}{N}}{|}\gamma=\varphi\underset{k,\frac{m}{N}}{|}(\omega(N)\gamma)=\varphi\underset{k,m}{|}(\widetilde{\gamma}\omega(N))=\left(\varphi\underset{k,m}{|}\widetilde{\gamma}\right)\underset{k,m}{|}\omega(N)=\varphi\underset{k,m}{|}\omega(N)=\varphi_2,
\end{equation}
where 
\begin{equation}
\gamma=\begin{pmatrix}
a&b\\cN&d
\end{pmatrix},\quad \widetilde{\gamma}=\begin{pmatrix}
d&-c\\-bN&a
\end{pmatrix}.
\end{equation}
To prove the elliptic transformation, we use \eqref{eq:Xgg'Xmdetg}. 
The desired identity follows once we observe that if we define
\begin{equation}
\begin{split}
&\widehat{X}_N=Xr_N,\quad \widetilde{X}_N=X\omega(N),
\end{split}
\end{equation}
then for every $X=(\lambda,\mu)\in\mathbb{Z}\times\mathbb{Z},$
\begin{equation}
\widehat{X}_{N}=(N\lambda,\mu)\in N\mathbb{Z}\times\mathbb{Z},\quad\widetilde{X}_N=(N\mu,-\lambda)\in N\mathbb{Z}\times\mathbb{Z}.    
\end{equation}
\end{proof}
\end{prop} 
One can use Proposition \ref{prop 1.7} to construct a mock Jacobi form of higher level starting from a mock Jacobi form for $\mathrm{SL}(2,\mathbb{Z})\ltimes\mathbb{Z}^2$. Indeed, if  
\begin{equation}
    \phi(\tau,z)=\sum \limits_ { \ell \in \mathbb { Z } / 2 m \mathbb{Z} } h _ { \ell } ( \tau ) \vartheta _ { m , \ell } ( \tau , z )
\end{equation}
is a mock Jacobi form of index $m$ for $\mathrm{SL}(2,\mathbb{Z})$, then $h _ { \ell } ( \tau )$ admits a completion $\widehat{h} _ { \ell } ( \tau )$ such that $\widehat{\phi}(\tau,z)$ given by 
\begin{eqnarray}\label{eq:phihatcomp}
    \widehat{\phi}(\tau,z)=\sum \limits_ { \ell \in \mathbb { Z } / 2 m \mathbb{Z} } \widehat{h} _ { \ell } ( \tau ) \vartheta _ { m , \ell } ( \tau , z )
\end{eqnarray}
transforms as a Jacobi form of index $m$ for $\mathrm{SL}(2,\mathbb{Z})\ltimes\mathbb{Z}^2$. Then from the proof of Proposition \ref{prop 1.7}, we see that \footnote{Note that if $\widehat{\phi}$ transforms as a Jacobi form for $\mathrm{SL}(2,\mathbb{Z})\ltimes\mathbb{Z}^2$, then it transforms as a Jacobi form for $\Gamma\ltimes\mathbb{Z}^2$ for any $\Gamma\subset\mathrm{SL}(2,\Z)$.} $\widehat{\phi}(N\tau,z)$ transforms as a Jacobi form of index $m/N$ for $\Gamma_0(N)\ltimes(N\mathbb{Z}\times\mathbb{Z})$. This means that $\widehat{h} _ { \ell } ( N\tau )$ transforms as a modular form for $\Gamma(4mN)$ since $\widehat{h} _ { \ell } ( \tau )$ transforms as a modular form for $\Gamma(4m)$. Thus $h _ { \ell } ( N\tau )$ is a mock modular form and $\phi(N\tau,z)$ is a mock Jacobi form of index $m/N$ for $\Gamma_0(N)\ltimes(N\mathbb{Z}\times\mathbb{Z})$.  
In particular, $\textbf{H}(N\tau,z)$ is a mock Jacobi form of weight 2 and index $1/N$ for $\Gamma_0(N)\ltimes(N\mathbb{Z}\times\mathbb{Z})$.\\\\
Observe that any holomorphic function $\psi\in J_{k,\frac{m}{N}}^N(\Gamma)$ with $\Gamma\in\{\Gamma_0(N),\Gamma_1(N)\}$ has  a theta decomposition of the form given in \eqref{eq:thetadecN}. 
Now suppose that $\psi$ has simple or double poles at $z=z_s$ with $s\in S(\psi)\subset\mathbb{Q}^2$ as in Section \ref{subsec:dmzmockjacform}. Then $\psi$ does not admit a theta decomposition as in \eqref{eq:thetadecN}. Moreover $h_{\ell}$ as in \eqref{eq 2.23} is not unique $\bmod ~2m$ for every $z_0\in\mathbb{C}$ and hence the sum in \eqref{eq:thetadecN} may not be well defined. To get over this problem \cite{Dabholkar:2012nd}, we will make a choice for $z_0$ such that $h_{\ell}$ defined in \eqref{eq 2.23} is uniquely defined $\bmod ~2m$. To do this, observe that
\begin{equation}\label{eq:hellNmod2n}
\begin{split}
h_{\ell}^{(z_0+N\tau)}&=q^{-\frac{\ell^{2}N}{4 m}} \int_{z_{0}}^{z_{0}+1} \psi(\tau, z+N \tau) \mathbf{e}(-\ell(z+N \tau)) d z\\&=q^{-\frac{\ell^{2} N}{4 m}} \int_{z_{0}}^{z_{0}+1} \psi(\tau, z) \mathbf{e}\left(-\frac{m}{N}\left(N^{2} \tau+2 N z\right)\right) \mathbf{e}(-\ell z) \mathbf{e}\left(-\ell N\tau\right)d z\\&=q^{-\left(\frac{ \ell^{2}N}{4 m}+m N+\ell N\right)} \int_{z_{0}}^{z_0+1} \psi(\tau, z) \mathbf{e}(-(\ell+2 m) z) d z\\&=q^{-\frac{N}{4 m}\left(\ell^{2}+4 m^{2}+4 m \ell\right)} \int_{z_{0}}^{z_{0}+1} \psi(\tau, z) \mathbf{e}(-(\ell+2 m) z) d z\\&=q^{-\frac{N}{4m}(\ell+2 m)^{2}}\int_{z_{0}}^{z_{0}+1} \psi(\tau, z) \mathbf{e}(-(\ell+2 m) z) d z\\&=h_{\ell+2m}^{(z_0)}.
\end{split}
\end{equation} 
So changing $z_0$ by $z_{0}+N \tau$ changes $\ell$ by $\ell+2 m$. So we choose $z_0=-\frac{\ell N}{2m}\tau$ so that 
\begin{equation}
h_{\ell}(\tau):=h_{\ell}^{\left(-\frac{\ell N}{2 m}\tau\right)}(\tau)
\label{eq 2.23'}
\end{equation}
is unique $\bmod ~2 m$ and then we can consider
\begin{equation}
\psi^{F}(\tau, z):=\sum_{\ell(\bmod~ 2 m)} h_{\ell}(\tau) \widetilde{\vartheta}_{\frac{m}{N}, \ell}(\tau, z).
\label{eq 2.24}
\end{equation}
We will show that $\psi^F$ is a mixed mock Jacobi form of weight $k$, index $\frac{m}{N}$ with respect to $\Gamma\ltimes(N\mathbb{Z}\times\mathbb{Z}).$ 
\subsection{Universal Appell-Lerch sum of index $m/N$}
Let us define the index $m/N$ universal Appell-Lerch sums. Write $s_N=(\alpha/N,\beta)$ 
and define
\begin{equation}
\begin{split}
&\widetilde{\mathcal{A}}_{\frac{m}{N}}^{s}(\tau, z)=\mathcal{A}_{m}^{s_N}(N\tau, z)\\&\widetilde{\mathcal{B}}_{\frac{m}{N}}^{s}(\tau, z)=\mathcal{B}_{m}^{s_N}(N\tau, z).
\end{split}
\label{eq 2.27}
\end{equation}
To describe the completion of these Appell-Lerch sums we need the index $m/N$ unary theta series and their completions defined as
\begin{equation}
\begin{split}
&\widetilde{\Theta}_{\frac{m}{N}, \ell}^{s}(\tau)=\Theta_{m,\ell}^{s_N}(N\tau),\quad\widetilde{\Theta}_{\frac{m}{N}, \ell}^{s*}(\tau)=\Theta_{m,\ell}^{s_N*}(N\tau);\\&\widetilde{\Xi}_{\frac{m}{N}, \ell}^{s}(\tau)=\Xi_{m, \ell}^{s_N}(N\tau),\quad\widetilde{\Xi}_{\frac{m}{N}, \ell}^{s*}(\tau)=\Xi_{m, \ell}^{s_N*}(N\tau).
\end{split}
\label{eq 2.29}
\end{equation} 
Now define the completions of the Appell-Lerch sums
\begin{equation}
\begin{split}
&\widehat{\widetilde{\mathcal{A}}^{s}}_{\frac{m}{N}}(\tau,z):=\widetilde{\mathcal{A}}_{\frac{m}{N}}^{s}(\tau, z)+\sum_{\ell(\bmod 2 m)} \widetilde{\Theta}_{\frac{m}{N}, \ell}^{s *}(\tau) \widetilde{\vartheta}_{\frac{m}{N}, \ell}(\tau, z);\\&\widehat{\widetilde{\mathcal{B}}^{s}}_{\frac{m}{N}}(\tau,z):=\widetilde{\mathcal{B}}_{\frac{m}{N}}^{s}(\tau, z)+m\sum_{\ell(\bmod 2 m)} \widetilde{\Xi}_{\frac{m}{N}, \ell}^{s *}(\tau) \widetilde{\vartheta}_{\frac{m}{N}, \ell}(\tau, z).
\end{split}
\label{eq 2.30}
\end{equation}
Using \eqref{eq 2.27} and \eqref{eq 2.29}, we have
\begin{equation}
\begin{split}
&\widehat{\widetilde{\mathcal{A}}^{s}}_{\frac{m}{N}}(\tau,z)=\widehat{\mathcal{A}}^{s_N}_{m}(N\tau,z);\\&\widehat{\widetilde{\mathcal{B}}^{s}}_{\frac{m}{N}}(\tau,z)=\widehat{\mathcal{B}}^{s_N}_{m}(N\tau,z).
\end{split}
\label{eq 2.31}
\end{equation}
\begin{prop}
Let $s\in\mathbb{Q}^2$ and $\widetilde{\mathcal{A}}_{\frac{m}{N}}^{s}(\tau, z)$ and $\widetilde{\mathcal{B}}_{\frac{m}{N}}^{s}(\tau, z)$ be the Appell-Lerch sums as in \eqref{eq 2.27}. Then their completions $\widehat{\widetilde{\mathcal{A}}^{s}}_{\frac{m}{N}}(\tau,z)$ and $\widehat{\widetilde{\mathcal{B}}^{s}}_{\frac{m}{N}}(\tau,z)$ as in \eqref{eq 2.30} satisfy 
\begin{equation}
\begin{split}
&\widehat{\widetilde{\mathcal{A}}}^{(\alpha+\lambda, \beta+\mu)}_{\frac{m}{N}}(\tau, z)=\mathbf{e}\left(-\frac{m}{N}(\mu \alpha-\lambda \beta)\right) \widehat{\widetilde{\mathcal{A}}}_{\frac{m}{N}}^{(\alpha, \beta)}(\tau,z),\quad(\lambda, \mu) \in N\mathbb{Z}\times\mathbb{Z},\\
&\widehat{\widetilde{\mathcal{A}}^{s}}_{\frac{m}{N}}(\tau, z+\lambda \tau+\mu)=\mathbf{e}\left(-\frac{m}{N}\left(\lambda^{2} \tau+2 \lambda z\right)\right) \widehat{\widetilde{\mathcal{A}}^{s}}_{\frac{m}{N}}(\tau,z), \quad(\lambda, \mu) \in N\mathbb{Z}\times\mathbb{Z}, \\
&\widehat{\widetilde{\mathcal{A}}^{s}}_{\frac{m}{N}}\left(\frac{a \tau+b}{c \tau+d}, \frac{z}{c \tau+d}\right)=(c \tau+d) \mathbf{e}\left(\frac{m}{N}\frac{c z^{2}}{c \tau+d}\right) \widehat{\widetilde{\mathcal{A}}}^{s \gamma}_{\frac{m}{N}}(\tau, z), \quad\gamma=\begin{pmatrix}
a &b \\
c &d
\end{pmatrix}
\in \Gamma,
\end{split}
\end{equation} 
and similarly for $\widehat{\widetilde{\mathcal{B}}^{s}}_{\frac{m}{N}}(\tau,z)$ except that the last transformation becomes 
\begin{equation}
\widehat{\widetilde{\mathcal{B}}^{s}}_{\frac{m}{N}}\left(\frac{a \tau+b}{c \tau+d}, \frac{z}{c \tau+d}\right)=(c \tau+d)^2 \mathbf{e}\left(\frac{m}{N}\frac{c z^{2}}{c \tau+d}\right) \widehat{\widetilde{\mathcal{B}}}^{s \gamma}_{\frac{m}{N}}(\tau, z), \quad\gamma=\begin{pmatrix}
a &b \\
c &d
\end{pmatrix}
\in \Gamma.
\end{equation}
\label{prop 2.7}
\begin{proof}
We first prove the modular transformation. Define the following slash operator on the completions of the original Appell-Lerch sums: for $\gamma=\left(\begin{smallmatrix}a&b\\c&d\end{smallmatrix}\right)\in\mathrm{GL}^+(2,\mathbb{R})$,
\begin{equation}
\begin{split}
\left(\widehat{\mathcal{A}}_m^s\underset{1,m}{|}\gamma\right)(\tau,z)=(\text{det}\gamma)^{1/2}(c\tau+d)^{-1}\mathbf{e}\left(-\frac{m}{\text{det}\gamma}\frac{cz^2}{c\tau+d}\right)\widehat{\mathcal{A}}_m^{s\gamma^{-1}}\left(\frac{a\tau+b}{c\tau+d},\frac{z}{c\tau+d}\right);\\
\left(\widehat{\mathcal{B}}_m^s\underset{2,m}{|}\gamma\right)(\tau,z)=(\text{det}\gamma)(c\tau+d)^{-2}\mathbf{e}\left(-\frac{m}{\text{det}\gamma}\frac{cz^2}{c\tau+d}\right)\widehat{\mathcal{B}}_m^{s\gamma^{-1}}\left(\frac{a\tau+b}{c\tau+d},\frac{z}{c\tau+d}\right).
\end{split}
\end{equation} 
Note that we need these modified slash operators since the superscript $s$ transforms nontrivially under the modular action. Proposition \ref{prop 2.2} and Proposition \ref{prop 2.5} imply that $\widehat{\mathcal{A}}_m^s$ and $\widehat{\mathcal{B}}_m^s$ are invariant under the above slash operators for every $s\in\mathbb{Q}^2$ and $\gamma\in\mathrm{SL}(2,\Z)$. Moreover, \eqref{eq 2.31} implies that 
\begin{equation}
\begin{split}
&\widehat{\widetilde{\mathcal{A}}^{s}}_{\frac{m}{N}}(\tau,z)=N^{-1/2}\left(\mathcal{A}^{s}_{m}\underset{1,m}{|}r_N\right)(\tau,z);\\&\widehat{\widetilde{\mathcal{B}}^{s}}_{\frac{m}{N}}(\tau,z)=N^{-1}\left(\mathcal{B}^{s}_{m}\underset{2,m}{|}r_N\right)(\tau,z),
\end{split}
\label{eq 2.32}
\end{equation}
where $r_N$ is as in \eqref{eq 1.8}. Since $s\gamma_1^{-1}\gamma_2^{-1}=s(\gamma_2\gamma_1)^{-1},$ Proposition \ref{prop 1.7} implies the elliptic and modular invariance with respect to $N\mathbb{Z}\times\mathbb{Z}$ and $\Gamma.$ For the first transformation property, note that using \eqref{eq 2.31} and Proposition \ref{prop 2.2} we have 
\begin{equation}
\begin{split}
\widehat{\widetilde{\mathcal{A}}}^{(\alpha+\lambda, \beta+\mu)}_{\frac{m}{N}}(\tau, z)=\widehat{\mathcal{A}}^{\left(\frac{\alpha+\lambda}{N}, \beta+\mu\right)}_{m}(N\tau, z)&=\mathbf{e}\left(-m\left(\mu \frac{\alpha}{N}-\frac{\lambda}{N} \beta\right)\right) \widehat{\mathcal{A}}_{m}^{\left(\frac{\alpha}{N}, \beta\right)}(N\tau,z)\\&=\mathbf{e}\left(-\frac{m}{N}\left(\mu\alpha-\lambda \beta\right)\right)\widehat{\widetilde{\mathcal{A}}}^{(\alpha, \beta)}_{\frac{m}{N}}(\tau, z).
\end{split}
\end{equation}
\end{proof}
\end{prop}
\noindent Now we come back to $\psi$. Suppose $\psi$ has only simple poles at $z=z_s$ for $s\in S(\psi).$ Proceeding as usual, we define the first order residues by 
\begin{equation}
D_{s}(\tau)=2 \pi i \mathbf{e}\left(\frac{m}{N} \alpha z_{s}\right) \operatorname{Res}_{z=z_{s}}(\psi(\tau, z)) \quad\left(s=(\alpha, \beta) \in S(\psi), z_{s}=\alpha \tau+\beta\right).
\label{eq 2.33}
\end{equation}
The transformations for $\psi$ gives us the following proposition:
\begin{prop}
Let $\psi\in J_{k,\frac{m}{N}}^N(\Gamma)$ be a meromorphic function with simple poles at $z=z_s=\alpha\tau+\beta$ for $s=(\alpha,\beta)\in S(\psi)\subset\mathbb{Q}^2$. Then the residues $D_s(\tau)$ defined in \eqref{eq 2.33} satisfies the following properties:
\begin{equation}
\begin{split}
& D_{(\alpha+\lambda, \beta+\mu)}(\tau)=\mathbf{e}\left(\frac{m}{N}(\mu \alpha-\lambda \beta)\right) D_{(\alpha, \beta)}(\tau),\quad(\lambda, \mu) \in N\mathbb{Z}\times\mathbb{Z}\\
& D_{s}\left(\frac{a \tau+b}{c \tau+d}\right)=(c \tau+d)^{k-1} D_{s \gamma}(\tau),\quad\gamma=\begin{pmatrix}
a& b \\ c& d
\end{pmatrix}\in \Gamma,
\end{split}
\end{equation}
where $s\gamma=(a\alpha+c\beta,b\alpha+d\beta).$
\label{prop 2.8}
\begin{proof}
A slight modification of Proposition \ref{prop:Dstrans}.
\end{proof}
\end{prop}
\noindent We now define the polar part of $\psi$ as 
\begin{equation}
\psi^P(\tau,z):=\sum_{s\in S(\psi)/(N\mathbb{Z}\times\mathbb{Z})}D_s(\tau)\widetilde{\mathcal{A}}^{s}_{\frac{m}{N}}(\tau,z).
\label{eq 2.34}
\end{equation} 
This sum is well defined because $D_s(\tau)$ and $\widetilde{\mathcal{A}}^{s}_{\frac{m}{N}}$ transform with opposite factors under $s\mapsto s+(\lambda,\mu)$ with $(\lambda,\mu)\in N\mathbb{Z}\times\mathbb{Z}$. 
\begin{thm}\label{thm 2.9}
Let $\psi\in J_{k,\frac{m}{N}}^N(\Gamma)$ be a meromorphic function with simple poles at $z=z_s=\alpha\tau+\beta$ for $s=(\alpha,\beta)\in S(\psi)\subset\mathbb{Q}^2$. Let $\psi^F(\tau,z)$ and $\psi^P(\tau,z)$ be as in \eqref{eq 2.24} and \eqref{eq 2.34} respectively. Then 
\begin{equation}
\psi(\tau,z)=\psi^F(\tau,z)+\psi^P(\tau,z).
\label{eq:psip+f}
\end{equation}
\begin{proof}
Since the functions $\psi, \psi^{F}$ and $\psi^{P}$ are meromorphic, it is enough to prove the decomposition \eqref{eq:psip+f} on a horizontal line in the $z$-plane. Let us take a point $P=A \tau+B \in \mathbb{C}$ with $(A, B) \in \mathbb{R}^{2} \backslash S(\psi).$ Since $S(\psi)\subseteq \mathbb{Q}\tau+\mathbb{Q}$ for any $\tau\in\mathbb{H}$, it is easy to see that the horizontal line $\operatorname{Im}(z)=\operatorname{Im}(P)=A\cdot\text{Im}(\tau)$ does not intersect $S(\psi)$. We will prove \eqref{eq:psip+f} on this horizontal line. The Fourier expansion of $\psi$ on this line has the form 
$$
\psi(\tau, z)=\sum_{\ell \in \mathbb{Z}} q^{\frac{\ell^{2}N}{4 m}} h_{\ell}^{(P)}(\tau) \zeta^{\ell}
$$
where the coefficients $h_{\ell}^{(P)}$ are defined by \eqref{eq 2.23}.  Then using \eqref{eq 2.24} we have
\begin{equation}
\begin{split}
\psi(\tau, z)-\psi^{F}(\tau, z)&=\sum_{\ell \in \mathbb{Z}}\left(h_{\ell}^{(P)}(\tau)-h_{\ell}(\tau)\right) q^{\frac{\ell^{2}N}{4 m}} \zeta^{\ell} \quad(\operatorname{Im}(z)=\operatorname{Im}(P))\\&=\sum_{\ell \in \mathbb{Z}}\left(\int_{P}^{P+1}dz'\psi(\tau,z')\mathbf{e}(-\ell z')-\int_{-\frac{\ell N}{2m}\tau}^{-\frac{\ell N}{2m}\tau+1}dz'\psi(\tau,z')\mathbf{e}(-\ell z')\right)\zeta^{\ell}.
\label{eq 2.35}
\end{split}
\end{equation}
Since $\psi(\tau,z+1)=\psi(\tau,z)$ the above integral can be written as a contour integral 
\begin{equation}
\psi(\tau, z)-\psi^{F}(\tau, z)=\sum_{\ell \in \mathbb{Z}}\left(\int_{C(P,\tau,\ell)}dz'\psi(\tau,z')\mathbf{e}(-\ell z')\right)\zeta^{\ell},
\end{equation} 
where the contour is shown in Figure \ref{fig:contrec}. 
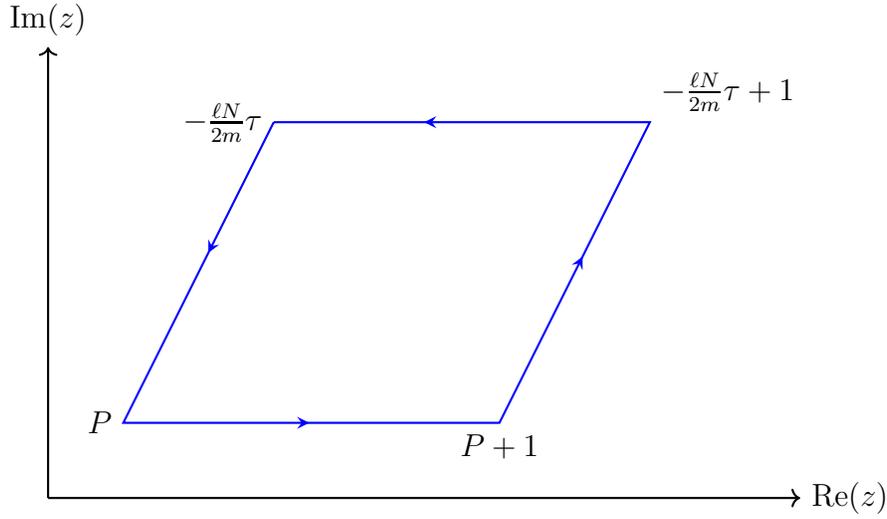
\begin{figure}[!h]
\centering
\begin{tikzpicture}[decoration={markings,
    mark=at position 2cm   with {\arrowreversed[line width=1pt]{stealth}},
    mark=at position 7cm   with {\arrowreversed[line width=1pt]{stealth}},
    mark=at position 12cm with {\arrowreversed[line width=1pt]{stealth}},
    mark=at position 17cm with {\arrowreversed[line width=1pt]{stealth}}
  }]
  \draw[thick, ->] (0,0) -- (10,0)node[right, black] {$\text{Re}(z)$} coordinate (xaxis);

  \draw[thick, ->] (0,0) -- (0,6)node[above, black] {$\text{Im}(z)$} coordinate (yaxis);

  \node[above] at (xaxis) {};

  \node[right]  at (yaxis) {};

  \path[draw,blue, line width=0.8pt, postaction=decorate] 
        (3,5)  node[left, black] {$-\frac{\ell N}{2m}\tau$}
    --  (8,5)  node[above right, black] {$-\frac{\ell N}{2m}\tau+1$} 
    --  (6,1)  node[midway, right, black] {} node[below, black] {$P+1$}
    --  (1,1)  node[midway, below, black] {} node[left, black] {$P$} 
    --  (3,5)  node[midway, left, black] {};

\end{tikzpicture}
\caption{Contour $C(P,\tau,\ell)$ for the integral}
\label{fig:contrec}
\end{figure}
Cauchy's residue theorem then implies that $q^{\ell^{2}N/ 4 m}\left(h_{\ell}^{(P)}(\tau)-h_{\ell}(\tau)\right)$ is $2 \pi i$ times the sum of the residues of poles of $\psi(\tau, z) \mathbf{e}(-\ell z)$ in the parallelogram shown in Figure \ref{fig:contrec} above. If there are poles on the upper horizontal line, then we deform the contour by a semicircle around the pole in the usual way and add an extra factor of $1/2$ to the residues of those poles \footnote{This procedure looks somewhat ad hoc, but we show in Appendix \ref{sec:poleattcont} that in the cases of interest, the residues at the poles on the upper horizontal line vanish and therefore there is no ambiguity involving the choice of integration contour.}. We thus have
\begin{equation}
\begin{split}
&\psi(\tau, z)-\psi^{F}(\tau, z)\\&=\sum_{\ell\in\Z}q^{\frac{\ell^{2}N}{4 m}}\zeta^\ell\left(h_{\ell}^{(P)}(\tau)-h_{\ell}(\tau)\right) \\&=2 \pi i \sum_{\ell\in\Z}\sum_{s=(\alpha, \beta) \in S(\psi) / \mathbb{Z}} \frac{\operatorname{sgn}(\alpha-A)-\operatorname{sgn}(\alpha+\ell N/ 2 m)}{2} \operatorname{Res}_{z=z_{s}}(\psi(\tau, z) \mathbf{e}(-\ell z))\zeta^\ell \\
&=\sum_{\ell\in\Z}\sum_{s=(\alpha, \beta) \in S(\psi) / \mathbb{Z}} \frac{\operatorname{sgn}(\alpha-A)-\operatorname{sgn}(\ell+2 m \alpha/N)}{2} D_{s}(\tau) \mathbf{e}\left(-(\ell+m \alpha/N) z_{s}\right)\zeta^\ell,
\end{split}
\end{equation}
where ``$(\alpha, \beta) \in S(\psi) / \mathbb{Z}$" means that we consider all $\alpha$, but $\beta$ only modulo 1 , which is the same by periodicity as considering only the $(\alpha, \beta)$ with $B \leq \beta<B+1.$ Note that $\text{sgn}(\ell+2m\alpha/N)$ is defined to be zero if $\ell+2m\alpha/N=0$, \textit{i,e.} if $s$ is on the upper horizontal line and 1 or $-1$ otherwise. Plugging this into \eqref{eq 2.35}, using \eqref{eq 2.7} and using the fact that on $\text{Im}(z)=\text{Im}(P)=A$
\begin{equation}
\operatorname{sgn}(\alpha-A)=-\operatorname{sgn}\left(\operatorname{Im}\left(z-z_{s}\right)\right),
\end{equation} 
we get
\begin{equation}
\begin{split}
&\psi (\tau, z)-\psi^{F}(\tau, z) \\&=-\sum_{s=(\alpha, \beta) \in S(\psi) / \mathbb{Z}} \mathbf{e}\left(-\frac{m}{N} \alpha z_{s}\right) D_{s}(\tau) \sum_{\ell \in \mathbb{Z}} \frac{\operatorname{sgn}\left(\operatorname{Im}\left(z-z_{s}\right)\right)+\operatorname{sgn}(\ell+2 m \alpha/N)}{2}\left(\frac{\zeta}{\zeta_{s}}\right)^{\ell} \\
&=\sum_{s=(\alpha, \beta) \in S(\psi) / \mathbb{Z}} \mathbf{e}\left(-\frac{m}{N} \alpha z_{s}\right) D_{s}(\tau) \mathcal{R}_{-2 \frac{m}{N} \alpha}\left(\zeta / \zeta_{s}\right) \\
&=\sum_{s=(\alpha, \beta) \in S(\psi) / (N\mathbb{Z}\times\mathbb{Z})} \sum_{\lambda \in \mathbb{Z}} \mathbf{e}\left(-\frac{m}{N}(\alpha-N\lambda)\left(z_{s}-N\lambda \tau\right)\right) D_{(\alpha-N\lambda, \beta)}(\tau)\\&\hspace{9cm}\times\mathcal{R}_{-2 \frac{m}{N}(\alpha-N\lambda)}\left(q^{N\lambda} \zeta / \zeta_{s}\right) \\
&=\sum_{s=(\alpha, \beta) \in S(\psi) / (N\mathbb{Z}\times\mathbb{Z})} D_{s}(\tau) \mathbf{e}\left(-\frac{m}{N} \alpha z_{s}\right) \sum_{\lambda \in \mathbb{Z}} q^{mN \lambda^{2}} \zeta^{2 m\lambda} \mathcal{R}_{-2 \frac{m}{N} \alpha}\left(q^{N\lambda} \zeta / \zeta_{s}\right)\\&=\sum_{s=(\alpha, \beta) \in S(\psi) / (N\mathbb{Z}\times\mathbb{Z})} D_{s}(\tau)\mathcal{A}^{s_N}_{m}(N\tau,z)\\&=\sum_{s\in S(\psi)/(N\mathbb{Z}\times\mathbb{Z})}D_s(\tau)\widetilde{\mathcal{A}}^{s}_{\frac{m}{N}}(\tau,z)\\&=\psi^P(\tau,z),
\end{split}
\end{equation}
where we have used the periodicity property of $D_{s}(\tau)$ in Proposition \ref{prop 2.8} together with the obvious periodicity property $\mathcal{R}_{c+n}(\zeta)=\zeta^{n} \mathcal{R}_{c}(\zeta)$ of $\mathcal{R}_{c}(\zeta).$
\end{proof}
\end{thm}
\noindent We now have the following theorem  using Proposition \ref{prop 2.7} and Theorem \ref{thm 2.9}.
\begin{thm}
\label{thm_rational_index}
Let $\psi\in J_{k,\frac{m}{N}}^N(\Gamma)$ be a meromorphic function with simple poles on the set $S(\psi)\subset\mathbb{Q}^2$ as above. Let $h_{\ell},\psi^{F}$ and $D_s(\tau)$ be as in \eqref{eq 2.23'}, \eqref{eq 2.24} and \eqref{eq 2.33} respectively. Then each $h_{\ell}$ is a mixed mock modular form of mixed weight $\left(k-1,\frac{1}{2}\right)$ with shadow $-\sqrt{\frac{m}{N}}\sum_{s \in S(\psi) / (N\mathbb{Z}\times\mathbb{Z})} D_{s}(\tau) \overline{\widetilde{\Theta}_{\frac{m}{N}, \ell}^{s}(\tau)},$ and the function $\psi^{F}$ is a mixed mock Jacobi form. More precisely, for each $\ell \in \mathbb{Z} / 2 m \mathbb{Z}$ the completion of $h_{\ell}$ defined by
$$
\widehat{h}_{\ell}(\tau):=h_{\ell}(\tau)-\sum_{s \in S(\psi) / (N\mathbb{Z}\times\mathbb{Z})} D_{s}(\tau) \widetilde{\Theta}_{\frac{m}{N}, \ell}^{s *}(\tau)
$$
with $\widetilde{\Theta}_{\frac{m}{N}, \ell}^{s *}$ as defined in \eqref{eq 2.29}, transforms like a modular form of weight $k-\frac{1}{2}$ with respect to $\Gamma(4mN)$  and the completion of $\psi^{F}$ defined by
$$
\widehat{\psi}^{F}(\tau, z):=\sum_{\ell(\bmod 2 m)} \widehat{h}_{\ell}(\tau) \widetilde{\vartheta}_{\frac{m}{N}, \ell}(\tau, z)
$$
transforms like a Jacobi form of weight $k$ and index $\frac{m}{N}$ with respect to $\Gamma\ltimes(N\mathbb{Z}\times\mathbb{Z})$.
\label{thm 2.10}
\begin{proof}
Let us start by showing that $\widehat{h}_{\ell}(\tau)$ transforms like a modular form of weight $k-\frac{1}{2}$. Observe that 
\begin{equation}
\begin{split}
\widehat{\psi}^{F}(\tau, z)&=\sum_{\ell(\bmod 2 m)} \widehat{h}_{\ell}(\tau) \widetilde{\vartheta}_{\frac{m}{N}, \ell}(\tau, z)\\&=\psi^F(\tau,z)-\sum_{s \in S(\psi) / (N\mathbb{Z}\times\mathbb{Z})} D_{s}(\tau) \sum_{\ell(\bmod 2 m)}\widetilde{\Theta}_{\frac{m}{N}, \ell}^{s *}(\tau)\widetilde{\vartheta}_{\frac{m}{N}, \ell}(\tau, z)\\&=\psi^F(\tau,z)-\sum_{s \in S(\psi) / (N\mathbb{Z}\times\mathbb{Z})} D_{s}(\tau) \widehat{\widetilde{\mathcal{A}}^{s}}_{\frac{m}{N}}(\tau,z)+\sum_{s \in S(\psi) / (N\mathbb{Z}\times\mathbb{Z})} D_{s}(\tau)\widetilde{\mathcal{A}}_{\frac{m}{N}}^{s}(\tau, z)\\&=\psi^F(\tau,z)+\psi^P(\tau,z)-\widehat{\psi}^{P}(\tau, z)\\&=\psi(\tau,z)-\widehat{\psi}^{P}(\tau, z),
\end{split}
\end{equation}
where we defined 
\begin{equation}
\widehat{\psi}^{P}(\tau, z):=\sum_{s \in S(\psi) / (N\mathbb{Z}\times\mathbb{Z})} D_{s}(\tau) \widehat{\widetilde{\mathcal{A}}^{s}}_{\frac{m}{N}}(\tau,z)
\end{equation}
and used Theorem \ref{thm 2.9}.
Using Proposition \ref{prop 2.7} and \ref{prop 2.8}, it is easy to see that $\widehat{\psi}^{P}(\tau, z)$ transforms like a Jacobi form of weight $k$ and index $m/N$ for $\Gamma\ltimes (N\mathbb{Z}\times\mathbb{Z})$. It follows that $\widehat{\psi}^{F}(\tau, z)$ transforms like a Jacobi form of weight $k$ and index $m/N$ for $\Gamma\ltimes (N\mathbb{Z}\times\mathbb{Z})$. Since $\widehat{h}_{\ell}$ are the theta coefficients of $\widehat{\psi}^F(\tau,z)$, which transforms like a  Jacobi form  of weight $k$ and index $m/N$ for $\Gamma\ltimes (N\mathbb{Z}\times\mathbb{Z})$, it transforms like a modular form of weight $k-\frac{1}{2}$ with respect to $\Gamma(4mN)$. This implies that $h_{\ell}$ are mock modular forms and $\psi^F$ is a mock Jacobi form. The mixed weight follows from the fact that $D_s$ has weight $k-1$ and $\widetilde{\Theta}_{\frac{m}{N}, \ell}^{s *}$ has weight $1/2$.
\end{proof}
\end{thm}
Finally, when $\psi$ has double poles at $z=z_s$ with $s\in S(\psi)$ then we define the first and second order residues $D_s(\tau)$ and $E_s(\tau)$ as before by
\begin{equation}
\mathbf{e}\left(\frac{m}{N} \alpha z_{s}\right) \psi\left(\tau, z_{s}+\varepsilon\right)=\frac{E_{s}(\tau)}{(2 \pi i \varepsilon)^{2}}+\frac{D_{s}(\tau)-2 \frac{m}{N} \alpha E_{s}(\tau)}{2 \pi i \varepsilon}+O(1) ~~\text{as}~~\varepsilon\to 0.
\label{eq:resdoupole}
\end{equation}
The proof of the proposition below is on the same lines.
\begin{prop}
The functions $E_{s}(\tau)$ and $D_{s}(\tau)$ defined by \eqref{eq:resdoupole} satisfy: for $\begin{pmatrix}a& b \\ c &d\end{pmatrix} \in \Gamma$ we have
$$
E_{s}\left(\frac{a \tau+b}{c \tau+d}\right)=(c \tau+d)^{k-2} E_{s \gamma}(\tau), \quad D_{s}\left(\frac{a \tau+b}{c \tau+d}\right)=(c \tau+d)^{k-1} D_{s \gamma}(\tau).
$$
\label{prop 2.11}
\begin{proof}
A slight modification of Proposition \ref{prop:DsEstrans}.
\end{proof}
\end{prop}
We now define the polar part by 
\begin{equation}
\psi^{P}(\tau, z)=\sum_{s \in S(\psi) / (N\mathbb{Z}\times\mathbb{Z})}\left(D_{s}(\tau) \widetilde{\mathcal{A}}_{\frac{m}{N}}^{s}(\tau, z)+E_{s}(\tau) \widetilde{\mathcal{B}}_{\frac{m}{N}}^{s}(\tau, z)\right).
\label{eq 2.38}
\end{equation}
Let us outline the proof of Theorem \ref{thm 2.9} in this setting. As before we can write 
\begin{equation}
\psi(\tau, z)-\psi^{F}(\tau, z)=\sum_{\ell \in \mathbb{Z}}\left(\int_{C(P,\tau,\ell)}dz'\psi(\tau,z')\mathbf{e}(-\ell z')\right)\zeta^{\ell},
\end{equation}
but now the contour integral on the right hand side will give us residue which will recieve contribution from both the functions $D_s$ and $E_s$. More precisely, we have 
\begin{equation}
\begin{split}
\psi(\tau, z)&-\psi^{F}(\tau, z)\\&=\sum_{\ell\in\Z}\left(h_{\ell}^{(P)}(\tau)-h_{\ell}(\tau)\right)q^{\frac{\ell^{2}N}{4 m}}\zeta^\ell \\&=2 \pi i \sum_{\ell\in\Z}\sum_{s=(\alpha, \beta) \in S(\psi) / \mathbb{Z}} \frac{\operatorname{sgn}(\alpha-A)-\operatorname{sgn}(\alpha+\ell N/ 2 m)}{2} \operatorname{Res}_{z=z_{s}}(\psi(\tau, z) \mathbf{e}(-\ell z))\zeta^\ell \\
&=\sum_{\ell\in\Z}\sum_{s=(\alpha, \beta) \in S(\psi) / \mathbb{Z}} \frac{\operatorname{sgn}(\alpha-A)-\operatorname{sgn}(\ell+2 m \alpha/N)}{2} \\&\hspace{4cm}\times\left[D_{s}(\tau)-\left(\ell+\frac{2m\alpha}{N}\right) E_{s}(\tau)\right] \mathbf{e}\left(-(\ell+m \alpha/N) z_{s}\right)\zeta^\ell.
\end{split}
\end{equation}
This can be expressed as,
\begin{equation}
\begin{split}
&\psi (\tau, z)-\psi^{F}(\tau, z) \\&=\sum_{s\in S(\psi)/(N\mathbb{Z}\times\mathbb{Z})}D_s(\tau)\widetilde{\mathcal{A}}^{s}_{\frac{m}{N}}(\tau,z)\\&+\sum_{s=(\alpha, \beta) \in S(\psi) / \mathbb{Z}} \mathbf{e}\left(-\frac{m}{N} \alpha z_{s}\right) E_{s}(\tau) \sum_{\ell \in \mathbb{Z}} \frac{(\ell+2m\alpha/N)\operatorname{sgn}\left(\operatorname{Im}\left(z-z_{s}\right)\right)+|\ell+2m\alpha/N|}{2}\left(\frac{\zeta}{\zeta_{s}}\right)^{\ell} \\
&=\sum_{s\in S(\psi)/(N\mathbb{Z}\times\mathbb{Z})}D_s(\tau)\widetilde{\mathcal{A}}^{s}_{\frac{m}{N}}(\tau,z)+\sum_{s=(\alpha, \beta) \in S(\psi) / \mathbb{Z}} \mathbf{e}\left(-\frac{m}{N} \alpha z_{s}\right) E_{s}(\tau) \mathcal{R}^{(2)}_{-2 \frac{m}{N} \alpha}\left(\zeta / \zeta_{s}\right).
\end{split}
\end{equation}  
Using \eqref{eq:Ess+lm} and $\mathcal{R}_{c+n}^{(2)}(\zeta)=\zeta^n\mathcal{R}_{c}^{(2)}(\zeta)$ and doing similar manipulations as in the proof of Theorem \ref{thm 2.9},  we obtain the required identity. 

The mock modularity in this case is exactly as in Theorem \ref{thm 2.6} with obvious modifications. We record it for completeness.
\begin{thm}
Let $\psi\in J_{k,\frac{m}{N}}^N(\Gamma)$ be as above, with double poles at $z=z_{s}\left(s \in S(\psi) \subset \mathbb{Q}^{2}\right)$. Then the finite part $\psi^{F}$ as defined by \eqref{eq 2.24} are mixed mock Jacobi forms. More precisely, the coefficients $h_{\ell}(\tau)$ as defined in \eqref{eq 2.23'} are mixed mock modular forms, with completion given by
$$
\widehat{h}_{\ell}(\tau)=h_{\ell}(\tau)-\sum_{s \in S(\psi) / (N\mathbb{Z}\times\mathbb{Z})}\left(D_{s}(\tau) \widetilde{\Theta}_{\frac{m}{N}, \ell}^{s *}(\tau)+E_{s}(\tau) \widetilde{\Xi}_{\frac{m}{N}, \ell}^{s *}(\tau)\right).
$$
The first sum in $\widehat{h}_{\ell}$ has mixed weight $\left(k-1,\frac{1}{2}\right)$ and the second sum has mixed weight $\left(k-2,\frac{3}{2}\right)$. Moreover the completion $\widehat{\psi}^{F}$ as in Theorem \ref{thm 2.10} transforms like a Jacobi form of weight $k$ and index $\frac{m}{N}$ with respect to $\Gamma\ltimes (N\mathbb{Z}\times\mathbb{Z}).$
\label{thm 2.12}
\end{thm} 
\section{Mock Jacobi forms from CHL model partition functions}\label{sec:mockjacchl}

Our goal in this section is to show that the degeneracies of single-centered dyons in CHL models \cite{Chaudhuri:1995fk, David:2006ji, David:2006ud, David:2006yn, Sen:2007qy, Sen:2009md, Sen:2010ts, Sen:2011ktd} are given by Fourier coefficients of Mock Jacobi forms for a large class of charges. To construct these CHL models, we start with Type IIB string theory compactified on $\mathcal{M}\times\widetilde{S}^1\times S^1$ where $\mathcal{M}$ is either K3 or $T^4$. We then consider a $\mathbb{Z}_N$ orbifold of this theory where the $\mathbb{Z}_N$ symmetry is generated by a transformation $g$ that involves $1/N$ units of shift along the circle $S^1$ together with an order $N$ transformation $\widetilde{g}$ in $\mathcal{M}$. The transformation $\widetilde{g}$ is chosen in such a way that the orbifold theory has $\mathcal{N}=4$ supersymmetry. The elliptic genus for the worldsheet orbifold SCFT is then expanded as 
\begin{equation}\label{eq:Frsdef}
F^{(r,s)}(\tau,z)\equiv\sum_{b=0}^1\sum_{\substack{\ell\in 2\mathbb{Z}+b\\n\in\mathbb{Z}/N}}c_b^{(r,s)}(4n-\ell^2)q^{n}\zeta^\ell,\quad q=\mathbf{e}(\tau),\quad\zeta=\mathbf{e}(z),
\end{equation} 
with the coefficients $c_b^{(r,s)}$ taking different values in different
CHL models.\par CHL models have quarter BPS dyons carrying both electric and magnetic charges. Our analysis in this section will focus on the counting of these states. Part of this contribution comes from black holes with a single center and the rest comes from black holes with two centers, but our interest will be on single-centered black holes.

\subsection{The CHL Partition Function}\label{sec:chl_partition_function}
Let us start with the partition functions of quater BPS dyons for general CHL models -- their relation
to degeneracies of single-centered dyons will be discussed later. 
It is given by  $\widetilde{\Phi}_k^{-1}$ where $\widetilde{\Phi}_k:\mathbb{H}_2\longrightarrow\mathbb{C}$ is a Siegel modular form of weight $k$ with respect to a subgroup $\widetilde{G}$ of $\mathrm{Sp}(2,\mathbb{Z})$.  We refer to Appendix \ref{app:Phiktrans} for an explicit description of $\widetilde{G}$.
There exists an infinite product representation for $\widetilde{\Phi}_k$, 
given by \cite{David:2006ud}
\begin{equation}
\begin{array}{l}
\widetilde{\Phi}_k(\tau, z,\sigma)=e^{2 \pi i(\widetilde{\alpha} \tau+\widetilde{\gamma} \sigma+z)} \\
\qquad \times \prod\limits_{b=0}^{1} \prod\limits_{r=0}^{N-1} \prod\limits_{k^{\prime} \in \mathbb{Z}+\frac{r}{N},j\in 2\mathbb{Z}+b,l\in\Z \atop k^{\prime}, l \geq 0, j<0 \text { for } k^{\prime}=l=0}\left(1-\exp \left(2 \pi i\left(k^{\prime}\sigma+l \tau+j z\right)\right)\right)^{\sum_{s=0}^{N-1} e^{-2 \pi i s l / N} c_{b}^{(r, s)}\left(4 k^{\prime} l-j^{2}\right)}
\end{array}
\label{eq:Phiprodrep}
\end{equation}
where
\begin{equation}
\widetilde{\alpha}:=\frac{1}{24 N} Q_{0,0}-\frac{1}{2 N} \sum_{s=1}^{N-1} Q_{0, s} \frac{e^{-2 \pi i s / N}}{\left(1-e^{-2 \pi i s / N}\right)^2}, \quad \widetilde{\gamma}:=\frac{1}{24 N} Q_{0,0}=\frac{1}{24 N} \chi(\mathcal{M}) 
\end{equation} 
with $Q_{r,s}$ defined by 
\begin{equation}
Q_{r,s}:=N\left(c_0^{(r,s)}(0)+2c_1^{(r,s)}(-1)\right).
\end{equation}
Also $\chi(\mathcal{M})$ denotes the Euler number of the manifold $\mathcal{M}$. Note that $\widetilde{\alpha},N\widetilde{\gamma}\in\mathbb{Z}$. \par The weight $k$ of the Siegel modular form $\Phi_k$ is given by 
\begin{equation}
    k=\frac{1}{2}\sum_{s=0}^{N-1}c_0^{(0,s)}(0).
\end{equation}
For $\mathcal{M}=K3$, the weight $k$ is related to $N$ by the formula
\begin{equation}
    k=\begin{cases}
        \frac{24}{N+1}-2,&N=1,2,3,5,7;\\3,&N=4;\\2,& N=6;\\1,& N=8.
    \end{cases}
\end{equation}
Note that for $N=1$, the CHL partition function $\widetilde{\Phi}_k^{-1}$ is exactly the inverse of the Igusa cusp form $\Phi_{10}$ and was analyzed in detail in \cite{Dabholkar:2012nd}.\par 
The Siegel modular form $\widetilde{\Phi}_k$ has double zeros on a hypersurface given by \cite{David:2006ud}\cite[Eq. (5.6.4)]{Sen:2007qy} 
\begin{equation}
\begin{aligned}
& n_2\left(\sigma \tau-z^2\right)+j z+n_1 \sigma-m_1 \tau+m_2=0
\end{aligned}
\label{eq:polesPhi}
\end{equation}
where 
\begin{equation}
m_1 \in N \mathbb{Z},\quad n_1 \in \mathbb{Z},\quad j \in 2 \mathbb{Z}+1,\quad m_2, n_2 \in \mathbb{Z}, \quad m_1 n_1+m_2 n_2+\frac{j^2}{4}=\frac{1}{4} .
\end{equation}
In particular, near $z=0$ the function 
$\widetilde{\Phi}_k$ satisfies the following property (\cite[Eq. (C.26)]{Sen:2007qy}):
\begin{equation}\label{eq:Phikz=0pole}
\widetilde{\Phi}_k(\tau,z,\sigma)=-4\pi^2z^2f^{(k)}(N\tau)g^{(k)}(\sigma/N)+O(z^4),
\end{equation}
where 
\begin{equation}\label{eq:fkgkproddef}
\begin{gathered}
f^{(k)}(N \tau)=e^{2 \pi i \widetilde{\alpha} \rho} \prod_{l=1}^{\infty}\left(1-e^{2 \pi i l \tau}\right)^{s_l}, \\
g^{(k)}(\sigma / N)=e^{2 \pi i \widetilde{\gamma} \sigma} \prod_{r=0}^{N-1} \prod_{\substack{k^{\prime} \in \mathbb{Z}+r / N \\
k^{\prime}>0}}\left(1-e^{2 \pi i k^{\prime} \sigma}\right)^{t_r}, \\
s_r=\frac{1}{N} \sum_{s^{\prime}=0}^{N-1} e^{-2 \pi i r s^{\prime} / N} Q_{0, s^{\prime}}=\sum_{s^{\prime}=0}^{N-1} e^{-2 \pi i r s^{\prime} / N}\left(c_0^{\left(0, s^{\prime}\right)}(0)+2 c_1^{\left(0, s^{\prime}\right)}(-1)\right), \\
t_r=\frac{1}{N} \sum_{s=0}^{N-1} Q_{r, s}=\sum_{s=0}^{N-1}\left(c_0^{(r, s)}(0)+2 c_1^{(r, s)}(-1)\right) .
\end{gathered}
\end{equation}
This gives us the residue of the pole of the partition function 
$\widetilde{\Phi}^{-1}_k$ at the double pole at $z=0$ and its $\widetilde{G}$ images. 
The product representation \eqref{eq:Phiprodrep} also gives us the following periodicity property:
\begin{equation}
\widetilde{\Phi}_k(\tau+1,z,\sigma)=\widetilde{\Phi}(\tau,z+1,\sigma)=\widetilde{\Phi}(\tau,z,\sigma+N).
\end{equation} 
Using this, one can write a Fourier expansion for the partition function but because of the double poles one has to expand the functions differently in different regions. This has the consequence that we cannot write a uniform Fourier-Jacobi expansion for the partition function which is valid in all of $\mathbb{H}_2$. \\\\
Let us recall how we extract the single-centered degeneracies from the partition function. Let $\vec{Q}$ and $\vec{P}$ be the electric and magnetic charge vectors, then for $\mathbb{Z}_N$ CHL models, we have 
\begin{equation}
\frac{Q^{2}}{2} \in \frac{1}{N} \mathbb{Z}, \quad \frac{P^{2}}{2} \in \mathbb{Z}, \quad P \cdot Q \in \mathbb{Z}.
\end{equation}
The degeneracy of single-centered black hole microstates of charge $(\vec{Q},\vec{P})$ is then given by  the following contour integral:
\begin{equation}
d^{*}\left(\frac{P^{2}}{2N},Q \cdot P,\frac{NQ^{2}}{2}\right)\equiv \frac{(-1)^{Q\cdot P+1}}{N} \int_{\mathcal{C}} d\tau d\sigma dz\, e^{-2 \pi i\left(\frac{NQ^{2}}{2}\tau+\frac{P^{2}}{2N}\sigma+(Q \cdot P) z\right)} \frac{1}{\widetilde{\Phi}_{k}(\Omega)},
\label{eq 3.2}
\end{equation}
where the contour $\mathcal{C}$ is given as: 
\begin{equation}
\begin{split}
\mathcal{C}:\quad \operatorname{Im}(\tau)=\frac{P^{2}}{N\varepsilon}\equiv\frac{2n}{N\varepsilon}, \quad \operatorname{Im}(\sigma)=\frac{NQ^{2}}{\varepsilon}\equiv\frac{2m}{\varepsilon}, \quad \operatorname{Im}(z)=-\frac{P \cdot Q}{\varepsilon}\equiv-\frac{\ell}{\varepsilon},\\
0 \leq \operatorname{Re}(\tau), \operatorname{Re}(z)<1,~0\leq \operatorname{Re}(\sigma)<N.
\end{split}
\label{eq 3.3}
\end{equation}
Here $\varepsilon>0$ is a small real number and $n,m,\ell\in\mathbb{Z}$. \par The partition function $\widetilde{\Phi}^{-1}_{k}(\Omega)$ as well as the integration measure $d\tau d\sigma dz$ is invariant under a transformation $g_4(a,b,c,d)\in\mathrm{Sp}(2,\Z)$ described in \eqref{eq:Phitildeg4} with $\left(\begin{smallmatrix}
    a&b\\c&d
\end{smallmatrix}\right)\in\Gamma_1(N)$, therefore we have the freedom of transforming the contour $\mathcal{C}$ in \eqref{eq 3.3} and the charges by this transformation without changing the degeneracy. More precisely, we can transform the contour to $\mathcal{C}'$:
\begin{equation}
\begin{split}
\mathcal{C}':\quad \operatorname{Im}(\tau)=\frac{P'^{2}}{N\varepsilon}\equiv\frac{2n'}{N\varepsilon}, \quad \operatorname{Im}(\sigma)=\frac{NQ'^{2}}{\varepsilon}\equiv\frac{2m'}{\varepsilon},\quad \operatorname{Im}(z)=-\frac{P' \cdot Q'}{\varepsilon}\equiv-\frac{\ell'}{\varepsilon},
\end{split}
\label{eq:C(l1)cont}   
\end{equation}
with 
\begin{equation}
    0 \leq \operatorname{Re}(\tau), \operatorname{Re}(z)<1,~0\leq \operatorname{Re}(\sigma)<N,
\end{equation}
where $m',n',\ell'$ can be computed using the action of $g_4(a,b,c,d)$ on $\tau,z,\ell$ using \eqref{eq:Phitildeg4}:
\begin{equation}\label{eq:sdualtrans}
\begin{pmatrix}
\frac{2n'}{N}&\ell'\\\ell'&2m'
\end{pmatrix}\equiv\begin{pmatrix}
a&b\\c&d
\end{pmatrix}\begin{pmatrix}\frac{2n}{N}&\ell\\\ell&2m
\end{pmatrix}\begin{pmatrix}
a&c\\b&d
\end{pmatrix},\quad \begin{pmatrix}
a&b\\c&d
\end{pmatrix}\in\Gamma_1(N).
\end{equation}
This is the S-duality invariance of the single-centered degeneracies:
\begin{equation}
d^*(n/N,\ell,m)=d^*(n'/N,\ell',m').
\end{equation}
From the expressions for $(\frac{2n'}{N},\ell',2m')=(\frac{P'^2}{N},P'\cdot Q',NQ'^2)$, we can read out the transformation laws of the charge vectors $\Vec{Q}$ and $\Vec{P}$: 
\begin{eqnarray}\label{eq:SdonQP}
    \begin{pmatrix}
      \Vec{Q}\\\Vec{P}  
    \end{pmatrix}\to \begin{pmatrix}
      \Vec{Q}'\\\Vec{P}'  
    \end{pmatrix} =\begin{pmatrix}
d&c/N\\bN&a
\end{pmatrix}\begin{pmatrix}
      \Vec{Q}\\\Vec{P}  
    \end{pmatrix}= \begin{pmatrix}
      d\Vec{Q}+\frac{c}{N}\Vec{P}\\bN\Vec{Q}+a\Vec{P}  
    \end{pmatrix}.
\end{eqnarray}
\begin{prop}
The partition function $\widetilde{\Phi}_k^{-1}$ has the following Fourier-Jacobi expansion: 
\begin{equation}
\begin{split}
&\frac{1}{\widetilde{\Phi}_k(\tau,z,\sigma)}=\sum_{m=-\widetilde{\alpha}}^{\infty}\phi_m(\sigma,z)q^{m},\quad q=\mathbf{e}(\tau),\quad \mathrm{Im}(\tau)\to\infty\\&\frac{1}{\widetilde{\Phi}_k(\tau,z,\sigma)}=\sum_{n=-\widetilde{\gamma}N}^{\infty}\psi_n(\tau,z)p^{n/N},\quad p=\mathbf{e}(\sigma),\quad \mathrm{Im}(\sigma)\to\infty,
\end{split} 
\label{eq:FJexpPhikDD'}
\end{equation}
for some functions $\phi_m(\sigma,z)$ and $\psi_n(\tau,z)$.
\label{prop:FJexpPhikDD'}
\begin{proof}
We see from \eqref{eq:Phiprodrep} that near $q=0$, the function $\widetilde\Phi_k$ is proportional to
$q^{\widetilde\alpha}$.
We also know that
$\widetilde\Phi_k$ has (double) zeros at \eqref{eq:polesPhi}.
Therefore, in the $\tau$ plane, the poles of $q^{\widetilde\alpha} \widetilde\Phi_k^{-1}$
are at:
\begin{equation} \label{epoles}
\tau= {n_2 z^2 - j z - n_1\sigma - m_2\over n_2\sigma - m_1}\, .
\end{equation}
Since $\text{Im}(\sigma)>0$,
as long as either $n_2$ or $m_1$ is non-zero, the denominator does not vanish. 
Therefore
the poles occur at finite values of $\tau$, where $q=\mathbf{e}(\tau)$ is non-zero. On the other hand,
if $n_2=0$ and $m_1=0$ then the pole locations given in \eqref{eq:polesPhi} are independent of
$\tau$. Therefore, for fixed $\sigma$ and $z$ away from the zeros of $jz+n_1\sigma+m_2$, the function $q^{\widetilde\alpha}
\widetilde\Phi_k^{-1}$ has poles 
in the $q$ plane at finite distance away from the origin. Furthermore, we see 
from \eqref{eq:Phiprodrep}
that  $q^{\widetilde\alpha}
\widetilde\Phi_k^{-1}$ is a single-valued function of $q$ around $q=0$. Therefore, 
$q^{\widetilde\alpha}
\widetilde\Phi_k^{-1}$ admits a Taylor series expansion in $q$ with $\sigma,z$ dependent
expansion coefficients. Denoting the coefficient of $q^m$ by 
$\phi_{m-\widetilde\alpha}(\sigma,z)$, we  prove the first 
expansion in \eqref{eq:FJexpPhikDD'}.

Similarly, we see from \eqref{eq:Phiprodrep} that near $p=0$, 
$\widetilde\Phi_k$ is proportional to
$p^{\widetilde\gamma}$. On the other hand the zeros of $p^{-\widetilde\gamma}\widetilde
\Phi_k$ in the $\sigma$ plane lie on the subspace \eqref{eq:polesPhi}.
This gives 
\begin{equation}\label{epolesigma}
\sigma= {n_2 z^2 - j z + m_1\tau - m_2\over n_2\tau + n_1}\, .
\end{equation}
Since $\text{Im}(\tau)>0$,
as long as either $n_2$ or $n_1$ is non-zero, the denominator does not vanish. 
Therefore,
the poles occur at finite values of $\sigma$, where $p=\mathbf{e}(\sigma)$ is non-zero. On the other hand
if $n_2=0$ and $n_1=0$ then the pole locations given by \eqref{eq:polesPhi} are independent of
$\sigma$. Therefore, for fixed $\tau$ and $z$ away from the zeros of $m_1\tau-jz-m_2$, the function $p^{\widetilde\gamma}
\widetilde\Phi_k^{-1}$ has poles 
in the $p$ plane at finite distance away from the origin.
Furthermore, we see 
from \eqref{eq:Phiprodrep}
that  $p^{\widetilde\gamma}
\widetilde\Phi_k^{-1}$ is a single-valued function of $p^{1/N}$ around $p=0$. 
Therefore 
$p^{\widetilde\gamma}
\widetilde\Phi_k^{-1}$ admits a Taylor series expansion in $p^{1/N}$ with $\tau,z$ dependent
expansion coefficients. Denoting by $\psi_{m-\widetilde\gamma}(\tau,z)$ the expansion
coefficient of $p^{m/N}$, we get the second 
expansion in \eqref{eq:FJexpPhikDD'}.
\end{proof}
\end{prop}
\subsection{Single-centered degeneracy and mock Jacobi forms} \label{s8.2}
We now show that the Fourier-Jacobi coefficients $\phi_m$ and $\psi_n$ in \eqref{eq:FJexpPhikDD'} are meromorphic Jacobi forms.
\begin{thm}
\label{CHL_thm1}
Let $\phi_m$ and $\psi_n$ be as in Proposition \ref{prop:FJexpPhikDD'}. Then $\phi_m$ is a meromorphic Jacobi form of weight $-k$ and index $m$ with respect to $\Gamma^1(N)\ltimes\mathbb{Z}^2$ with poles at $z\in\mathbb{Z}\sigma+\mathbb{Z}$ and $\psi_n$ is a meromorphic Jacobi form of weight $-k$ and index $n/N$ with respect to $\Gamma_1(N)\ltimes(N\mathbb{Z}\times\mathbb{Z})$ with poles at $z\in N\mathbb{Z}\tau+\mathbb{Z}$.
\label{thm:FJcoePhijactrans}
\begin{proof} 
Let us start by expanding the partition function in $q$ first:
\begin{equation}
\frac{1}{\widetilde{\Phi}_k(\tau,z,\sigma)}=\sum_{m=-\widetilde{\alpha}}^{\infty}\phi_m(\sigma,z)q^{m}.
\end{equation}
From Proposition \ref{prop:Phitransg1}, we have 
\begin{equation}
g^1(a,b,c,d),~g^3(\lambda,\mu)\in\widetilde{G},\quad b\in N\mathbb{Z},~ ab-bc=1,~a,d\equiv 1(\bmod N),\quad \lambda,\mu\in\mathbb{Z}
\label{eq:tildeGg13}
\end{equation}
where 
\begin{equation}
g^1(a,b,c,d):=\begin{pmatrix}
1&0&0&0\\0&a&0&b\\0&0&1&0\\0&c&0&d
\end{pmatrix},\quad g^3(\lambda,\mu):=\begin{pmatrix}
1&\lambda&0&\mu\\0&1&\mu&0\\0&0&1&0\\0&0&-\lambda&1
\end{pmatrix}.
\end{equation}
Now observe that 
\begin{equation}
\begin{split}
&g^1(a,b,c,d)\circ\Omega=\begin{pmatrix}
\tau-\frac{cz^2}{c\sigma+d}&\frac{z}{c\sigma+d}\\\frac{z}{c\sigma+d}&\frac{a\sigma+b}{c\sigma+d}
\end{pmatrix},\\
& g^3(\lambda,\mu)\circ\Omega=\begin{pmatrix}
\lambda^2\sigma+2\lambda z+\tau+\lambda\mu&z+\lambda\sigma+\mu\\z+\lambda\sigma+\mu&\sigma
\end{pmatrix},
\end{split}
\end{equation}
where 
\begin{eqnarray}
    \Omega = \begin{pmatrix}
\tau&z\\z&\sigma\end{pmatrix}\in\mathbb{H}_2.
\end{eqnarray}
Also 
\begin{equation}
\text{det}(C_1\Omega+D_1)=(c\sigma+d),\quad \text{det}(C_3\Omega+D_3)=1
\end{equation}
where we write out the matrices $g^1$ and $g^3$ in the form
\begin{equation}
g^{1}(a,b,c,d)=\begin{pmatrix}
A_1&B_1\\C_1&D_1
\end{pmatrix},\quad g^{3}(\lambda,\mu)=\begin{pmatrix}
A_3&B_3\\C_3&D_3
\end{pmatrix}.
\end{equation}
Since $\widetilde{\Phi}_k^{-1}$ transforms like a Siegel modular form of weight $-k$ for $\widetilde{G}$ we have 
\begin{equation}
\widetilde{\Phi}_k^{-1}(g^1(a,b,c,d)\circ\Omega)=(c\tau+d)^{-k}\,
\widetilde{\Phi}_k^{-1}(\Omega),\quad \widetilde{\Phi}_k^{-1}(g^3(\lambda,\mu)\circ\Omega)=\widetilde{\Phi}_k^{-1}(\Omega).
\label{eq:Phig1g3inv}
\end{equation}
In the limit Im$(\tau)\to\infty$ both sides of \eqref{eq:Phig1g3inv} admits the Fourier-Jacobi expansion as in \eqref{eq:FJexpPhikDD'}. First equation of \eqref{eq:Phig1g3inv} gives, 
\begin{equation}
\begin{split}
(c\sigma+d)^{-k}\sum_{m=-\widetilde{\alpha}}^{\infty}\phi_m(\sigma,z)q^{m}&=\sum_{m=-\widetilde{\alpha}}^{\infty}\phi_m\left(\frac{a\sigma+b}{c\sigma+d},\frac{z}{c\sigma+d}\right)\mathbf{e}\left(m\left(\tau-\frac{cz^2}{c\sigma+d}\right)\right)\\&=\sum_{m=-\widetilde{\alpha}}^{\infty}\mathbf{e}\left(-\frac{mcz^2}{c\sigma+d}\right)\phi_m\left(\frac{a\sigma+b}{c\sigma+d},\frac{z}{c\sigma+d}\right)q^m.
\end{split}
\end{equation}
Comparing coefficients of $q^m$ we obtain
\begin{equation}
\phi_m\left(\frac{a\sigma+b}{c\sigma+d},\frac{z}{c\sigma+d}\right)=(c\sigma+d)^{-k}\mathbf{e}\left(\frac{mcz^2}{c\sigma+d}\right)\phi_m(\sigma,z)
\end{equation}
which is precisely the modular transformation of a Jacobi form. Note that
\begin{equation}
\begin{pmatrix}
a&b\\c&d
\end{pmatrix}\in\Gamma^1(N)
\end{equation}
and hence $\phi_m$ has the correct modular transformation of a Jacobi form of weight $-k$ and index $m$ with respect to $\Gamma^1(N).$
Next the second equation of \eqref{eq:Phig1g3inv} gives, 
\begin{equation}
\begin{split}
\sum_{m=-\widetilde{\alpha}}^{\infty}\phi_m(\sigma,z)q^{m}&=\sum_{m=-\widetilde{\alpha}}^{\infty}\phi_m(\sigma,z+\lambda\sigma+\mu)\,
\mathbf{e}(m(\lambda^2\sigma+2\lambda z+\tau+\lambda\mu))\\&=\sum_{m=-\widetilde{\alpha}}^{\infty}\mathbf{e}(m(\lambda^2\sigma+2\lambda z+\lambda\mu))\, \phi_m(\sigma,z+\lambda\sigma+\mu)\, q^m.
\end{split}
\end{equation}
Again comparing Fourier coefficients, we obtain  
\begin{equation}
\phi_m(\sigma,z+\lambda\sigma+\mu)=\mathbf{e}(-m(\lambda^2\sigma+2\lambda z))\, \phi_m(\sigma,z)
\end{equation}
which is the elliptic transformation. Thus $\phi_m$ transforms like Jacobi forms of weight $-k$ and index $m$ with respect to $\Gamma^1(N)\ltimes\mathbb{Z}^2$. 

Let us now expand the partition function around $\text{Im}(\sigma)=\infty$: 
\begin{equation}
\frac{1}{\widetilde{\Phi}_k(\tau,z,\sigma)}=\sum_{n=-\widetilde{\gamma}N}^{\infty}\psi_n(\tau,z)p^{n/N}.
\end{equation}
From Proposition \ref{prop:Phitransg1}, we see that 
\begin{equation}
\begin{split}
& \widetilde{\Phi}_k\underset{k}{|}g_1(a,b,cN,d)=\widetilde{\Phi}_k,\quad ad-bcN=1,\quad a,d\equiv 1(\bmod~N),\\&\widetilde{\Phi}_k\underset{k}{|}g_3(\lambda N,\mu)=\widetilde{\Phi}_k,\quad \lambda,\mu\in\mathbb{Z},
\end{split}
\end{equation}
where $g_1,g_3$, given in \eqref{eq:g_123def}
has the following action on $\Omega$:
\begin{equation}
\begin{split}
& g_1(a,b,c,d)\circ \Omega=\begin{pmatrix}
\frac{a\tau+b}{c\tau+d}&\frac{z}{c\tau+d}\\\frac{z}{c\tau+d}&\sigma-\frac{cz^2}{c\tau+d}
\end{pmatrix}\\&g_3(\lambda,\mu)\circ\Omega=\begin{pmatrix}
\tau&z+\lambda\tau+\mu\\z+\lambda\tau+\mu&\lambda^2\tau+2\lambda z+\lambda\mu+\sigma
\end{pmatrix}.
\end{split}
\end{equation}
Performing similar calculations as above it follows that $\psi_n(\tau,z)$ transforms as a Jacobi form of weight $-k$, index $n/N$ with respect to $\Gamma_1(N)\ltimes(N\mathbb{Z}\times\mathbb{Z})$. Thus, 
\begin{equation}
\phi_m(\tau,z)\in J_{-k,m}(\Gamma^1(N)),\quad \psi_n(\tau,z)\in J^N_{-k,\frac{n}{N}}(\Gamma_1(N)).
\end{equation}
The poles of $\phi_m$ and $\psi_n$ follows from the general structure \eqref{eq:polesPhi} of poles of $\widetilde{\Phi}_k^{-1}$ after setting $n_2=0$ and $m_1=0$ (for $\phi_m$) or $n_1=0$ (for $\psi_n$).
\end{proof}
\end{thm}
\begin{cor}
\label{CHL_cor1}
Let $\phi_m^F$ and $\psi_n^F$ be finite parts as defined in \eqref{eq 2.1} and \eqref{eq 2.24} respectively of the Fourier-Jacobi coefficients of $\widetilde{\Phi}_k^{-1}$ defined by \eqref{eq:FJexpPhikDD'}. Then $\phi_m^F$ and $\psi_n^F$ are mixed mock Jacobi forms of the type described in Theorem \ref{thm 2.6} and Theorem \ref{thm 2.12} respectively.
\begin{proof}
The poles of $\phi_m$ and $\psi_n$ are at the torsion points $z\in\mathbb{Z}\tau+\mathbb{Z}$ and $z\in(N\mathbb{Z}\sigma+\mathbb{Z})$ respectively. Thus Theorem \ref{thm 2.6} and Theorem \ref{thm 2.12} combined with Theorem \ref{thm:FJcoePhijactrans} give the required result.
\end{proof}
\end{cor}
We shall now give an explicit form of the polar parts $\phi_m^P$ and $\psi_n^P$ of the mock Jacobi forms $\phi_m^F$ and $\psi_n^F$ respectively. To do this we need the coefficient of the Laurent expansion of $\phi_m,\psi_n$ at $z=0$. This can be obtained from  \eqref{eq:Phikz=0pole}. The coefficient of the Laurent expansion of the partition function at $z=0$ is given by
\begin{equation}
\frac{1}{\widetilde{\Phi}_k(\tau,z,\sigma)}\simeq -\frac{1}{4\pi^2}\tilde{f}^{(k)}(\tau)\tilde{g}^{(k)}(\sigma)\frac{1}{z^2}+O(z^{-4}),
\end{equation}  
where $\tilde{f}^{(k)},\tilde{g}^{(k)}$ are the inverses of the modular forms given by \eqref{eq:fkgkproddef}:
\begin{equation}
   \tilde{f}^{(k)}(\tau):=\frac{1}{f^{(k)}(N\tau)},\quad  \tilde{g}^{(k)}(\sigma):=\frac{1}{g^{(k)}(\sigma/N)}.
\end{equation}
Consequently, the residue of $\phi_m$ at $z=0$ is given by (see \eqref{eq:resdoupole}) 
\begin{equation}
E_{0~m}^\phi(\sigma)=\hat{f}^{(k)}(m)\tilde{g}^{(k)}(\sigma),
\qquad D_{0~m}^\phi(\sigma)=0\, ,
\end{equation}
where $\hat{f}^{(k)}(m)$ is the coefficient of $q^m$ in $\tilde{f}^{(k)}(\tau)$. Similarly the residue of $\psi_n$ is given by 
\begin{equation}
E_{0~n/N}^\psi(\sigma)=\hat{g}^{(k)}(n)\tilde{f}^{(k)}(\tau),
\qquad D_{0~n/N}^\psi(\sigma)=0\, ,
\end{equation}
where $\hat{g}^{(k)}(n)$ is the coefficient of $p^{n/N}$ in $\tilde{g}^{(k)}(\sigma)$.  Thus using \eqref{eq 2.15} and \eqref{eq 2.38}, the polar part can be written as  
\begin{equation}
\begin{split}
&\phi_m^P(\sigma,z)={\hat f^{(k)}(m)}{\tilde g^{(k)}(\sigma)}\mathcal{B}^{0}_{m}(\sigma,z),\\& \psi_n^P(\tau,z)=
{\hat g^{(k)}(n)}{\tilde f^{(k)}(\tau)}\widetilde{\mathcal{B}}^{0}_{\frac{n}{N}}(\tau,z),\end{split}    
\end{equation}
where the Appell-Lerch sums $\mathcal{B}^{0}_{m}$ and 
$\widetilde{\mathcal{B}}^{0}_{\frac{n}{N}}$ can be computed from \eqref{eq 2.14} and \eqref{eq 2.27} respectively. This gives
\begin{equation}
\begin{split}
&\phi_m^P(\sigma,z)={\hat f^{(k)}(m)}{\tilde g^{(k)}(\sigma)}\sum_{s\in\mathbb{Z}}\frac{p^{ms^2+s}\zeta^{2ms+1}}{(1-p^{s}\zeta)^2},\\& \psi_n^P(\tau,z)=
{\hat g^{(k)}(n)}{\tilde f^{(k)}(\tau)}\sum_{s\in\mathbb{Z}}\frac{q^{N(ns^2+s)}\zeta^{2ns+1}}{(1-q^{Ns}\zeta)^2}.\end{split}  
\end{equation}
\par
We now prove our main theorem. Let us introduce the following set of triplets of charge. Let $A$ denote the set of integrs $(n,\ell,m)$ 
satisfying 
\begin{equation}
    m,n>0,\quad 4\frac{n}{N}m-\ell^2>0,
\end{equation}
and
\begin{equation}\label{eq:AdefMin}
\hbox{Min}_{m_1\in N\mathbb{Z},\, s\in \mathbb{Z},\, s(s+1)= 0\,  {\rm mod}\, m_1
    }
    \left[\left({m_1\ell+2ms +m\over m}\right)^2 +{m_1^2\over m^2}
    \left(4 {n\over N} m -\ell^2\right)\right] \ge 1\, . 
\end{equation}
Let us also introduce the set $B$ of integers $(n,\ell,m)$ 
satisfying 
\begin{equation}
    m,n>0,\quad 4\frac{n}{N}m-\ell^2>0,
\end{equation}
and
\begin{equation}\label{eq:BdefMin}
    \hbox{Min}_{n_1,s\in\mathbb{Z},\, s(s+1)=0 \, {\rm mod}\, n_1} \left[\left({n_1 N\ell-2ns-n\over n}\right)^2 
    +{n_1^2N^2\over n^2}
    \left(4 {n\over N} m -\ell^2\right)\right] \ge 1\, . 
\end{equation}
We then have the following theorem.
\begin{thm}
\begin{enumerate}[label=(\alph*)]
\item For  $(n,\ell,m)\in A$, 
the single-centered black hole degeneracies with  $NQ^2/2=m,P^2/2=n$ and $Q\cdot P=\ell$ are the Fourier coefficients\footnote{Up to a constant multiple, see \eqref{efinal} below.} of the mock Jacobi form $\phi_m^F$.
\item For $(n,\ell,m)\in B$, 
the single-centered black hole degeneracies with $NQ^2/2=m,P^2/2=n$ and $Q\cdot P=\ell$ are the Fourier coefficients of the mock Jacobi form $\psi_n^F$.
\end{enumerate}
\label{CHL_thm2}
\begin{proof}
Let us start with charges in the set $A$. 
We now want to use the first Fourier-Jacobi expansion in \eqref{eq:FJexpPhikDD'} to perform the contour integral for single-centered degeneracy. This requires us to deform the contour $\mathcal{C}$ to $\mathcal{C}'$ which is defined identical to $\mathcal{C}$ except that Im$(\tau)\to \infty$. To do this we need to show that there are no poles of the partition function between the region bounded by the contour $\mathcal{C}$ and the deformed contour. This is proved in Appendix \ref{sectionb1}. We can hence write
\begin{equation}
d^{*}\left(\frac{n}{N}, \ell,m\right)= \frac{(-1)^{\ell+1}}{N} \int_{\mathcal{C}'} d\sigma dzd\tau  e^{-2 \pi i\left(m\tau+\frac{n}{N}\sigma+\ell z\right)} \sum_{m'=-\widetilde{\alpha}}^{\infty}\phi_{m'}(\sigma,z)q^{m'}.
\end{equation} 
Performing the $\tau$ integral, we get 
\begin{equation}
\begin{split}
d^{*}\left(\frac{n}{N}, \ell,m\right)=\frac{(-1)^{\ell+1}}{N} \int_{\mathcal{C}(n,\ell)} d\sigma dz e^{-2 \pi i\left(\frac{n}{N}\sigma+\ell z\right)}\phi_{m}(\sigma,z)
\end{split}
\end{equation}
where $\mathcal{C}(n,\ell)$ is the projection of $\mathcal{C}$ onto $\sigma$-$z$ plane. 

Now we specify the $z$ integral. Note that the attractor contour \eqref{eq 3.3} implies that 
\begin{equation}
\text{Im}(z)=-\frac{\ell}{2m}\text{Im}(\sigma).
\end{equation}
Moreover we can shift the contour in purely horizontal direction since we do not cross any pole in doing so. So we shift the contour such that on the contour 
\begin{equation}
-\frac{\ell}{2m}\text{Re}(\sigma)\leq\text{Re}(z)\leq-\frac{\ell}{2m}\text{Re}(\sigma)+1. 
\end{equation}
Thus we have
\begin{equation}\label{eq:contintd*phim}
\begin{split}
d^{*}\left(\frac{n}{N}, \ell,m\right)&= \frac{(-1)^{\ell+1}}{N} \int d\sigma e^{-2 \pi i\frac{n}{N}\sigma}\int_{-\frac{\ell}{2m}\sigma}^{-\frac{\ell}{2m}\sigma+1} dz e^{-2 \pi i\ell z}\phi_{m}\left(\sigma,z\right)\\&=\frac{(-1)^{\ell+1}}{N} \int d\sigma e^{-2 \pi i\frac{n}{N}\sigma}e^{2\pi i\ell^2\sigma/4m}h_{\ell}^{(-\frac{\ell}{2m}\sigma)}(\sigma),
\end{split}
\end{equation}
where 
\begin{equation}
h_{\ell}^{(-\frac{\ell}{2m}\sigma)}(\sigma):=e^{-2\pi i\ell^2\sigma/4m}\int_{-\frac{\ell}{2m}\sigma}^{-\frac{\ell}{2m}\sigma+1} dz e^{-2 \pi i\ell z}\phi_{m}(\sigma,z)\, .
\end{equation}
It follows from the discussion at the end of Appendix \ref{sec:poleattcont} that 
$h_{\ell}^{(-\frac{\ell}{2m}\sigma)}(\sigma)$ does not have a pole
in the upper half $\sigma$ plane. Hence the
integration over $\sigma$ can be taken to be over any horizontal line of length $N$.

Now, from \eqref{eq 2.23'}, \eqref{eq 2.24}, and the fact that
$h_{\ell}^{(-\frac{\ell}{2m}\sigma)}(\sigma)$ only depends on 
$\ell\bmod 2m$ (see the discussion below \eqref{eq 2.23}),
we get
\begin{eqnarray}
\phi_m^F(\sigma,z)
&=&\sum_{\ell=0}^{2m-1}h_{\ell}^{(-\frac{\ell}{2m}\sigma)}(\sigma)\vartheta_{m,\ell}(\sigma,z)\nonumber
\\&
=&\sum_{\ell=0}^{2m-1}\sum_{r\in\mathbb{Z}}e^{2\pi i(\ell+2mr)^2\sigma/4m}e^{2\pi i(\ell+2mr) z}h_{\ell}^{(-\frac{\ell}{2m}\sigma)}(\sigma)\nonumber \\
&=& \sum_{\ell\in \mathbb{Z}}e^{2\pi i\ell^2\sigma/4m}e^{2\pi i\ell 
z}h_{\ell}^{(-\frac{\ell}{2m}\sigma)}(\sigma).
\label{eq:phiFfa}
\end{eqnarray}
Using this we can express \eqref{eq:contintd*phim} as
\begin{equation}\label{efinal}
\begin{split}
d^{*}\left(\frac{n}{N}, \ell,m\right)&= \frac{(-1)^{\ell+1}}{N} \int d\sigma \int dz \,
e^{-2 \pi i\frac{n}{N}\sigma}\, e^{-2\pi i \ell z}
\phi_m^F(\sigma,z)\, .
\end{split}
\end{equation}
In this form, the $\sigma$ and $z$ integrals run over any  horizontal lines 
of length $N$ and 1 respectively.
This proves (a).

\par To prove (b), we shall use the fact that for charges in set $B$,
the integration contour $\mathcal{C}$ can be deformed to 
$\mathcal{C}''$, where $\mathcal{C}''$ is defined identical to $\mathcal{C}$ except that Im$(\sigma)\to\infty$. The validity of such a deformation without picking up any residues 
has been proved 
in Appendix \ref{sectionb1}. Then following similar calculation as for (a) we get 
\begin{equation}
d^*(n/N,\ell,m)= \frac{(-1)^{\ell+1}}{N}\int_{\mathcal{C}(m,\ell)} d\tau e^{-2 \pi im\tau}e^{2\pi i\ell^2N\tau/4n}\tilde{h}_{\ell}^{(-\frac{\ell N}{2n}\tau)}(\tau),
\label{eq:d*hllmNn}
\end{equation}
where 
\begin{equation}
\tilde{h}_{\ell}^{(-\frac{\ell N}{2n}\tau)}(\tau):=e^{-2\pi i\tau\frac{\ell^{2}N}{4n}} \int_{-\frac{\ell N}{2n}\tau}^{-\frac{\ell N}{2n}\tau+1} \psi_n(\tau, z) \mathbf{e}(-\ell z) d z.    
\end{equation}
Now, from \eqref{eq 2.23'}, \eqref{eq 2.24}, we get
\begin{equation}
\begin{split}
\psi_n^F(\tau,z)&=\sum_{\ell=0}^{2n-1}\widetilde{h}_{\ell}^{(-\frac{\ell N}{2n}\tau)}(\tau)\widetilde{\vartheta}_{\frac{n}{N},\ell}(\tau,z)\\&=\sum_{\ell=0}^{2n-1}\sum_{r\in\mathbb{Z}}e^{2\pi i(\ell+2nr)^2N\tau/4n}e^{2\pi i(\ell+2nr) z}\widetilde{h}_{\ell}^{(-\frac{\ell N}{2n}\tau)}(\tau)\\&=\sum_{\ell\in\mathbb{Z}}e^{2\pi i\ell^2N\tau/4n}e^{2\pi i\ell z}\widetilde{h}_{\ell}^{(-\frac{\ell N}{2n}\tau)}(\tau).
\end{split}
\label{eq:psiFftilde}
\end{equation}  
Using this we can express \eqref{eq:d*hllmNn} as
\begin{equation}
\begin{split}
d^{*}\left(\frac{n}{N}, \ell,m\right)&= \frac{(-1)^{\ell+1}}{N} \int d\tau \int dz \,
e^{-2 \pi im\tau}\, e^{-2\pi i \ell z}
\psi_n^F(\tau,z)\, ,
\end{split}
\end{equation}
where the $\tau$ and $z$ integrals run over any  horizontal lines 
of length 1.
This proves (b).
\end{proof}
\end{thm}
\begin{remark}Single-centered degeneracies for charges 
belonging to the set $A$
can be obtained by performing contour integrals with integrand $\phi_m(\sigma,z)$ or $\phi_m^F(\sigma,z)$. The former is given by the expression \eqref{eq:contintd*phim}
where the limits on the $z$ integral  depends explicitly on $\ell$. 
On the other hand the latter is given by \eqref{efinal} where the limits of
integration are independent of $\ell$. Analogous statement holds for single-centered degeneracies for charges belonging to the set $B$.
\end{remark}
\noindent It is useful to list some special properties of the set $A$ and $B$ proved
in appendix \ref{app:Cdeform}.
\begin{enumerate}
 \item
The sets $A$ and $B$ are invariant under the elliptic transformations\footnote{This is just the S-duality transformation \eqref{eq:sdualtrans} for the matrix $\begin{pmatrix}1&\lambda_1\\0&1\end{pmatrix}\in\Gamma_1(N)$.} \begin{equation}\label{eq:Adeftrans}(n/N,\ell, m) \mapsto (n/N+m\lambda_1^2+\lambda_1\ell, \ell+2m\lambda_1,m), \qquad \lambda_1\in \mathbb{Z}\,     \end{equation}
and\footnote{This is just the S-duality transformation \eqref{eq:sdualtrans} for the matrix $\begin{pmatrix} 1&0\\\lambda_2N&1\end{pmatrix}\in\Gamma_1(N)$.} \begin{equation}\label{eq:Bdeftrans}(n/N,\ell,m)\mapsto (n/N, \ell+2n\lambda_2,m+Nn\lambda_2^2+N\lambda_2\ell), \qquad \lambda_2\in\mathbb{Z}\,     \end{equation}
respectively. 
\item 
For given $(n,\ell, m)$ it only requires a finite amount of work to determine if it is in set $A$. The second term in \eqref{eq:AdefMin} is bounded by 1 only for a finite number of choices of $m_1$, so we can restrict to this set for testing the inequality given in \eqref{eq:AdefMin}. For a given value of $m_1$ in this set, only a finite number of choices of $s$ can violate this inequality. Once we have exhausted the list of $m_1$ and $s$ that could possibly violate the inequality and found no candidate, we can declare that the corresponding $(n,\ell, m)$ belongs to set $A$. A similar remark holds for set $B$.
\item For given $m$, it requires only a finite amount of work to classify all the orbits under \eqref{eq:Adeftrans} that belong to the set $A$. 
Since \eqref{eq:Adeftrans} can bring $\ell$ to the range $0\le \ell <2m$, we can choose $\ell$ to be in this range. For each of these values of $\ell$, only finite number of allowed values of $n$ can keep the second term in \eqref{eq:AdefMin} below one for $m_1\ne 0$. For each of these choices of $n$ and $m_1$, we then examine the list of $s$ for which the inequality can be violated. This is again a finite operation. For $m_1=0$ the inequality is satisfied for all $n$ and $s$ and no further work is needed. A similar result holds for the set $B$.
\item
It is also shown in Appendix \ref{app:b3} and \ref{app:b4} that the set of charges 
\begin{equation}
\begin{split}
&A_1:=\left\{(n,\ell,m):~ m,n>0,~ 4\frac{n}{N}m-\ell^2\geq \frac{m^2}{N^2}\right\},
\\&A_2:=\left\{(n,\ell,m):~m,n>0,~ \frac{n}{N}\geq m,~0\leq|\ell|\leq 2m\right\},     
\end{split}     
\end{equation}
are subsets of $A$ and 
\begin{equation}
\begin{split}
&B_1:=\left\{(n,\ell,m):~ m,n>0,~ 4\frac{n}{N}m-\ell^2\geq \frac{n^2}{N^2}\right\},
\\&B_2:=\left\{(n,\ell,m):~ m,n>0,~ m\geq\frac{n}{N},~0\leq|\ell|\leq 2n\right\},     
\end{split}     
\end{equation}
are subsets of $B$. Note that $A_1$ is invariant under the elliptic transformation \eqref{eq:Adeftrans} and $B_1$ is invariant under the elliptic transformation \eqref{eq:Bdeftrans}. But $A_2,B_2$ are not invariant under the transformations \eqref{eq:Adeftrans} and \eqref{eq:Bdeftrans} respectively. It follows that the images of $A_2$ under \eqref{eq:Adeftrans} and $B_2$ under \eqref{eq:Bdeftrans} are also subsets of $A$ and $B$ respectively. 
\end{enumerate}

\subsection{Comparison with Dabholkar-Murthy-Zagier}
We have found the set of charges for which the degeneracies are given by the coefficients of the mock Jacobi forms. In particular, for $N=1$ we must recover the result of \cite[Section 11.1]{Dabholkar:2012nd}. From our analysis, it turns out that the italicized statement below eq. (11.17) of \cite{Dabholkar:2012nd} in its most general form is not correct. The general claim in \cite[Section 11.1]{Dabholkar:2012nd} is that the degeneracy for every charge $(n,\ell,m)$ with $n>m$ is given by the corresponding coefficient of \footnote{Note that the mock Jacobi form $\psi_m^F$ of \cite{Dabholkar:2012nd} is denoted as $\phi_m^F$ in this paper. Also the roles of $Q$ and $P$ are interchanged compared to \cite{Dabholkar:2012nd} since we have $Q^2/2=m,P^2/2=n$.} $\phi_m^F$. We will give some counterexamples to demonstrate that this is not true. In particular we will find charges $(n,\ell,m)$ with $n>m$ which do not belong to the set $A$. \par Consider for example the charge $(n,\ell,m)=(10,14,5)$. 
It can easily be checked that $(n,\ell,m)=(10,14,5)$ does not belong to the set $A$,
since the choice $m_1=1,s=-2$
violates \eqref{eq:AdefMin}.
This can also be seen more explicitly as follows. For the choice of charges given
above, the attractor contour is at 
\begin{eqnarray}
\text{Im}(\tau)=\frac{20}{\varepsilon},\quad \text{Im}(\sigma)=\frac{10}{\varepsilon},\quad \text{Im}(z)=-\frac{14}{\varepsilon}.    
\end{eqnarray}
To express the degeneracy for this charge as the coefficient of $\phi_5^F$, we need to deform the contour from $\text{Im}(\tau)=\frac{20}{\varepsilon}$ to $\text{Im}(\tau)\to i\infty$. But using \eqref{eq:polesPhi} we see that there is a pole at $ \text{Im}(\tau)=22/\varepsilon$ corresponding to $n_2=0, m_1=1, n_1=-2, j=-3$. Thus this deformation is not allowed. More precisely, wall crossing formula gives a nonzero jump in the degeneracy if we cross this pole. Let us explicitly demonstrate this. If the single-centered degeneracy for charges $(n,\ell,m)$ had been given by the Fourier expansion coefficients of $\phi^F_m$, then the degeneracy $d^*(10,14,5)$ would have been given by the coefficient of $p^{10}\zeta^{14}$ of $\phi_5^F(\sigma,z)$ where $p=e^{2\pi i\sigma}, \zeta=e^{2\pi i z}$ as before. By elliptic invariance \eqref{eq:Adeftrans}, this is same as the coefficient of $p\zeta^4$. By explicit computation on Mathematica, we found the coefficient to be $2023536$. Thus we find \footnote{The first sign of trouble comes from the negative sign, since it violates the positivity conjecture of \cite{Sen:2011ktd}.} 
\begin{equation}
    d^*_{\text{DMZ}}(10,14,5)=-2023536.
\end{equation}
On the other hand, following the prescription in \cite{Sen:2011ktd}, $d^*(10,14,5)$ can be computed by first bringing the charge in the range given in Eq. (3.7) of \cite{Sen:2011ktd} and then computing the corresponding coefficient of the partition function. By the S-duality transformation (which leaves the single-centered degeneracies invariant), 
the charge can be brought to 
\begin{eqnarray}
    (10,14,5)\mapsto (1,0,1).
\end{eqnarray}
From \cite[Table 1]{Sen:2011ktd}, we find 
\begin{eqnarray}\label{eq:d*10145}
 d^*(10,14,5)= d^*(1,0,1)=50064.   
\end{eqnarray}
Let us now explicitly calculate the jump in the degeneracy across the wall at $\text{Im}(\tau)=\frac{22}{\varepsilon}$. The wall crossing formula for the decay $(\vec{Q},\vec{P})\to (\vec{Q}_1,\vec{P}_1)+(\vec{Q}_2,\vec{P}_2)$ across a wall of marginal stability 
is given by \cite[Eq. (5.5.17)]{Sen:2007qy} 
\begin{equation}
    (-1)^{\vec{Q}_1 \cdot \vec{P}_2-\vec{Q}_2 \cdot \vec{P}_1+1}\left|\vec{Q}_1 \cdot \vec{P}_2-\vec{Q}_2 \cdot \vec{P}_1\right| d_{\text {half}}(\vec{Q}_1, \vec{P}_1) d_{\text {half}}(\vec{Q}_2, \vec{P}_2),
\end{equation}
where $d_{\text {half}}(\vec{Q}_i, \vec{P}_i)$ denotes the number of bosonic minus fermionic half BPS supermultiplets carrying charges $(\vec{Q}_i, \vec{P}_i)$. The pole  at \footnote{This corresponds to a pole of the partition function with $n_2=0$ in \eqref{eq:polesPhi}.}
\begin{equation}
    n_1 \sigma+j z-m_1 \tau+m_2=0,
\end{equation}
corresponds to the wall of marginal stability 
associated with the
decay (see for example Eq. (2.21), Eq. (2.22) of \cite{Mandal:2010cj}) \footnote{The convention used in \cite{Mandal:2010cj} is related to the current convention by an exchange of $\rho$ and $\sigma$.}:
\begin{equation} \label{eq:edecay}
    (\vec{Q},\vec{P}) \to (\alpha \vec{Q}+\beta \vec{P}, \gamma \vec{Q}+\delta \vec{P})
    + ( (1-\alpha) \vec{Q}-\beta \vec{P}, -\gamma \vec{Q}+(1-\delta)\vec{P})\, ,
\end{equation}
where
\begin{equation}\label{eq:edecayints}
    \gamma=n_1, \qquad \beta=m_1, \qquad \delta={1-j\over 2}, \qquad 
    \alpha={j+1\over 2}\, .
\end{equation}
For $m_1=1,n_1=-2,j=-3$ the decay is 
\begin{equation}
 (\vec{Q},\vec{P}) \to (\vec{P}-\vec{Q},2(\vec{P}-\vec{Q}))+(2\vec{Q}-\vec{P},2\vec{Q}-\vec{P}).   
\end{equation}
By S-duality, the charges $(\vec{P}-\vec{Q},2(\vec{P}-\vec{Q}))$ and $(2\vec{Q}-\vec{P},2\vec{Q}-\vec{P}))$  can be transformed to purely magnetic charges $(0,\vec{P}-\vec{Q})$ and $(0,2\vec{Q}-\vec{P})$ respectively. 
The degeneracies $d_{\text {half}}(\vec{Q}_i, \vec{P}_i)$ can thus be computed using the coefficients of the reciprocal $\Delta^{-1}$ of the Ramanujan's discriminant function. To compute these degeneracies for our triplet of charges $(n,\ell,m)=(10,14,5)$, we compute:
\begin{equation}
\begin{split}
&\frac{(\vec{P}-\vec{Q})^2}{2} =\frac{Q^2+P^2-2Q\cdot P}{2}=n+m-\ell
=1,\\&\frac{(2\vec{Q}-\vec{P})^2}{2} =\frac{4Q^2+P^2-4Q\cdot P}{2}=4m+n-2\ell
=2.       
\end{split}
\end{equation}
Thus $d_{\text {half}}(\vec{Q}_1, \vec{P}_1)$ and $d_{\text {half}}(\vec{Q}_2, \vec{P}_2)$ are the coefficients of $q^2$ and $q^3$ respectively of $\Delta(q)^{-1}$. Explicitly, we have 
\begin{eqnarray}
d_{\text {half}}(\vec{Q}_1, \vec{P}_1)=324,\quad d_{\text {half}}(\vec{Q}_2, \vec{P}_2)= 3200.    
\end{eqnarray}
Moreover, we have 
\begin{eqnarray}
(-1)^{\vec{Q}_1 \cdot \vec{P}_2-\vec{Q}_2 \cdot \vec{P}_1+1}\left|\vec{Q}_1 \cdot \vec{P}_2-\vec{Q}_2 \cdot \vec{P}_1\right|=-2.    
\end{eqnarray}
Thus the jump is given by 
\begin{eqnarray}
    d_{\text{jump}}(10,14,5)=-2\cdot 324\cdot 3200=-2073600.
\end{eqnarray}
From this, we see that 
\begin{eqnarray}
 d^*_{\text{DMZ}}(10,14,5)-d_{\text{jump}}(10,14,5)=-2023536+2073600=50064    
\end{eqnarray}
which matches with the degeneracy \eqref{eq:d*10145}.\par Another such example is the charge $(n,\ell,m)=(13,19,7)$ where we have nonzero jump in the degeneracy.

It is easy to check that neither (10,14,5) nor (13,19,7) belong to the set $B$. Thus, they provide examples of charges that cannot be directly related to the Fourier coefficients of mock Jacobi forms.

Finally, we note that while the coefficients of mock modular forms do not capture single centered black hole degeneracies for these charges, they do capture the index in a different chamber of the moduli space. The precise description of this chamber can be read out from the prescription given in \cite{Cheng:2007ch} relating the location of the integration contour to the asymptotic values of the closed string moduli. As we change the asymptotic moduli from this value to the attractor values where only single centered black holes exist, we encounter a codimensional one wall of marginal stability across which the index jumps\cite{Sen:2007vb,Dabholkar:2006bj}.

\subsection{Mock modularity in $\mathbb{Z}_M\times\mathbb{Z}_N$ orbifolds CHL models}
Finally, we comment on the validity of the above prescription to get single-centered degeneracy for $\mathbb{Z}_M\times\mathbb{Z}_N$ orbifolds of Type IIB superstring theory compactified on $\mathcal{M}\times S^1\times\widetilde{S}^1$ where $\mathcal{M}=K3$ or $T^4$ and $\mathbb{Z}_M$ acts on $S^1$ by shifts of $1/M$ units and $\mathbb{Z}_N$ acts on $\widetilde{S}^1$ by shifts of $1/N$ units, besides respective actions on $K3$. 
We can determine the S-duality group by demanding that they preserve the vectors $(1/M,0)$ and $(0,1/N)$ on the fundamental parallelogram defining the torus $T^2$ in the compactification space $\mathcal{M}\times T^2$. It is easy to see that under the action 
\begin{equation}
    \begin{pmatrix}
        x\\y
    \end{pmatrix}\mapsto \begin{pmatrix}
        a&b\\c&d
    \end{pmatrix}\begin{pmatrix}
        x\\y
    \end{pmatrix},\quad \begin{pmatrix}
        a&b\\c&d
    \end{pmatrix}\in\mathrm{SL}(2,\mathbb{Z})
\end{equation}
the subgroup of $\mathrm{SL}(2,\mathbb{Z})$ which leaves  $(1/M,0)$ and $(0,1/N)$ invariant is given by 
\begin{equation}
    \Gamma(M,N):=\left\{\begin{pmatrix}
        a&b\\c&d
    \end{pmatrix}\in\mathrm{SL}(2,\mathbb{Z}):\begin{array}{c}
a \equiv 1(\bmod~M),~~ b \equiv 0(\bmod ~M) \\
c \equiv 0(\bmod~ N),~~ d \equiv 1(\bmod~N)
\end{array} \right\}.
\end{equation}
It can be seen that $\Gamma(M,N)$ is a congruence subgroup of level $MN$.\par  The partition function $\Phi^{(M,N)}(\tau,z,\sigma)$ and the prescription for the calculation of single-centered degeneracy appeared in \cite{Sen:2010ts}. The partition function is again a Siegel modular form for some subgroup $G'$ of $\mathrm{Sp}(2,\mathbb{Z})$ that contains the S-duality group. 
Expanding the partition function in $p=e^{2\pi i\sigma}$ first, we obtain a family of meromorphic Jacobi forms of rational index $n/M$ for $\Gamma(M,N)\ltimes (M\mathbb{Z}\times\mathbb{Z})$. The location of poles of these meromorphic Jacobi forms are again at $M\mathbb{Z}\sigma+\mathbb{Z}$. One can then apply the general formalism of Section \ref{subsec:dmzmockjacform} and Section \ref{rational_index} to write down the polar and finite parts of these meromorphic Jacobi forms. The finite part is again a mock Jacobi form. Similarly, one can expand in $q=e^{2\pi i\tau}$ first to obtain mock Jacobi forms for the group $S\Gamma(M,N)S^{-1}\ltimes(N\mathbb{Z}\times\mathbb{Z})=\Gamma(N,M)\ltimes(N\mathbb{Z}\times\mathbb{Z})$ with index $m/N$, where 
\begin{equation}
    S=\begin{pmatrix}
        0&-1\\1&0
    \end{pmatrix}\in\mathrm{SL}(2,\mathbb{Z}).
\end{equation}
\noindent\textbf{Acknowledments.} We would like to thank Nabamita Banerjee and Suvankar Dutta for discussions and collaboration at the initial stages of this work. We would also like to thank Atish Dabholkar and Sameer Murthy for useful discussions. R.K.S. would like to thank IISER Bhopal and ICTS Bangalore for hospitality where part of this work was done. R.K.S. would also like to thank Madhav Sinha and Spencer Stubbs for help with Mathematica code. Finally we thank the anonymous referees for comments which led to improvements in the paper.  The work of R.K.S. is supported by the US Department of Energy under grant DE-SC0010008. The work of A.B. was supported by grant MTR/2019/000582 from the SERB, Government of India. A.B. would like to thank ICTS Bangalore for hospitality where part of this work was done. 
The work of A.S. was supported by ICTS-Infosys Madhava 
Chair Professorship, the J. C. Bose fellowship of the Department of Science
and Technology, India and the Department of Atomic Energy, Government of India, under project no. RTI4001.
\begin{appendices}
\section{Transformation properties of $\widetilde{\Phi}_k$}\label{app:Phiktrans}
In this appendix, we prove some transformation properties of the reciprocal of the partition function given by
$\widetilde{\Phi}_k$. The Siegel modular form $\widetilde{\Phi}_k$ is related to another Siegel modular form $\Phi_k$ as \footnote{We have ignored a constant in the relation which will not be important for the discussion below.} follows \cite[Eq. (C.21)]{Sen:2007qy}
\begin{equation}
\widetilde{\Phi}_k(\tau,z,\sigma)=(N\tau)^{-k}\Phi_k\left(\frac{\tau\sigma-z^2}{N\tau},\frac{\tau\sigma-z^2+z}{N\tau},\frac{\tau\sigma-(z-1)^2}{N\tau}\right).
\label{eq:PhitildePhirel}
\end{equation}
$\Phi_k$ is a Siegel modular form of weight $k$ with respect to a subgroup $G$ of $\mathrm{Sp}(2,\mathbb{Z})$ generated by
\begin{equation}\label{eq:g_123def}
\begin{aligned}
g_{1}(a, b, c, d) & \equiv\left(\begin{array}{cccc}
a & 0 & b & 0 \\
0 & 1 & 0 & 0 \\
c & 0 & d & 0 \\
0 & 0 & 0 & 1
\end{array}\right), ~~a d-b c=1, \quad c=0 \bmod N, \quad a, d=1 \bmod N,\\
g_{2} & \equiv\left(\begin{array}{cccc}
0 & 1 & 0 & 0 \\
-1 & 0 & 0 & 0 \\
0 & 0 & 0 & 1 \\
0 & 0 & -1 & 0
\end{array}\right), \\
g_{3}(\lambda, \mu) & \equiv\left(\begin{array}{cccc}
1 & 0 & 0 & \mu \\
\lambda & 1 & \mu & 0 \\
0 & 0 & 1 & -\lambda \\
0 & 0 & 0 & 1
\end{array}\right),~~\lambda, \mu \in \mathbb{Z}.
\end{aligned}
\end{equation} 
It follows from \eqref{eq:PhitildePhirel}, 
\eqref{eq:g_123def} that these Siegel modular forms satisfy the following periodicity properties \cite{Sen:2011ktd}:
\begin{equation}\label{eq:Phitperiodtsz}
\begin{split}
\Phi_k(\tau+1,z,\sigma)=\Phi_k(\tau,z+1,\sigma)=\Phi_k(\tau,z,\sigma+1)\\\widetilde{\Phi}_k(\tau+1,z,\sigma)=\widetilde{\Phi}(\tau,z+1,\sigma)=\widetilde{\Phi}(\tau,z,\sigma+N).
\end{split}
\end{equation}
One can use the relation \eqref{eq:PhitildePhirel} to find the subgroup $\widetilde{G}$ from $G$. We will content ourselves by proving certain transformation properties of $\widetilde{\Phi}_k$ used in Section \ref{sec:mockjacchl}. We begin by noting that 
$\widetilde{G}$ has a subgroup consisting of the elements \cite[Section 5]{David:2006ud}
\begin{equation}
g_4(a,b,cN,d):=\begin{pmatrix}
a&b&0&0\\cN&d&0&0\\0&0&d&-cN\\0&0&-b&a
\end{pmatrix},\quad ad-bcN=1,\quad a,d\equiv 1(\bmod~N).
\end{equation}
This means that 
\begin{equation}
\widetilde{\Phi}_k\underset{k}{|}g_4(a,b,cN,d)=\widetilde{\Phi}_k.
\label{eq:Phitildeg4}
\end{equation}
This was used for proving the S-duality invariance of dyon degeneracy in Section \ref{sec:chl_partition_function}.
\begin{prop}
We have that 
\begin{equation}
\begin{split}
&\widetilde{\Phi}_k\underset{k}{|}g_1(a,b,cN,d)=\widetilde{\Phi}_k\\
&\widetilde{\Phi}_k\underset{k}{|}g^1(a,bN,c,d)=\widetilde{\Phi}_k,\quad ad-bcN=1,\quad a,d\equiv 1(\bmod~N),
\end{split}
\end{equation}
and 
\begin{equation}
\begin{split}
&\widetilde{\Phi}_k\underset{k}{|}g_3(\lambda N,\mu)=\widetilde{\Phi}_k\\&\widetilde{\Phi}_k\underset{k}{|}g^3(\lambda,\mu)=\widetilde{\Phi}_k,\quad \lambda,\mu\in\mathbb{Z},
\end{split}
\end{equation}
where 
\begin{equation}
\begin{split}
g^1(a,b,c,d):=\begin{pmatrix}
1&0&0&0\\0&a&0&b\\0&0&1&0\\0&c&0&d
\end{pmatrix},\quad g^3(\lambda,\mu):=\begin{pmatrix}
1&\lambda&0&\mu\\0&1&\mu&0\\0&0&1&0\\0&0&-\lambda&1
\end{pmatrix}.
\end{split}
\end{equation}
\label{prop:Phitransg1}
\begin{proof}
Since $\Phi_k$ transforms as a Siegel modular form with respect to $G$, we can write 
\begin{equation}
\begin{split}
&\Phi_{k}\underset{k}{|}g_1(a,b,cN,d)=\Phi_{k},\quad ad-bcN=1,\quad a,d\equiv 1(\bmod~N);\\&\Phi_{k}\underset{k}{|}g_2=\Phi_{k};\\&\Phi_{k}\underset{k}{|}g_3(\lambda,\mu)=\Phi_{k},\quad \lambda,\mu\in\mathbb{Z}.
\end{split}
\label{eq 4.4}
\end{equation}
We begin by simplifying \eqref{eq:PhitildePhirel}. Using $\Phi_{k}\underset{k}{|}g_3(-1,0)=\Phi_{k}$ we get 
\begin{equation}
   \Phi_{k}(\tau,z-\tau,\sigma+\tau-2z)=\Phi_{k}(\tau,z,\sigma). 
\end{equation}
Using this along with \eqref{eq:PhitildePhirel}, we get
\begin{equation}\label{eq:PhitPhirelsim}
 \widetilde{\Phi}_k(\tau,z,\sigma)=(N\tau)^{-k}\Phi_k\left(\frac{\sigma}{N}-\frac{z^2}{N\tau},\frac{z}{N\tau},-\frac{1}{N\tau}\right).   
\end{equation}
Now put 
\begin{equation}
W_N:=\frac{1}{\sqrt{N}}\begin{pmatrix}
0&N&0&0\\1&0&0&0\\0&0&0&1\\0&0&N&0
\end{pmatrix},\quad\mu_N:=\left(\begin{array}{cccc}
1 & 0 & 0 & 0 \\
0 & 0 & 0 & -\frac{1}{N} \\
0 & 0 & 1 & 0 \\
0 & N & 0 & 0
\end{array}\right).
\end{equation}
Note that 
\begin{equation}
\begin{split}
(g_2g_1(0,-1,1,0)W_Ng_2)\circ \Omega = \left(\begin{array}{cr}
\frac{\sigma }{N}-\frac{z ^2}{\tau N}&\frac{z }{\tau N}\\ \; &\;\\ \frac{z }{\tau N}&-\frac{1}{\tau N}
\end{array}\right)
\end{split}
\end{equation}
Hence using \eqref{eq:PhitPhirelsim} we have that 
\begin{equation}
\widetilde{\Phi}_k(\Omega)=\left(\Phi_k\underset{k}{|}g_2g_1(0,-1,1,0)W_Ng_2\right)(\Omega).
\label{eq 4.5}
\end{equation}
Since $g_2\in G$, we can write this as  
\begin{equation}\label{eq:Phislash=Phiti}
\widetilde{\Phi}_k(\Omega)=\left(\left(\Phi_k\underset{k}{|}g_2\right)\underset{k}{|}g_1(0,-1,1,0)W_Ng_2\right)(\Omega)=\left(\Phi_k\underset{k}{|}g_1(0,-1,1,0)W_Ng_2\right)(\Omega),
\end{equation}
where we used \eqref{eq 4.4} in the last equality.
Moreover, we have that 
\begin{equation}
g_1(0,-1,1,0)W_N=W_N\mu_N.
\end{equation} 
This implies that
\begin{equation}
\widetilde{\Phi}_k(\Omega)=\left(\Phi_k\underset{k}{|}W_N\mu_Ng_2\right)(\Omega).
\label{eq 4.6}
\end{equation} 
It is easy to see that
\begin{equation}
\begin{split}
g_2g^1(a,b,c,d)=g_1(a,b,c,d)g_2,\\
\mu_Ng_1(a,b,c,d)=g_1(a,b,c,d)\mu_N,\\
W_Ng_1(a,bN,c,d)=g_2^{-1}g_1(a,b,cN,d)g_2W_N;
\end{split}
\label{eq 4.7}
\end{equation}
and 
\begin{equation}
\begin{split}
g_2g^3(\lambda,\mu)=g_3(-\lambda,-\mu)g_2,\\
W_Ng_3(-\lambda,-\mu)=g^3(-N\lambda,-\mu)W_N,\\
g_1(0,-1,1,0)g^3(0,-\mu)=g_3(\mu,0)g_1(0,-1,1,0).
\end{split}
\label{eq 4.8}
\end{equation}
Using  \eqref{eq 4.6}, \eqref{eq 4.7}, and  \eqref{eq 4.4}, we get 
\begin{equation}
\widetilde{\Phi}_k\underset{k}{|}g^1(a,bN,c,d)=\Phi_k\underset{k}{|} W_N\mu_Ng_2g^1(a,bN,c,d)=\widetilde{\Phi}_k,\quad ad-bcN=1,\quad a,d\equiv 1(\bmod ~N).
\end{equation}
Also using \eqref{eq:Phislash=Phiti}, \eqref{eq 4.8} and \eqref{eq 4.4}, we get,
\begin{equation}
\widetilde{\Phi}_k\underset{k}{|}g^3(0,\mu)=\Phi_k\underset{k}{|}g_1(0,-1,1,0) W_Ng_2g^3(0,\mu)=\widetilde{\Phi}_k.
\label{eq 4.9}
\end{equation}
Moreover we have 
\begin{equation}\label{eq:g^3g_1g_4rellm}
g^3(\lambda,\mu)=g_1(1,-\lambda\mu,0,1)g_4(1,\lambda,0,1)g^3(0,\mu).
\end{equation}
$g_1(1,-\lambda\mu,0,1)$ leaves $\widetilde{\Phi}_k$ invariant as a consequence of the $\tau\to\tau+1$ symmetry described in the second line of \eqref{eq:Phitperiodtsz}. The other elements on the right hand side of \eqref{eq:g^3g_1g_4rellm} have already been shown to leave  $\widetilde{\Phi}_k$ invariant. Therefore, we have

\begin{equation}
\widetilde{\Phi}_k\underset{k}{|}g^3(\lambda,\mu)=\widetilde{\Phi}_k.
\label{eq 4.11}
\end{equation}
We can also check that 
\begin{equation}
\begin{split}
g_2g_1(a,b,c,d)=g^1(a,b,c,d)g_2,\\
\mu_Ng^1(a,b,c,d)=g^1(d,-c/N^2,-bN^2,a)\mu_N,\\
W_Ng^1(d,-c/N^2,-bN^2,a)=g_1(d,-c/N,-bN,a)W_N;\\
\end{split}
\label{eq 4.7'}
\end{equation}
and 
\begin{equation}
\begin{split}
g_2g_3(\lambda,\mu)=g^3(-\lambda,-\mu)g_2,\\
W_Ng^3(-N\lambda,-\mu)=g_3(-\lambda,-\mu)W_N,\\
g_1(0,-1,1,0)g_3(\lambda,\mu)=g_3(-\mu,\lambda)g_1(0,-1,1,0).
\end{split}
\label{eq 4.8'}
\end{equation}
These can be used to prove the remaining parts of the statement following the same procedure.
\end{proof}
\end{prop}

\section{Deformation of the contour $\mathcal{C}$ to $\mathcal{C}'$ and $\mathcal{C}''$ without crossing a pole}\label{app:Cdeform}
In this appendix, we determine under what condition the attractor
contour $\mathcal{C}$ can be deformed to the contour
$\mathcal{C}'$ at $\tau=i\infty$ with fixed $\text{Im}(\sigma)$ and
$\text{Im}(z)$ and to the contour
$\mathcal{C}''$ at $\sigma=i\infty$ with fixed $\text{Im}(\tau)$ 
and
$\text{Im}(z)$ without crossing a pole. We shall focus on charges $(\vec{Q},\vec{P})$ satisfying $Q^2P^2-(Q\cdot P)^2>0$, since single-centered black holes exist only for such charges. For $P^2=2n,Q^2=\frac{2m}{N},Q\cdot P=\ell$, this translates to 
\begin{eqnarray}
    \frac{4mn}{N}-\ell^2>0.
\end{eqnarray}
\subsection{$\mathcal{C}\longrightarrow\mathcal{C}',\mathcal{C}''$} \label{sectionb1}
Let us first consider the deformation from  $\mathcal{C}\longrightarrow\mathcal{C}'$. \par
Let $A$ be the set of triples $(n, \ell, m) \in \mathbb{Z}^3$ satisfying $ m,n>0,\left(4\frac{n}{N}m- \ell^2\right)> 0$ and 

\begin{equation}
    \hbox{Min}_{m_1\in N\mathbb{Z},\, s\in \mathbb{Z},\, s(s+1)= 0\,  {\rm mod}\, m_1
    }
    \left[\left({m_1\ell+2ms +m\over m}\right)^2 +{m_1^2\over m^2}
    \left(4 {n\over N} m -\ell^2\right)\right] \ge 1\, . 
    \label{ediapheq3}
\end{equation}
We have the following result. 

\begin{prop}
If $(n, \ell, m) \in A$, then the partition function $\widetilde{\Phi}_k^{-1}$ does not have poles in the region bounded by the contours $\mathcal{C}$ and $\mathcal{C}^{\prime}$.
\end{prop}
\begin{proof}
We have seen in \eqref{eq:polesPhi} that the general structure of poles of the partition function is given by 
\begin{equation}
\begin{aligned}
& n_2\left(\sigma \tau-z^2\right)+j z+n_1 \sigma-m_1 \tau+m_2=0
\end{aligned}
\label{eq:polesPhiapp}
\end{equation}
where 
\begin{equation} \label{econtraintmn}
m_1 \in N \mathbb{Z},\quad n_1 \in \mathbb{Z},\quad j \in 2 \mathbb{Z}+1,\quad m_2, n_2 \in \mathbb{Z}, \quad m_1 n_1+m_2 n_2+\frac{j^2}{4}=\frac{1}{4} .
\end{equation} 
We also saw in \eqref{eq 3.3} that on the contour $\mathcal{C}$ we have
\begin{equation} \label{eimaginaryparts}
   \text{Im}(\tau) ={2n\over N\varepsilon}, \qquad
   \text{Im}(\sigma) ={2m\over \varepsilon}, \qquad
   \text{Im}(z) =-{\ell\over \varepsilon}\, ,
\end{equation}
and the real parts of the variables are bounded. We first claim that if $(\tau,z,\sigma)$ is a pole of the partition function, then we must have
\begin{equation}
n_1 \sigma-m_1 \tau+j z+m_2=0,
\label{eq:poln2=0}
\end{equation}
where
\begin{equation}
m_1 \in N \mathbb{Z}, \quad n_1, m_2 \in \mathbb{Z}, \quad j \in 2 \mathbb{Z}+1, \quad m_1 n_1+\frac{j^2}{4}=\frac{1}{4} .
\label{eq:polmnj}
\end{equation}
To see this, note that for sufficiently small $\varepsilon$, the dominant contribution to the LHS of \eqref{eq:polesPhiapp} comes from the term 
\begin{equation}
-n_2\left(\text{Im}(\sigma)\text{Im}(\tau)-\text{Im}(z)^2\right).
\end{equation} 
Now since $\text{Im}(\sigma)\text{Im}(\tau)-\text{Im}(z)^2>0$ in $\mathbb{H}_2$, the coefficient of $n_2$ in \eqref{eq:polesPhiapp} does not vanish. For sufficiently small $\varepsilon$ this term dominates in the region between $\mathcal{C}$ and $\mathcal{C}'$, and \eqref{eq:polesPhiapp} is not satisfied if $n_2\neq 0$. This means that if we want to look for poles in the region bounded by $\mathcal{C}$ and $\mathcal{C}'$ we must take $n_2=0$. So the only relevant poles are given by \eqref{eq:poln2=0}.\par  
Taking imaginary part of \eqref{eq:poln2=0} gives
\begin{equation} 
\begin{split}
& n_1 \operatorname{Im}(\sigma)-m_1 \operatorname{Im}(\tau)+j \operatorname{Im}(z)=0.
\end{split}
\end{equation}
Since the imaginary parts of 
$\sigma$ and $z$ are fixed at their values on the attractor contour 
given in \eqref{eimaginaryparts}, we get 
\begin{eqnarray}\label{eimaginarypart}
\operatorname{Im}(\tau) =\frac{2 m n_1-j \ell}{m_1 \varepsilon} =  \frac{2 m}{\varepsilon}\left(\frac{n_1-j \frac{\ell}{2 m}}{m_1}\right).   
\end{eqnarray}
Comparing \eqref{eimaginarypart} with  \eqref{eimaginaryparts} we see that, in order that there are no poles between $\mathcal{C}$ and
$\mathcal{C}'$, we must have
\begin{equation}
\frac{n_1-j \frac{\ell}{2 m}}{m_1}\le {n\over Nm}\, ,
\label{eq:mainineqalt}
\end{equation}  
for all allowed values of $j$, $n_1$ and $m_1$, so that all the relevant poles
occur for $\operatorname{Im}(\tau)\le 2n/(N\varepsilon)$. Special care is
needed to address the situation when the inequality in \eqref{eq:mainineqalt}
is replaced by an equality. In this case there is a pole on the attractor
contour, which signals that the attractor point is on a wall of marginal
stability. However the jump in the index across this wall vanishes, reflecting
that $\widetilde\Phi_k^{-1}$ has vanishing residue at the pole. This will be
demonstrated in Section \ref{sec:poleattcont}.  

Using \eqref{eq:polmnj} and $n_2=0$, we get
\begin{equation}
\frac{n_1}{m_1}=\frac{1-j^2}{4 m_1^2}
\end{equation}
Substituting this in \eqref{eq:mainineqalt}, the inequality becomes
\begin{equation}\label{ediapheq'}
\frac{1-j^2}{4 m_1^2}-\frac{j \frac{\ell}{2 m}}{m_1}-{n\over Nm}\le 0.
\end{equation}
By multiplying this by $4m_1^2$ we can express this as
\begin{equation}
    \left(j + \ell\, {m_1\over m}\right)^2 +{m_1^2\over m^2}
    \left(4 {n\over N} m -\ell^2\right) \ge 1\, .
    \label{ediapheq}
\end{equation}
This condition must hold for all
$m_1\in N\mathbb{Z}$, 
$j\in2\mathbb{Z}+1$,
$m_1|(j^2 -1)/4$. Writing $j=2s+1$ with $s\in\mathbb{Z}$ we can
express this as
\begin{equation}
    \hbox{Min}_{m_1\in N\mathbb{Z},\, s\in \mathbb{Z},\, s(s+1)= 0\,  {\rm mod}\, m_1
    }
    \left[\left({m_1\ell+2ms +m\over m}\right)^2 +{m_1^2\over m^2}
    \left(4 {n\over N} m -\ell^2\right)\right] \ge 1\, . 
    \label{ediapheq3}
\end{equation}
This is invariant under the transformation
\begin{equation} \label{einvariance}
(n/N,\ell, m) \mapsto (n/N+m\lambda_1^2+\lambda_1\ell, \ell+2m\lambda_1,
m), \qquad \lambda_1\in \mathbb{Z}\, ,
\end{equation}
since the term inside the square bracket as well as the conditions on
$m_1, s$ are invariant under this transformation as long as we transform
$s$ to $s-m_1\lambda_1$. This proves the desired result.
\end{proof}
\noindent Next we shall consider the deformation from $\mathcal{C}$ to $\mathcal{C}''$. \par Let $B$ be the set of triples $(n, \ell, m) \in \mathbb{Z}^3$ satisfying $ m,n>0,\left(4\frac{n}{N}m- \ell^2\right)> 0$ and 
\begin{equation}
    \hbox{Min}_{n_1,s\in\mathbb{Z},\, s(s+1)=0 \, {\rm mod}\, n_1} \left[\left({n_1 N\ell-2ns-n\over n}\right)^2 
    +{n_1^2N^2\over n^2}
    \left(4 {n\over N} m -\ell^2\right)\right] \ge 1\, . 
    \label{ediapheq322}
\end{equation}
Then we have the following result.
\begin{prop}
If $(n, \ell, m) \in B$, then the partition function $\widetilde{\Phi}_k^{-1}$ does not have poles in the region bounded by the contours $\mathcal{C}$ and $\mathcal{C}^{\prime\prime}$.    
\end{prop}
\begin{proof}
The analysis is similar except that we deform the imaginary part of $\sigma$
to $i\infty$ keeping the imaginary parts of $\tau$ and $z$ fixed at the
values given in \eqref{eimaginaryparts}. This leads to the analog of
\eqref{eq:mainineqalt}
\begin{equation}
\frac{m_1+j \frac{N\ell}{2 n}}{n_1}\le {Nm\over n}\, ,
\label{eq:mainineqalt22}
\end{equation}
and the analog of \eqref{ediapheq},
\begin{equation}
    \left(j - \ell\, {n_1 N\over n}\right)^2 +{n_1^2 N^2\over n^2}
    \left(4 {n\over N} m -\ell^2\right) \ge 1\, .
    \label{ediapheq22}
\end{equation}
Demanding that this holds for all $j$ and $n_1$, we get
\begin{equation}
    \hbox{Min}_{n_1,s\in\mathbb{Z},\, s(s+1)=0 \, {\rm mod}\, n_1} \left[\left({n_1 N\ell-2ns-n\over n}\right)^2 
    +{n_1^2N^2\over n^2}
    \left(4 {n\over N} m -\ell^2\right)\right] \ge 1\, . 
\end{equation}
This is invariant under the transformation
\begin{equation} \label{esecondellip}
    (n/N,\ell,m)\mapsto (n/N, \ell+2n\lambda_2,
m+Nn\lambda_2^2+N\lambda_2\ell), \qquad \lambda_2\in\mathbb{Z}\, ,
\end{equation}
since the term inside the square bracket remains invariant under this
transformation accompanied by a shift $s\to s+Nn_1\lambda_2$. This proves the desired result.
\end{proof}
\subsection{Poles on the attractor contour}\label{sec:poleattcont}

We shall now show that when there is a pole on the attractor contour,
the residue of $\widetilde\Phi_k^{-1}$ at the pole vanishes and we can
deform the contour through such a pole. For this we first note that the
if the pole \eqref{eq:poln2=0} lies on the attractor contour \eqref{eimaginaryparts}, then we have,
\begin{equation} \label{epoleattr}
2m n_1 - 2 {n\over N} m_1 - j \ell = 0\, .
\end{equation}
We now use \eqref{eq:edecay} and \eqref{eq:edecayints} to conclude that the pole
at \eqref{eq:poln2=0} corresponds to the wall of marginal stability 
associated with the
decay:
\begin{equation} \label{edecay}
    (Q,P) \to (\alpha Q+\beta P, \gamma Q+\delta P)
    + ( (1-\alpha) Q-\beta P, -\gamma Q+(1-\delta)P)\, ,
\end{equation}
where
\begin{equation}
    \gamma=n_1, \qquad \beta=m_1, \qquad \delta={1-j\over 2}, \qquad 
    \alpha={j+1\over 2}\, .
\end{equation}
On the other hand on a marginal stability wall corresponding to the decay
$(Q,P)\to (Q_1,P_1)+(Q_2,P_2)$ the jump across the wall is proportional to
\begin{equation}
(Q_1\cdot P_2-Q_2\cdot P_1)\, ,
\end{equation}
which for \eqref{edecay} evaluates to
\begin{equation}
    \beta P^2 -\gamma Q^2 +(\alpha-\delta) Q\cdot P
    = -2 mn_1+ 2m_1{n\over N} + j\ell = 0\, ,
\end{equation}
where in the last step we used \eqref{epoleattr}. This shows that when
a pole of $\widetilde\Phi_k^{-1}$ lies on the attractor contour, the
jump across the corresponding wall of marginal stability vanishes, and hence
the residue of $\widetilde\Phi_k^{-1}$ at the pole also vanishes. This proves the
desired result.

As a corollary of this, we can also prove the following useful result.
When \eqref{ediapheq3} is satisfied and we can deform the $\tau$ contour
to $i\infty$ keeping the $\sigma$ and $z$ contours fixed, there remains the
possibility that there will be poles on the contour in the $(\sigma,z)$ plane.
To explore this possibility, we note that 
to satisfy
\eqref{eimaginarypart} for $\tau=i\infty$ we have to set $m_1=0$. 
\eqref{econtraintmn} now sets $j=1$ 
so that the poles are at:
\begin{equation} \label{epoless}
    n_1 \operatorname{Im}(\sigma)+ \operatorname{Im}(z)=0\, .
\end{equation}
On the other hand the integration contour for $z$ and $\sigma$ are at
\begin{equation}\label{esigmaz}
    \operatorname{Im}(\sigma)={2 m\over \varepsilon}, \qquad
    \operatorname{Im}(z) = -{\ell\over \varepsilon}=-{\ell\over 2m} \operatorname{Im}(\sigma) \\, .
\end{equation}
Now since we have set $m_1=0$, it follows 
that if any of the poles \eqref{epoless} lie on
\eqref{esigmaz}, they will also lie on the attractor contour \eqref{eimaginarypart}.
Therefore by our earlier result the residue at the pole vanishes and they
are harmless. We also see from \eqref{epoless}, \eqref{esigmaz} 
that as long as we keep the ratio $\operatorname{Im}(z)/\operatorname{Im}(\sigma)$
fixed at $-\ell/(2m)$, we can deform the imaginary
part of $\sigma$ without crossing any pole.

A similar argument can be used to show that when \eqref{ediapheq322} is satisfied and we
deform the $\sigma$ integration contour to $i\infty$, there is no pole
with non-vanishing residue in the $(\tau,z)$ plane as long as the
ratio $\operatorname{Im}(z)/\operatorname{Im}(\tau)$
fixed at $-N\ell/(2n)$.

\subsection{Special elements of the set $A$}\label{app:b3}
We shall now construct certain special elements in the set $A$. For this we need to find $(n,\ell,m)$ which satisfy the inequality \eqref{ediapheq} for every allowed value of $j,m_1$. The following set of charges are easily seen to be in the set $A$:
\begin{equation}
\begin{split}
&A_1:=\left\{(n,\ell,m):~ m,n>0,~ 4\frac{n}{N}m-\ell^2\geq \frac{m^2}{N^2}\right\}.
\end{split}  
\end{equation}
Indeed since $m_1\in N\Z$, for charges in $A_1$ the second term of \eqref{ediapheq} is bounded below by 1 and hence satisfies the inequality unless $m_1=0$. However, for $m_1=0$, the first term is bounded from below by 1 since $j$ is odd. Note that $A_1$ is invariant under elliptic transformation \eqref{einvariance}.
\par We now construct a slightly nontrivial subset of $A$. We prove that the set
\begin{equation}
A_2:=\left\{(n,\ell,m):~ m,n>0,~ \frac{n}{N}\geq m,~0\leq|\ell|\leq 2m\right\}    
\end{equation}
is contained in\footnote{The set $A_2$ contains charges with $\ell\equiv 0\bmod 2m$ which corresponds to a pole lying on the attractor contour. This is harmless as discussed in section \ref{sec:poleattcont} above.} $A$. 
Since $\frac{n}{N}\geq m$, we have
\begin{equation}
    \frac{1-j^2}{4 m_1^2}-\frac{j \frac{\ell}{2 m}}{m_1}-{n\over Nm}\leq \frac{1-j^2}{4 m_1^2}-\frac{j \frac{\ell}{2 m}}{m_1}-1.
\end{equation}
Thus to prove \eqref{ediapheq'}, it suffices to show that 
\begin{eqnarray}
\frac{1-j^2}{4 m_1^2}-\frac{j \frac{\ell}{2 m}}{m_1}-1\leq 0.    
\end{eqnarray}
Put $a=\frac{\ell}{2 m}$. Then $-1\leq a\leq 1$ for charges in $A_2$. The inequality can be written as
\begin{equation}
j^2+4 a j m_1+4 m_1^2-1\geq 0 .
\label{eq:mainineq'}
\end{equation}
Since $|a|\leq 1$, we have 
\begin{equation}
    \begin{split}
        j^2+4 a j m_1+4 m_1^2-1&\geq j^2-4|j m_1|+4 m_1^2-1\\&=(|j|-2|m_1|)^2-1\geq 0,
    \end{split}
\end{equation}
since $|j|-2|m_1|$ is an odd integer. This proves the desired result.\par 
It follows from \eqref{einvariance} that any other $(n,\ell,m)$ obtained from
the set $A_2$ by the transformation \eqref{einvariance} will also
belong to the set $A$.

\subsection{Special elements of the set $B$}\label{app:b4}
We shall now construct certain special elements in the set $B$ which consists of integers $(n,\ell,m)$ satisfying the inequality \eqref{ediapheq22} for every allowed value of $j,n_1$. The following set of charges are easily seen to be in the set $B$:
\begin{equation}
\begin{split}
&B_1:=\left\{(n,\ell,m):~ m,n>0,~ 4\frac{n}{N}m-\ell^2\geq \frac{n^2}{N^2}\right\}. 
\end{split}  
\end{equation}
Indeed for charges in $B_1$, the second term of \eqref{ediapheq22} is bounded below by 1 and hence satisfies the inequality unless $n_1=0$. On the other hand, for $n_1=0$, the first term is bounded from below by 1. Note that $B_1$ is invariant under elliptic transformation \eqref{esecondellip}. 
\par We now prove that the set 
\begin{equation}
B_2:=\left\{(n,\ell,m):~ m,n>0,~ m\geq \frac{n}{N},~0\leq|\ell|\leq 2n\right\}   \end{equation}
is contained in $B$. 
Since $m\geq \frac{n}{N}$, we have 
\begin{equation}
\frac{m_1+j \frac{N\ell}{2 n}}{n_1}- {Nm\over n}\leq \frac{m_1+j \frac{N\ell}{2 n}}{n_1N^2}-1. 
\label{eq:mainineq''}
\end{equation} 
Thus to prove \eqref{eq:mainineqalt22}, it suffices to prove that 
\begin{equation}
\frac{m_1+j \frac{N\ell}{2 n}}{n_1N^2}-1\leq 0    
\end{equation}
for all allowed values of
$n_1, m_1, j$ satisfying \eqref{eq:polmnj}. Using \eqref{eq:polmnj} we get
\begin{equation}
\frac{m_1}{n_1}=\frac{1-j^2}{4 n_1^2}
\end{equation}
Substituting this in \eqref{eq:mainineq''}, the inequality becomes
\begin{equation}
\frac{1-j^2}{4 n_1^2N^2}-\frac{j \frac{\ell}{2 n}}{n_1N}-1\leq 0.
\end{equation}
Put $a=\frac{\ell}{2 n}$. Then $-1\leq a\leq 1$ for charges in $B_2$. The inequality can be written as
\begin{equation}
j^2+4 a j Nn_1+4 n_1^2N^2-1\geq 0 .
\label{eq:mainineq'''}
\end{equation}
To prove this, we note that 
\begin{equation}
    \begin{split}
     j^2+4 a j Nn_1+4 n_1^2N^2-1 &\geq  j^2-4 |j Nn_1|+4 n_1^2N^2-1\\&=(|j|-2|n_1N|)^2-1\geq 0.   
    \end{split}
\end{equation}
This proves the desired relation.\par
As before, any other $(n,\ell, m)$ related to charges in the set $B_2$ by the
transformation \eqref{esecondellip} will also be part of set $B$.
\end{appendices}
\\

\noindent\textbf{Data availability:} Data sharing not applicable to this article as no datasets were generated or
analysed during the current study.

\printbibliography

@article{Jatkar:2005bh,
    author = "Jatkar, Dileep P. and Sen, Ashoke",
    title = "{Dyon spectrum in CHL models}",
    eprint = "hep-th/0510147",
    archivePrefix = "arXiv",
    doi = "10.1088/1126-6708/2006/04/018",
    journal = "JHEP",
    volume = "04",
    pages = "018",
    year = "2006"
}

@article{David:2006yn,
    author = "David, Justin R. and Sen, Ashoke",
    title = "{CHL Dyons and Statistical Entropy Function from D1-D5 System}",
    eprint = "hep-th/0605210",
    archivePrefix = "arXiv",
    doi = "10.1088/1126-6708/2006/11/072",
    journal = "JHEP",
    volume = "11",
    pages = "072",
    year = "2006"
}

@article{David:2006ud,
    author = "David, Justin R. and Jatkar, Dileep P. and Sen, Ashoke",
    title = "{Dyon spectrum in generic N=4 supersymmetric Z(N) orbifolds}",
    eprint = "hep-th/0609109",
    archivePrefix = "arXiv",
    doi = "10.1088/1126-6708/2007/01/016",
    journal = "JHEP",
    volume = "01",
    pages = "016",
    year = "2007"
}

@unpublished{Dabholkar:2012nd,
    author = "Dabholkar, Atish and Murthy, Sameer and Zagier, Don",
    title = "{Quantum Black Holes, Wall Crossing, and Mock Modular Forms}",
    eprint = "1208.4074",
    archivePrefix = "arXiv",
    primaryClass = "hep-th",
    month = "8",
    year = "2012"
}

@book{eichler2013theory,
  title={The Theory of Jacobi Forms},
  author={Eichler, M. and Zagier, D.},
  isbn={9781468491623},
  lccn={84028250},
  series={Progress in Mathematics},
  url={https://books.google.co.in/books?id=p_PTBwAAQBAJ},
  year={2013},
  publisher={Birkh{\"a}user Boston}
}

@book{bringmann2017harmonic,
  title={Harmonic Maass Forms and Mock Modular Forms: Theory and Applications},
  author={Bringmann, K. and Folsom, A. and Ono, K. and Rolen, L.},
  isbn={9781470419448},
  lccn={2017026415},
  series={Colloquium Publications},
  url={https://books.google.co.in/books?id=dTdDDwAAQBAJ},
  year={2017},
  publisher={American Mathematical Society}
}

@article{Sen:2007qy,
    author = "Sen, Ashoke",
    title = "{Black Hole Entropy Function, Attractors and Precision Counting of Microstates}",
    eprint = "0708.1270",
    archivePrefix = "arXiv",
    primaryClass = "hep-th",
    doi = "10.1007/s10714-008-0626-4",
    journal = "Gen. Rel. Grav.",
    volume = "40",
    pages = "2249--2431",
    year = "2008"
}

@article{Sen:2011ktd,
    author = "Sen, Ashoke",
    title = "{How Do Black Holes Predict the Sign of the Fourier Coefficients of Siegel Modular Forms?}",
    eprint = "1008.4209",
    archivePrefix = "arXiv",
    primaryClass = "hep-th",
    doi = "10.1007/s10714-011-1175-9",
    journal = "Gen. Rel. Grav.",
    volume = "43",
    pages = "2171--2183",
    year = "2011"
}

@phdthesis{zwegers2008mock,
    author = {Sander Zwegers},
    title = {Mock Theta Functions} ,
    school ={Universiteit Utrecht} ,
    year = {2002}
}

@article{77c4633f-ad35-333d-8e41-59ec3cf33f0b,
 ISSN = {00029327, 10806377},
 URL = {http://www.jstor.org/stable/2371774},
 author = {Carl Ludwig Siegel},
 journal = {American Journal of Mathematics},
 number = {1},
 pages = {1--86},
 publisher = {Johns Hopkins University Press},
 title = {Symplectic Geometry},
 volume = {65},
 year = {1943}
}

@article{Chattopadhyaya_2019,
	doi = {10.1007/jhep05(2019)005},
  
	url = {https://doi.org/10.1007%2Fjhep05%282019%29005},
  
	year = 2019,
	month = {may},
  
	publisher = {Springer Science and Business Media {LLC}
},
  
	volume = {2019},
  
	number = {5},
  
	author = {Aradhita Chattopadhyaya and Justin R. David},
  
	title = {Properties of dyons in $\mathcal{N} = 4$ theories at small charges},
  
	journal = {Journal of High Energy Physics}
}

@unpublished{Zagier1975NombresDC,
  title={Nombres de classes et formes modulaires de poids 3/2.},
  author={Don Zagier},
  year={1975},
  url={https://api.semanticscholar.org/CorpusID:117393824}
}

@misc{bhand2022zagiers,
      title={Zagier's weight $3/2$ mock modular form}, 
      author={Ajit Bhand and Ranveer Kumar Singh},
      year={2022},
      eprint={2012.00539},
      archivePrefix={arXiv},
      primaryClass={math.NT}
}

@book{cohen2017modular,
  title={Modular Forms: A Classical Approach},
  author={Cohen, H. and Str{\"o}mberg, F.},
  isbn={9781470440817},
  series={Graduate studies in mathematics},
  url={https://books.google.co.in/books?id=1MmctQEACAAJ},
  year={2017},
  publisher={American Mathematical Society}
}

@article{FOLSOM2017500,
title = {Perspectives on mock modular forms},
journal = {Journal of Number Theory},
volume = {176},
pages = {500-540},
year = {2017},
issn = {0022-314X},
doi = {https://doi.org/10.1016/j.jnt.2017.02.001},
url = {https://www.sciencedirect.com/science/article/pii/S0022314X17300653},
author = {Amanda Folsom}
}

@article{Chaudhuri:1995fk,
    author = "Chaudhuri, Shyamoli and Hockney, George and Lykken, Joseph D.",
    title = "{Maximally supersymmetric string theories in D \ensuremath{<} 10}",
    eprint = "hep-th/9505054",
    archivePrefix = "arXiv",
    reportNumber = "FERMILAB-PUB-95-099-T, NSF-ITP-94-37",
    doi = "10.1103/PhysRevLett.75.2264",
    journal = "Phys. Rev. Lett.",
    volume = "75",
    pages = "2264--2267",
    year = "1995"
}

@article{David:2006ji,
    author = "David, Justin R. and Jatkar, Dileep P. and Sen, Ashoke",
    title = "{Product representation of Dyon partition function in CHL models}",
    eprint = "hep-th/0602254",
    archivePrefix = "arXiv",
    doi = "10.1088/1126-6708/2006/06/064",
    journal = "JHEP",
    volume = "06",
    pages = "064",
    year = "2006"
}

@article{Sen:2009md,
    author = "Sen, Ashoke",
    title = "{A Twist in the Dyon Partition Function}",
    eprint = "0911.1563",
    archivePrefix = "arXiv",
    primaryClass = "hep-th",
    doi = "10.1007/JHEP05(2010)028",
    journal = "JHEP",
    volume = "05",
    pages = "028",
    year = "2010"
}

@article{Sen:2010ts,
    author = "Sen, Ashoke",
    title = "{Discrete Information from CHL Black Holes}",
    eprint = "1002.3857",
    archivePrefix = "arXiv",
    primaryClass = "hep-th",
    doi = "10.1007/JHEP11(2010)138",
    journal = "JHEP",
    volume = "11",
    pages = "138",
    year = "2010"
}

@article{Mandal:2010cj,
    author = "Mandal, Ipsita and Sen, Ashoke",
    editor = "Baulieu, Laurent and de Boer, Jan and Douglas, Michael and Rabinovici, Eliezer and Vanhove, Pierre and Windey, Paul",
    title = "{Black Hole Microstate Counting and its Macroscopic Counterpart}",
    eprint = "1008.3801",
    archivePrefix = "arXiv",
    primaryClass = "hep-th",
    doi = "10.1088/0264-9381/27/21/214003",
    journal = "Class. Quant. Grav.",
    volume = "27",
    pages = "214003",
    year = "2010"
}

@book{murty2016problems,
  title={Problems in the Theory of Modular Forms},
  author={Murty, M.R. and Dewar, M. and Graves, H.},
  isbn={9789811026515},
  series={HBA Lecture Notes in Mathematics},
  url={https://books.google.com/books?id=v8eWDQAAQBAJ},
  year={2016},
  publisher={Springer Nature Singapore}
}

@article{DAS2015351,
title = {Jacobi forms and differential operators},
journal = {Journal of Number Theory},
volume = {149},
pages = {351-367},
year = {2015},
issn = {0022-314X},
doi = {https://doi.org/10.1016/j.jnt.2014.10.006},
url = {https://www.sciencedirect.com/science/article/pii/S0022314X14003539},
author = {Soumya Das and B. Ramakrishnan}
}

@book{koblitz2012introduction,
  title={Introduction to Elliptic Curves and Modular Forms},
  author={Koblitz, N.I.},
  isbn={9781461209096},
  series={Graduate Texts in Mathematics},
  url={https://books.google.com/books?id=0hTSBwAAQBAJ},
  year={2012},
  publisher={Springer New York}
}

@article{ojm/1502092828,
author = {Hiroki Aoki},
title = {{On Jacobi forms of real weights and indices}},
volume = {54},
journal = {Osaka Journal of Mathematics},
number = {3},
publisher = {Osaka University and Osaka Metropolitan University, Departments of Mathematics},
pages = {569 -- 585},
year = {2017},
}

@article{Gritsenko:1999fk,
    author = "Gritsenko, V.",
    title = "{Elliptic genus of Calabi-Yau manifolds and Jacobi and Siegel modular forms}",
    eprint = "math/9906190",
    archivePrefix = "arXiv",
    month = "6",
    year = "1999"
}

@article{Cheng:2007ch,
    author = "Cheng, Miranda C. N. and Verlinde, Erik",
    title = "{Dying Dyons Don't Count}",
    eprint = "0706.2363",
    archivePrefix = "arXiv",
    primaryClass = "hep-th",
    reportNumber = "ITFA-07-22",
    doi = "10.1088/1126-6708/2007/09/070",
    journal = "JHEP",
    volume = "09",
    pages = "070",
    year = "2007"
}

@article{Dabholkar:2006bj,
    author = "Dabholkar, Atish and Gaiotto, Davide",
    title = "{Spectrum of CHL dyons from genus-two partition function}",
    eprint = "hep-th/0612011",
    archivePrefix = "arXiv",
    reportNumber = "TIFR-TH-06-33, HUTP-06-A0043",
    doi = "10.1088/1126-6708/2007/12/087",
    journal = "JHEP",
    volume = "12",
    pages = "087",
    year = "2007"
}

@article{Sen:2007vb,
    author = "Sen, Ashoke",
    title = "{Walls of Marginal Stability and Dyon Spectrum in N=4 Supersymmetric String Theories}",
    eprint = "hep-th/0702141",
    archivePrefix = "arXiv",
    doi = "10.1088/1126-6708/2007/05/039",
    journal = "JHEP",
    volume = "05",
    pages = "039",
    year = "2007"
}

@incollection {Ono2009,
    AUTHOR = {Ono, Ken},
     TITLE = {Unearthing the visions of a master: harmonic {M}aass forms and
              number theory},
 BOOKTITLE = {Current developments in mathematics, 2008},
     PAGES = {347--454},
 PUBLISHER = {Int. Press, Somerville, MA},
      YEAR = {2009},
      ISBN = {978-1-57146-139-1},
   MRCLASS = {11F37 (11-02 11F11 11P82 11P84)},
  MRNUMBER = {2555930},
MRREVIEWER = {Kathrin\ Bringmann},
}
\end{document}